\documentclass[a4paper,reqno,11pt]{amsart} %Loads amsthm, amsmath, amsfonts + a4paper and right-numbered equations

\usepackage[T1]{fontenc} 
\usepackage[utf8]{inputenc}
\usepackage[english]{babel} 
\usepackage{newpxtext} %Palatino

\usepackage{scalerel}

\usepackage{amssymb,textcomp,mathrsfs,stmaryrd,braket,commath,relsize}
\usepackage{bm} %To use many alphabets
	\SetSymbolFont{stmry}{bold}{U}{stmry}{m}{n} %Avoids meaningless warning for bold math symbols
\usepackage[new]{old-arrows} %For long hook-arrows
    
\usepackage{mathtools}
\usepackage{thm-restate}

\theoremstyle{plain} %Math enviroments I
    \newtheorem{theorem}{Theorem}[section]
    \newtheorem*{theorem*}{Theorem}
	
	\newtheorem{citedtheorem}{Theorem}%[section]

    \newtheorem{proposition}[theorem]{Proposition}
    \newtheorem*{proposition*}{Proposition}
    
	\newtheorem{corollary}[theorem]{Corollary}
    \newtheorem*{corollary*}{Corollary}
    
	\newtheorem{lemma}[theorem]{Lemma}
    \newtheorem*{lemma*}{Lemma}
    
	\newtheorem{conjecture}[theorem]{Conjecture}
    \newtheorem*{conjecture*}{Conjecture}
    
\theoremstyle{definition} %Math enviroments II
    \newtheorem{definition}[theorem]{Definition}
    \newtheorem*{definition*}{Definition}

    \newtheorem*{notation*}{Notation}

\theoremstyle{remark} %Math enviroments III
    \newtheorem{remark}[theorem]{Remark}
    \newtheorem*{remark*}{Remark}

	\newtheorem*{example*}{Example}

\numberwithin{equation}{section}
\numberwithin{figure}{section}

\usepackage{etoolbox} 
    \newcommand{\addQEDstyle}[2]{\AtBeginEnvironment{#1}{\pushQED{\qed}\renewcommand{\qedsymbol}{#2}}
    \AtEndEnvironment{#1}{\popQED}} %Symbol at the end of environment: call as \addQEDstyle{environmentname}{symbolname}
    \addQEDstyle{remark}{$\circ$}
    \addQEDstyle{remark*}{$\circ$} %Triangles at the end of remarks
	\addQEDstyle{example}{$\circ$}
    \addQEDstyle{example*}{$\circ$} %Triangles at the end of examples
	\addQEDstyle{definition}{$\circ$}
    \addQEDstyle{definition*}{$\circ$} %Triangles at the end of definitions

\usepackage{tikz}
	\usetikzlibrary{cd} %Commutative diagrams
	
\usepackage{subcaption}

\AtBeginDocument{\def\MR#1{}} %No MR number in bibliography
\apptocmd{\sloppy}{\hbadness 10000\relax}{}{} %For badboxes in .bbl (uses etoolbox)

\makeatletter %Print 2020AMSsubjclass (instead of 2010)
	\@namedef{subjclassname@2020}{
		\textup{2020} Mathematics Subject Classification}
\makeatother

\usepackage{ytableau}
\usepackage[percent]{overpic}
\usepackage{enumitem}
\usepackage{rotating}

\usepackage{multirow}
\usepackage{tabularx,stackengine,collcell}
\let\endminwd\relax
\newcolumntype{L}[1]{>{\collectcell\xminwd l{#1}}l<{\endminwd\endcollectcell}}
\newcolumntype{C}[1]{>{\collectcell\xminwd c{#1}}c<{\endminwd\endcollectcell}}
\newcolumntype{R}[1]{>{\collectcell\xminwd r{#1}}r<{\endminwd\endcollectcell}}
\def\minwd#1#2#3\endminwd{\stackengine{0pt}{#3}{\rule{#2}{0pt}}{O}{#1}{F}{F}{L}}
\newcommand\xminwd[1]{\minwd#1}
\allowdisplaybreaks

\usepackage[colorlinks,psdextra,pdfencoding=auto]{hyperref}
\hypersetup{colorlinks,linkcolor=red,filecolor=green,urlcolor=blue,citecolor=blue} 

\begin{document}

%%%%%%%%%%%%%%%%%%%%%%%%%%%%%%%%%%%%%%%%%%%%%%%%%%%%%%%%%%%%%%%%%%%%%%%%%%%%%%%%%%%%%%%%%%%%%%%%

\global\long\def\Mob{\varphi}
\global\long\def\domain{D}
%\global\long\def\linkpatt{\alpha}%{\omega}

\global\long\def\summ{\sigma}
\global\long\def\mult{\mathrm{m}}

\global\long\def\link#1#2{\{#1,#2\}}

\newcommand{\maxu}{\Lambda}

\global\long\def\bR{\mathbb{R}}
\global\long\def\bRpos{\mathbb{R}_{> 0}}
\global\long\def\bRnn{\mathbb{R}_{\geq 0}}
\global\long\def\bZ{\mathbb{Z}}
\global\long\def\bN{\mathbb{N}}
\global\long\def\bZpos{\mathbb{Z}_{>0}}
\global\long\def\bZneg{\mathbb{Z}_{<0}}
\global\long\def\bZnn{\mathbb{Z}_{\geq 0}}
\global\long\def\bQ{\mathbb{Q}}
\global\long\def\bC{\mathbb{C}}
\global\long\def\bH{\mathbb{H}}
\global\long\def\bD{\mathbb{D}}
\global\long\def\bA{\mathbb{A}}

 \global\long\def\sF{\mathcal{F}}
 \global\long\def\sZ{\mathcal{Z}}
 \global\long\def\sD{\mathcal{D}}
 \global\long\def\sG{\mathcal{G}}
 \global\long\def\sC{\mathcal{C}}
 \global\long\def\sL{\mathcal{L}}
 \global\long\def\sA{\mathcal{A}}
 \global\long\def\sE{\mathcal{E}}
 \global\long\def\sR{\mathcal{R}}
 \global\long\def\sS{\mathcal{S}}
 \global\long\def\sP{\mathcal{P}}
 \global\long\def\sM{\mathcal{M}}
 \global\long\def\sK{\mathcal{K}}
 \global\long\def\sV{\mathcal{V}}
 \global\long\def\sP{\mathcal{P}}
 
\global\long\def\ii{\mathfrak{i}}

\newcommand{\Sym}{\mathfrak{S}}
\newcommand{\LP}{\mathsf{LP}}
\newcommand{\GLP}{\mathsf{MLP}}
\newcommand{\PLP}{\mathsf{PLP}}
\newcommand{\TL}{\mathsf{TL}}
\newcommand{\vTL}{\mathsf{TL}}
\newcommand{\Uqsltwo}{{\mathsf{U}_q}}

\newcommand{\dir}{\mathrm{dir}}

\newcommand{\multii}{\varsigma}

\newcommand{\NB}{\textnormal{NB}^\lambda}
\newcommand{\SYT}{\textnormal{SYT}^\lambda}
\newcommand{\RSYT}{\textnormal{RSYT}^\lambda_\multii}
\newcommand{\CSYT}{\textnormal{CSYT}^\lambda_\multii}
\newcommand{\CSYTBar}{\textnormal{CSYT}^{\Bar{\lambda}}_\multii}
\newcommand{\Fill}{\textnormal{Fill}^\lambda_\multii}
\newcommand{\Fillof}[2]{{\textnormal{Fill}_{#1}^{#2}}}
\newcommand{\NBof}[1]{{\textnormal{NB}^{#1}}}
\newcommand{\RSYTof}[1]{{\textnormal{RSYT}^{#1}_\multii}}

\newcommand{\ord}{\textnormal{ord}}
\newcommand{\idpt}{\mathrm{p}_\multii}%{p_\multii}
\newcommand{\idptof}[1]{{\mathrm{p}_{#1}}}%{{p_{#1}}}
\newcommand{\sym}{\mathrm{s}_\multii}%{s_\multii}

\global\long\def\PartF{\mathcal{Z}}
\global\long\def\CobloF{\mathcal{U}}
\global\long\def\SpechtP{\mathcal{P}}
\global\long\def\chamber{\mathfrak{X}}

\newcommand{\re}{\text{\upshape Re}}
\newcommand{\im}{\text{\upshape Im}}

\global\long\def\SLE{\mathrm{SLE}}
\global\long\def\GFF{\mathrm{GFF}}
\global\long\def\aSLE{\alpha\textnormal{-}\mathrm{SLE}}
\global\long\def\iaSLE{\imath(\alpha)\textnormal{-}\mathrm{SLE}}

\global\long\def\PR{\mathbb{P}}
\global\long\def\EX{\mathbb{E}}

\newcommand{\giv}{\,|\,}

%Derivatives notation
\global\long\def\ud{\mathrm{d}}
\global\long\def\der#1{\frac{\ud}{\ud#1}}
\global\long\def\pder#1{\frac{\partial}{\partial#1}}
\global\long\def\pdder#1{\frac{\partial^{2}}{\partial#1^{2}}}
\global\long\def\pddder#1{\frac{\partial^{3}}{\partial#1^{3}}}

\newcommand{\GreenK}{\mathsf{G}}
\newcommand{\ExcK}{\mathsf{K}}
\newcommand{\Walks}{\mathscr{W}}
\newcommand{\walk}{\omega}

\newcommand{\GreenKH}{\mathscr{G}}
\newcommand{\PoissonKH}{\mathscr{P}}
\newcommand{\ExcKH}{\mathscr{K}}

\newcommand{\KWleq}{\stackrel{\scriptsize{()}}{\leftarrow}}
\newcommand{\DPleq}{\preceq} % partial order of DPs
\newcommand{\DPgeq}{\succeq} % partial order of DPs
\newcommand{\CItilingsof}{\mathcal{C}}
\newcommand{\diff}{\mathrm{d}} % stochastic differential

\global\long\def\summ{q}%{p}%{\sigma}
\global\long\def\mult{\mathrm{m}}
\newcommand{\conn}{\vartheta_{\mathrm{GFF}}}

\newcommand{\event}{\mathrm{Conn}}

\global\long\def\quote#1{$``${{#1}}$"$}

\global\long\def\SolSp{\mathcal{S}}%{\mathcal{C}}

\global\long\def\fugacity{\nu}
\newcommand{\Vir}{\mathrm{Vir}}
\newcommand{\SymGrp}{\mathfrak{S}}
\newcommand{\Rows}{\mathfrak{R}^\lambda}%{\textnormal{R}^\lambda}
\newcommand{\Columns}{\mathfrak{C}^\lambda}%{\textnormal{C}^\lambda}

\newcommand{\super}[1]{^{\scaleobj{0.85}{(#1)}}}
\newcommand{\sub}[1]{_{\scaleobj{0.85}{(#1)}}}
\newcommand{\superscr}[1]{^{\scaleobj{0.85}{#1}}}
\newcommand{\subscr}[1]{_{\scaleobj{0.85}{#1}}}

\newcommand{\bs}[1]{{\boldsymbol{#1}}}

\global\long\def\nested{\boldsymbol{\underline{\Cap}}}
\global\long\def\unnested{\boldsymbol{\underline{\cap\cap}}}
\newcommand{\rainbow}[1]{{\nested_{#1}}} 

\newcommand{\BSAOP}{\mathcal{P}}%{\Delta}

\global\long\def\one{\scalebox{0.9}{\textnormal{1}} \hspace*{-.75mm} \scalebox{0.6}{\raisebox{.3em}{\bf |}} }
\global\long\def\onesmall{\scalebox{0.7}{\textnormal{1}} \hspace*{-.75mm} \scalebox{0.5}{\raisebox{.3em}{\bf |}} }

\newcommand{\wt}{\widetilde}
\newcommand{\wh}{\widehat}
\newcommand{\wc}{\widecheck}
\newcommand{\ol}{\overline}
\newcommand{\ul}{\underline}

\newcommand{\Up}{\mathrm{U}}
\newcommand{\Down}{\mathrm{D}}
\newcommand{\Left}{\mathrm{L}}
\newcommand{\Right}{\mathrm{R}}
\newcommand{\Inner}{\mathrm{I}}

\newcommand{\DirReg}[1]{{\bs{\colon} \!\!\! \int_{\bH} |\nabla #1|^2 \ud z \bs{\colon}}}

\newcommand{\dist}{\mathrm{dist}}
\newcommand{\disthaus}{\dist_{\mathrm{H}}}

\newcommand{\OML}{\mathrm{LL}^{\!\mathrm{odd}}}
\newcommand{\config}{\mathrm{Config}}

\newcommand{\sgn}{\mathrm{sgn}}
\newcommand{\loc}{\mathrm{loc}}
\newcommand{\rhoSum}[1]{\varrho_{\Sigma{#1}}}%{{}^{#1}\!\varrho}%{\Sigma\varrho}%{\ol{\rho}}

\newcommand{\Index}{i}
\newcommand{\curve}{\gamma}%{\Gamma}%{\eta}

\newcommand{\DC}{\mathcal{DC}}
\newcommand{\PJ}{\mathcal{X}}
\newcommand{\NJ}{\mathcal{Y}}
\newcommand{\FN}{\mathfrak{N}}

\newcommand{\pp}{\tilde{x}}
\newcommand{\np}{\tilde{y}}

\newcommand{\ownvec}{\mathrm{n}}

%%%%%%%%%%%%%%%%%%%%%%%%%%%

\title{Level lines of the Gaussian free field and $c=1$ degenerate conformal blocks}
%\bigskip{}
%\author{Alex Karrila, Eveliina Peltola, and Lukas Schoug}
%\date{\today}

\vspace{2.5cm}

\begin{center}
{%\LARGE 
\huge
\bf \scshape{
Level lines of the Gaussian free field \\[.5em] and $c=1$ degenerate conformal blocks
}}
\end{center}

\vspace{0.75cm}

\begin{center}
{\Large \scshape Alex Karrila}{\footnotesize\footnotemark[1]} \\
{\footnotesize{\protect\url{alex.karrila@abo.fi}}}\\
\bigskip{}
{\Large \scshape Eveliina Peltola}{\footnotesize\footnotemark[2]\textsuperscript{,}\,\footnotemark[3]} \\
{\footnotesize{\protect\url{eveliina.peltola@aalto.fi}}}\\
\bigskip{}
{\Large \scshape Lukas Schoug}{\footnotesize\footnotemark[4]} \\
{\footnotesize{\protect\url{lschoug@kth.se}}} 
\end{center}

\footnotetext[1]{\r{A}bo Akademi Matematik, Henriksgatan 2, 20500, \r{A}bo, Finland.}
\footnotetext[2]{Department of Mathematics and Systems Analysis, 
P.O. Box 11100, 00076, Aalto University, Finland.}
\footnotetext[3]{Division of Mathematics, University of Cologne, Weyertal 86-90, 50931 Cologne, Germany.}
\footnotetext[4]{KTH Royal Institute of Technology, Lindstedtsv\"agen 25, 100 44, Stockholm, Sweden.}

\setcounter{footnote}{0}

\vspace{0.75cm}

\begin{center}
\begin{minipage}{0.85\textwidth} %\footnotesize
{\scshape Abstract.}
We consider Gaussian free field (GFF) on simply connected domains with piecewise constant Dirichlet boundary data. We show that the crossing probabilities for its level lines are determined by conformal blocks of primary fields in a conformal field theory (CFT) with central charge $c = 1$ which are degenerate at each insertion. Alternatively, the crossing probabilities are ratios of explicit partition functions of fused multiple $\SLE_4$ curves, which can be written in terms of fused Specht polynomials introduced recently by Lafay, Peltola~\&~Roussillon in a representation-theoretic context.

\smallskip

We also prove that for the metric graph GFF introduced by Lupu, with appropriate boundary conditions, the crossing probabilities for its level sets converge in the scaling limit to our formulas. In particular, the geometry of the level-set percolation for both the continuum GFF and the metric graph GFF has a CFT description in terms of the aforementioned $c=1$ conformal blocks, which are linearly independent and solve the Belavin-Polyakov-Zamolodchikov (BPZ) PDEs of arbitrary orders.

\smallskip

Interestingly, not all combinatorial boundary conditions are amenable for the GFF models --- while the CFT contains conformal blocks of primary fields labeled by generalized Dyck paths (i.e., semi-standard Young tableaux), the ones appearing in the above models satisfy specific monotonicity constraints.
\end{minipage}
\end{center}

\vspace*{2mm}

\begin{center}
\includegraphics[width=.6\textwidth]{exploration_step_4_filled}
\end{center}

%\vspace{0.75cm}
\newpage

{\hypersetup{linkcolor=black}
\setcounter{tocdepth}{2}
\tableofcontents}

\newpage

\section{Introduction}
The Gaussian free field (GFF), denoted $\Psi$ (\quote{massless free boson}) 
is a Gaussian random distribution whose covariance kernel is given by the Green's function of the (positive) Laplacian operator. It is ubiquitous in random geometry and mathematical physics. 
It can be defined in terms of a probability measure in a suitable Sobolev space. 
Heuristically speaking, the GFF is a model for a random harmonic function --- it however fails to be an honest function\footnote{Our exposition here is heuristic; the exactly analogous formulas however work for the discretized GFF, see Section~\ref{subsec:DGFF}. The precise definition of the (continuum) GFF is recalled in Section~\ref{subsec:GFF}.}
and has to be treated as a distribution instead.
Roughly speaking, the GFF (in two-dimensional space) has a formal (Euclidean) action functional
\begin{align*}
S(\Psi) =  \int_\bC |\nabla \Psi(z)|^2 \, \ud z ,
\end{align*}
where \quote{$\nabla \Psi$} is the formal gradient of $\Psi$ and \quote{$\ud z$} the two-dimensional Lebesgue measure. 
The expected value with respect to the probability measure of $\Psi$ is formally expressed as
\begin{align*}
\EX[ f(\Psi) ] = \frac{\int f(\Psi) e^{-S(\Psi)} D \Psi}{\int e^{-S(\Psi)} D \Psi} ,
\end{align*}
where the integral is thought of as taken over the space of distributions (the probability space for $\Psi$) and the denominator 
can be formally expressed in terms of the determinant of the Laplacian (which could be zeta-regularized to yield a finite quantity):
\begin{align*}
\int e^{-S(\Psi)} D \Psi = \int \exp \Big( \int_\bC \Psi(z) \Delta \Psi(z) \, \ud z \Big) D \Psi 
= (\det (-\Delta))^{-1/2} .
\end{align*}
On domains $\domain \subsetneq \bC$ with boundary, a boundary condition, or boundary data, for the GFF is defined by adding to it a bounded harmonic function $u \colon \domain \to \bR$ with given boundary values. 
The \emph{partition function} of the GFF \emph{with boundary data} $u$ is defined as~\cite{Dubedat:SLE_and_free_field} 
\begin{align}\label{eq:partition_function_dubedat}
\CobloF_u^{\GFF} \coloneqq \frac{ e^{- \frac{1}{2} ( u,u )_{\nabla}} }{(\det_\zeta (-\Delta))^{1/2}} ,
\end{align}
where $\det_\zeta (-\Delta)$ is the $\zeta$-regularized determinant of the positive Dirichlet Laplacian~\cite{OPS:Extremals_of_determinants_of_Laplacians} (which we will not need in the present work) 
and the regularized Dirichlet norm\footnote{More generally, one would integrate the volume form of the underlying surface, and the regularization procedure may depend on the metric.} 
\begin{align} \label{eq: regularized Dirichlet norm}
( u,u )_{\nabla} \coloneqq \; \DirReg{u} ,
\end{align}
is the contribution of the boundary data (see~\cite{Dubedat:SLE_and_free_field} and Section~\ref{subsec:partition_functions_conformal_blocks}).

To understand the underlying geometry, it is quite natural to consider level sets of the GFF, even though it is not an honest function.
Indeed, despite the fact that the GFF is not pointwise defined, one can make sense its \quote{zero-height} level lines
as Schramm's celebrated SLE (Schramm-Loewner evolution)
curves~\cite{Schramm:Scaling_limits_of_LERW_and_UST, LSW:Conformal_invariance_of_planar_LERW_and_UST, Schramm:ICM}.
This yields a natural coupling between GFF and the $\SLE_4$~\cite{Dubedat:SLE_and_free_field, Schramm-Sheffield:Contour_lines_of_2D_discrete_GFF, Schramm-Sheffield:A_contour_line_of_the_continuum_GFF, Miller-Sheffield:Imaginary_geometry1, Wang-Wu:Level_lines_of_Gaussian_free_field_I}.
The \quote{zero-height} level line can be thought of as a contour separating two regions where the boundary data for the GFF has a jump, see Equation~\eqref{eq:harmonic_function}. 
More generally, with boundary data involving several jumps of varying heights gives rise to a model for multiple level lines,
which are multiple $\SLE_4$ curves~\cite{Wang-Wu:Level_lines_of_Gaussian_free_field_I,Peltola-Wu:Global_and_local_multiple_SLEs_and_connection_probabilities_for_level_lines_of_GFF, Karrila:Computation_of_pairing_probabilities_in_multiple-curve_models}. 
More general multiple $\SLE$ curves have been investigated over a couple of decades in,
e.g.,~\cite{Cardy:SLE_and_Dyson_circular_ensembles, BBK:Multiple_SLEs_and_statistical_mechanics_martingales, Dubedat:Commutation_relations_for_SLE,Kozdron-Lawler:Configurational_measure_on_mutually_avoiding_SLEs, Lawler:Partition_functions_loop_measure_and_versions_of_SLE, 
Karrila:Multiple_SLE_local_to_global,
Peltola-Wu:Global_and_local_multiple_SLEs_and_connection_probabilities_for_level_lines_of_GFF, 
Karrila:UST_branches_martingales_and_multiple_SLE2, 
BPW:On_the_uniqueness_of_global_multiple_SLEs, 
Healey-Lawler:N_sided_radial_SLE,
Sun-Yu:SLE_partition_functions_via_conformal_welding_of_random_surfaces,
AMY:Multiple_SLE_from_CLE,HPW:Multiradial_SLE_with_spiral, AHSY:Conformal_welding_of_quantum_disks_and_multiple_SLE_the_non-simple_case,
Karrila:Computation_of_pairing_probabilities_in_multiple-curve_models}.

In the present work, we are interested in the GFF with specific Dirichlet boundary data (as in Equation~\eqref{eq:harmonic_function} given below,
and exemplified in Figure~\ref{fig:GFF_bdata_conn}) and its connections to conformal field theory (CFT). 
The point of junction of our first results are certain conformal block functions of a CFT with central charge $c = 1$
(which are degenerate at each insertion, and coincide with those introduced in~\cite{LPR:Fused_Specht_polynomials_and_c_equals_1_degenerate_conformal_blocks},
satisfying the Belavin-Polyakov-Zamolodchikov (BPZ) PDEs) ---  see Theorem~\ref{thm:CFT properties}.
Specifically, our results express the following observables (for suitably chosen GFF boundary data) in terms of these conformal blocks:
\begin{itemize}[leftmargin=*]
\item the probabilities of the different connectivities of the GFF level lines (see Theorem~\ref{thm:corss_proba_H});
\item the partition functions of the Dirichlet GFF (see Equation~\eqref{eq:GFF part fcn in terms of conf block fcn}); and
\item the associated $\SLE_4$ partition functions of the GFF level lines 
(see Sections~\ref{subsec:local_description_unconditional_case}--\ref{subsec:local_description_conditional_case}). 
\end{itemize}
We will also establish the following results: 
\begin{itemize}[leftmargin=*]
\item the analogous connection probabilities in the metric graph GFF introduced in~\cite{Lupu:From_loop_clusters_and_random_interlacements_to_the_free_field} converge to the above GFF  
connection probabilities (Theorem~\ref{thm:corss_proba_MGFF}); and
\item the GFF level lines give rise to a global multiple $\SLE_4$ model, uniquely characterized by invariance of distribution under $\SLE_4$ resampling (Theorems~\ref{thm:conditional_level_lines_SLE4_intro}~\&~\ref{thm:global_SLE}).
\end{itemize}

In the rest of the introduction, we explain the setup and results in some more detail, but still in part informally (leaving the more precise statements to the bulk of this article).

\subsection{Connection probabilities, partition functions, and conformal blocks}

\begin{figure}
\includegraphics[width=.47\textwidth]{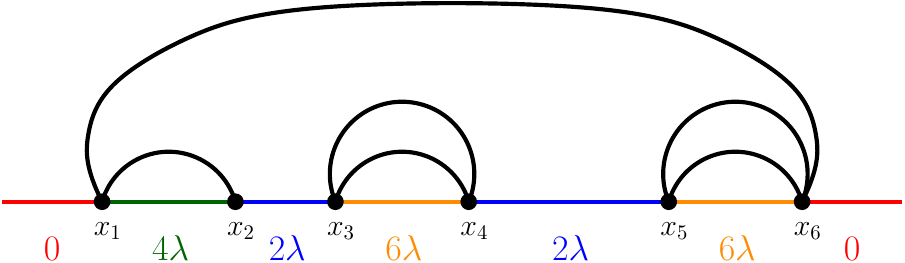}
\quad
\includegraphics[width=.47\textwidth]{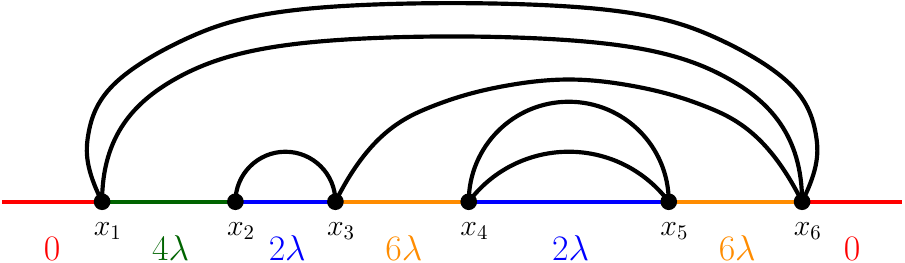}
\caption{
\label{fig:GFF_bdata_conn}
Two possible connectivities $\conn$ formed by the GFF level lines that can appear with the boundary conditions illustrated along the real line.
%A schematic illustration for Theorem~\ref{thm:corss_proba_H}: two possible connectivities $\conn = \{ \{1,2\}, \{1,6\}, \{3,4\},\{3,4\}, \{5,6\} \{5,6\} \}$ (left) and $\conn = \{ \{1,6\}, \{1,6\}, \{2,3\}, \{3,6\}, \{4,5\},\{4,5\}\}$ (right) for the level lines $\OML_\beta$ when the boundary condition is given by $u_\beta$, corresponding to the Dyck path $\beta = (0,2,1,3,1,3,0) \in \GLP_\multii$ of valence $\multii = (2,1,2,2,2,3)$.
}
\end{figure}

We consider the upper half-plane $\bH=\{z\in\bC\;|\;\im(z)>0\}$, and fix boundary points 
\begin{align*}
-\infty = x_0 < x_1 < \cdots < x_p < x_{p+1} = +\infty .
 \end{align*}
Throughout, fix $\lambda \coloneqq \sqrt{\pi/8}$.
Let $\Phi$ be a zero-boundary GFF on $\bH$ and let $u = u_\beta$ be the bounded harmonic function in $\bH$ with piecewise constant boundary values
\begin{align}\label{eq:harmonic_function}
u(x) = u_\beta(x) \coloneqq 
\begin{cases}
0 , & x \in (-\infty,x_1) \cup (x_p,+\infty), \\
2 \beta_j \lambda ,  & x\in (x_j,x_{j+1}), \ j \in \{1,\dots,p-1\},
\end{cases}
\end{align}
where the jumps of the boundary condition are encoded by a generalized Dyck path $\beta$ of $p \geq 2$ steps, i.e., a vector $ \beta = (\beta_0, \beta_1, \ldots, \beta_p) \in \bZnn^{p+1}$ with $\beta_0 = \beta_p =0$;
the step directions 
\begin{align*}
\dir \beta_{j} \coloneqq \beta_{j} - \beta_{j-1} ,
\qquad j \in \{1,\dots,p\} ,
\end{align*}
encode the jumps in the boundary data --- see Figure~\ref{fig:GFF_bdata}. 
We collect the step sizes of $\beta$ into a multiindex 
$\multii = (|\beta_1-\beta_0|,|\beta_2-\beta_1|, \ldots, |\beta_p - \beta_{p-1}|) \in \bZpos^p$. 
The set of all generalized Dyck paths $\beta$ with these step sizes will be denoted\footnote{The notation $\GLP_\multii$ comes from a combinatorial bijection between such boundary conditions with so-called 
\emph{maximally sloped $\multii$-valenced link patterns} (see Section~\ref{subsec:comb_notation} and especially Definition~\ref{def:Dyck_max_sloped} for details).
For instance, the statement of Theorem~\ref{thm:corss_proba_H} hence treats the boundary data $u_\beta$ as a (valenced) link pattern $\beta$.}
 as $\GLP_\multii$. 

We will show in Proposition~\ref{prop:partition_functions_agree} that 
the partition function $\CobloF^{\GFF}_\beta = \CobloF^{\GFF}_{u_\beta}$ of the GFF with boundary data $u_\beta$ 
yields a particular $\SLE_4$ partition function:
\begin{align}
\label{eq:GFF part fcn in terms of conf block fcn}
\CobloF^{\GFF}_\beta (x_1,\ldots,x_p) =  \frac{\CobloF_\beta (x_1,\ldots,x_p)}{(\det_\zeta (-\Delta))^{1/2}} ,
\end{align}
namely a \quote{conformal block function} introduced in~\cite{LPR:Fused_Specht_polynomials_and_c_equals_1_degenerate_conformal_blocks},
\begin{align} \label{eq:CobloF_ms}
\CobloF_\beta (x_1,\ldots,x_p) \coloneqq \prod_{1 \leq i < j \leq p} (x_j - x_i)^{\frac{1}{2} \dir \beta_{j} \,\dir \beta_i} , \qquad \beta \in \GLP_\multii .
\end{align}

\begin{figure}
\includegraphics[width=.75\textwidth]{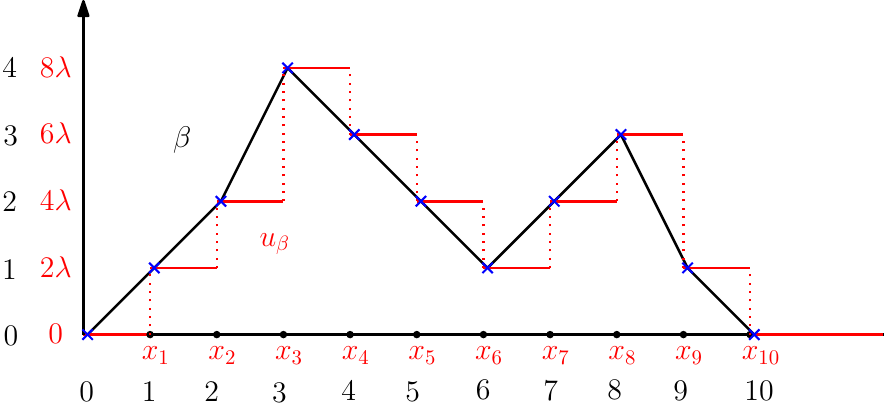}
\caption{
\label{fig:GFF_bdata}
Illustration of the induced boundary data~\eqref{eq:harmonic_function} for the GFF from the generalized Dyck path $\beta = (\beta_0, \beta_1, \ldots, \beta_{10}) = (0,1,2,4,3,2,1,2,3,1,0)$ over the multiindex $\multii = (s_1,s_2, \ldots,s_{10}) = (1,1,2,1,1,1,1,1,2,1)$. 
}
\end{figure}

Throughout, we denote by $\PR_\beta$ the law of a GFF with boundary condition $u_\beta$.
More generally, we will let $\Psi$ be a GFF with some boundary condition, and define 
\begin{align} \label{eq:OML}
\OML = \OML(\Psi) \coloneqq \big\{ \eta \;|\; \eta \textnormal{ is a level line of } \Psi \textnormal{ of height } (2k+1)\lambda \textnormal{ for some } k \in \bZ \big\} ,
\end{align} 
comprising the collections of level lines\footnote{These curves are such that any subset of curves and initial segments is a local set (that is, the law of the GFF outside said collection is that of a zero-boundary GFF plus a harmonic function), and the corresponding boundary condition for said GFF is $2k\lambda$ and $(2k+2)\lambda$, for some $k \in \bZ$, on the two respective sides of a level line of height $(2k+1)\lambda$. Furthermore, there is only one such local set coupling of the GFF with $N$ simple disjoint (except at the boundary) random curves, and in this coupling the curves are determined by the GFF.}  
of $\Psi$ whose heights are odd integer multiples of $\lambda$. 
In particular, under the measure $\PR_\beta$, the fields $\Psi$  and $\Phi + u_\beta$ have the same law, 
and $\OML = \OML_\beta$ has the law of the collection of level lines of $\Phi + u_\beta$ of heights $(2k+1)\lambda$ for $k \in \{0,1,\dots,\Lambda-1\}$, where $\Lambda = \Lambda_\beta \coloneqq \tfrac{1}{2\lambda} \max u_\beta$, and comprises curves $\bs\eta = (\eta_1,\dots,\eta_N)$ with $N$ given by Equation~\eqref{eq: definition of N}. 
These curves form a random topological configuration that can be combinatorially encoded into a random valenced link pattern $\conn$ in $\LP_\multii$ (see Figure~\ref{fig:GFF_bdata_conn} and Definition~\ref{def:link_pattern}).
As our main result, we compute these probabilities explicitly:

\begin{citedtheorem} \label{thm:corss_proba_H}
Fix valences $\multii = (s_1,\ldots,s_p) \in \bZpos^p$ and boundary data $\beta \in \GLP_\multii$. 
We have 
\begin{align} \label{eq:corss_proba_H}
\PR_{\beta} [\conn = \alpha] 
= \; & M_{\beta, \alpha}
\; \frac{\PartF_{\alpha}(x_1, \ldots, x_p)}{\CobloF_{\beta}(x_1, \ldots, x_p)} , \qquad \alpha \in \LP_\multii ,
\end{align} 
where $\PartF_{\alpha}$ are the pure partition functions defined below, and
$M_{\beta, \alpha} = \one\{ \beta \KWleq \alpha \} \in \{0,1\}$ are entries of 
the incidence matrix of a binary relation $\smash{\KWleq}$ obtained from~\cite[Definition~2.2]{KKP:Boundary_correlations_in_planar_LERW_and_UST}, 
for the link patterns obtained from $\beta$ and $\alpha$ via the \quote{unfusing} map $\imath$ defined in Equation~\eqref{eq:iotamap}.
\end{citedtheorem}

We prove Theorem~\ref{thm:corss_proba_H} in Section~\ref{subsec:convergence_connection_probabilities}.
In the present work, regarding the combinatorial matrix $M$, 
we only need to know that the interior link pattern $\alpha$ is possible for the GFF with boundary data $\beta$ if and only if $\smash{\beta \KWleq \alpha}$, that is, if and only if $M_{\beta, \alpha} =1$.
See Lemma~\ref{lem:KWleq importance}.

In Theorem~\ref{thm:corss_proba_MGFF} in Section~\ref{subsec:mGFF_Connection_probabilities}, we will prove that similar probabilities for the metric graph GFF
converge to our formula~\eqref{eq:corss_proba_H}. This is essentially a consequence of the fact that first passage sets 
for the metric graph GFF converge to those of the GFF in the Hausdorff metric~\cite{ALS:First_passage_sets_of_2D_GFF}.
The setup and details for the discrete case are given in Section~\ref{sec:mGFF} (see also \cite{Liu-Wu:Scaling_limits_of_crossing_probabilities_in_metric_graph_GFF}, which motivated us to include the discrete case).

Theorem~\ref{thm:corss_proba_H} could have been equivalently expressed in terms of families of first passage sets of the GFF. 
Intuitively, the \emph{first passage set} $\bA_{a}^u(\Phi) = \bA_{a}(\Phi+u)$ of $\Phi$ is a random fractal set 
comprising points in $\overline{\bH}$ that are reachable from the boundary by following points where $\Phi + u$ takes values at least $a \in \bR$. 
Consider the collection $S$ of frontiers (see Definition~\ref{def:frontier}) of the first passage sets $(\bA_{2k\lambda}(\Psi))_{k \in \bZ}$ of a GFF $\Psi$.  
Under the law $\PR_\beta$, the set $S$ has the law of the collection of frontiers of $\smash{\bA_0^u,\dots,\bA_{2(\maxu-1)\lambda}^u}$
with $\maxu = \maxu_\beta$ and $u = u_\beta$ as above.  
Again, this is a collection of curves $\eta_1,\dots,\eta_N$, with $N$ given by Equation~\eqref{eq: definition of N}. 
and one can show that $S$ and $\OML$ are almost surely the same under $\PR_\beta$, whenever $\beta \in \GLP_\multii$ (Lemma~\ref{lem:frontiers_are_level_lines}).

Theorem~\ref{thm:corss_proba_H} confirms the expectation from~\cite[Remark~3.19]{LPR:Fused_Specht_polynomials_and_c_equals_1_degenerate_conformal_blocks},
where the general conformal block functions were introduced. The detailed combinatorial formulas rely on 
some intuition acquired from an ongoing companion work~\cite{Karrila-Peltola:Boundary_double-dimer_patterns_and_CFT} concerning the double-dimer model. 
Special cases of Theorem~\ref{thm:corss_proba_H} were established in~\cite{Peltola-Wu:Global_and_local_multiple_SLEs_and_connection_probabilities_for_level_lines_of_GFF, Liu-Wu:Scaling_limits_of_crossing_probabilities_in_metric_graph_GFF}, respectively for level lines of the GFF with alternating boundary data ($\multii = (1,1,\ldots,1)$) and for a special first-order fusion case with $\multii = (2,2,\ldots,2)$.

\medskip

Let us also remark that the formula~\eqref{eq:corss_proba_H} in Theorem~\ref{thm:corss_proba_H} actually involves 
all conformal block functions $\{\CobloF_\beta \colon \beta \in \LP_\multii \}$ 
and not only the ones indexed by the so-called \quote{maximally sloped} patterns $\beta \in \GLP_\multii$. 
Interestingly, while the latter naturally index the possible GFF boundary data~\eqref{eq:harmonic_function}, the interior connectivity pattern can be more general.
The probability amplitudes for the interior connectivity patterns $\alpha \in \LP_\multii$ are obtained from the \quote{pure\footnote{The term \quote{pure partition function} is motivated by the key property of these functions given in Lemma~\ref{lem:PartF_fusion} and proven in~\cite{Karrila-Peltola:Boundary_double-dimer_patterns_and_CFT}, and by their role as $\SLE_4$ partition functions 
(Proposition~\ref{prop:alpha-conditional marginal law of one level line}).} partition functions}
$\PartF_\alpha$, which are linear combinations of general conformal blocks.
Their definition involves some combinatorial quantities of independent interest, for which we refer to the literature and to the bulk of this article
(and which are not necessary for understanding the results of the present work).

\begin{definition*}%[{Pure partition functions}] 
We define the \emph{pure partition functions} as the finite linear combinations
\begin{align}\label{eq:PartF}
\PartF_\alpha(x_1,\ldots,x_p) \coloneqq \sum_{\substack{\beta \in \LP_\multii \\ \alpha \DPleq \beta }} (-1)^{|\alpha/\beta|} \# \CItilingsof (\alpha / \beta ) \, (\# \LP_{\geq_\multii \beta}) \, \CobloF_\beta (x_1,\ldots,x_p) , \qquad \alpha \in \LP_\multii ,
\end{align}
where \quote{$\DPleq$} is a natural nesting partial order on $\LP_\multii$ discussed in Section~\ref{subsec:comb_notation}, 
$|\alpha/\beta|$ is the number of atomic square tiles in the skew Young diagram~$\alpha/\beta$ 
(see Figure~\ref{fig:skew_Young_diagram}) naturally extended to $\LP_\multii$ via the \quote{unfusing} map $\imath \colon \LP_\multii \to \LP_N$ defined in Equation~\eqref{eq:iotamap},
$\# \CItilingsof (\alpha / \beta)$ the number of cover-inclusive Dyck tilings 
of~$\alpha/\beta$ according to~\cite[Definitions~2.1~\&~2.8]{KKP:Boundary_correlations_in_planar_LERW_and_UST}
(whose precise definition we will not need in the present work),
the set $\LP_{\geq_\multii \beta}$ is defined in~\eqref{eq: partial order larger set} using a refined partial order $\geq_\multii$ defined in Equation~\eqref{eq: partial order}, 
and $\{\CobloF_\beta \colon \beta \in \LP_\multii \}$ are the general conformal blocks from Definition~\ref{def:CobloF}
(originally introduced in~\cite{LPR:Fused_Specht_polynomials_and_c_equals_1_degenerate_conformal_blocks}). 
\end{definition*}

Because the functions $\{\PartF_\alpha \colon \alpha \in \LP_\multii \}$ 
are linear combinations of the conformal block functions $\{\CobloF_\beta \colon \beta \in \LP_\multii \}$,
they can be interpreted as boundary correlation functions 
\begin{align}\label{eq:fusion_corrf}
`` \langle \phi_{1,s_1+1}(x_1) \cdots \phi_{1,s_p+1}(x_p) \rangle_\alpha " 
\end{align}
in a CFT with central charge $c=1$ of primary fields 
$(\phi_{1,s_1+1},\ldots, \phi_{1,s_p+1})$ with conformal weights 
$h_{1,s_j+1} = s_j^2/4$ for all $j$, 
which in particular belong to the first row of the Kac table of degenerate weights~\cite{Kac:Highest_weight_representations_of_infinite_dimensional_Lie_algebras}.
This is a consequence of the results in~\cite{LPR:Fused_Specht_polynomials_and_c_equals_1_degenerate_conformal_blocks}. 
Moreover, using their probabilistic/combinatorial interpretation established in our main Theorem~\ref{thm:corss_proba_H}, one can show that 
the functions $\{\PartF_\alpha \colon \alpha \in \LP_\multii \}$ are positive and linearly independent (which is far from obvious from their definition) 
--- so they are honest multiple $\SLE_4$ pure partition functions. 
Although the below statement mostly follows from the literature, 
we gather the key properties of the functions $\PartF_\alpha$
in the below Theorem~\ref{thm:CFT properties} and in Section~\ref{subsec:properties_pure_partition_functions}:

\begin{citedtheorem} \label{thm:CFT properties} 
The pure partition functions $\PartF_\alpha$ satisfy the following properties. 
\begin{itemize}
\item[\textnormal{(PDE)}] 
\textnormal{\bf System of $c=1$ BPZ PDEs:} 
With valences $\multii = (s_1,\ldots,s_p) \in \bZpos^p$, we have
\begin{align} \label{eq: BPZ PDE at kappa equals 4} 
\sD_{s_j+1}\super{x_j} 
\; \PartF_\alpha(x_1,\ldots ,x_p) = 0 , \qquad \textnormal{for all } j \in \{ 1,\ldots,p \},
\end{align}
where $\sD_{s_j+1}\super{x_j}$ are the BPZ differential operators
\begin{align} \label{eq: BPZ operator at kappa equals 4} 
\sD_{s_j+1}\super{x_j} 
:= \sum_{k=1}^{s_j+1} \sum_{\substack{m_1,\ldots ,m_k \geq 1 \\ m_1+\cdots +m_k = {s_j+1}}} \frac{ (-1)^{k-s_j-1} (s_j!)^2}{\prod_{l=1}^{k-1} (\sum_{i=1}^l m_i)(\sum_{i=l+1}^k m_i)} \times \sL_{-m_1}\super{x_j} \; \cdots \; \sL_{-m_k}\super{x_j} ,
\end{align} 
and where $\sL_{-m}\super{x_j}$ are the first order differential operators
\begin{align*} 
\sL_{-m}\super{x_j} 
:= -\sum_{\substack{1\leq i \leq p\\i\neq j}} \Big( (x_i-x_j)^{1-m} \pder{x_i} + \frac{(1-m) (s_i + 1)s_i}{2} (x_i-x_j)^{-m}\Big) .
\end{align*}

\medskip

\item[\textnormal{(COV)}] 
\textnormal{\bf Covariance:} 
For all M\"obius maps $\Mob \colon \bH \to \bH$ such that $\Mob(x_1) < \cdots < \Mob(x_p)$, we have
\begin{align}
\label{eq: COV general at kappa equals 2} 
& \PartF_\alpha(x_1, \ldots, x_p) 
= \prod_{j=1}^p |\Mob'(x_j)|^{s_j^2/4} \times \PartF_\alpha(\Mob(x_1), \ldots, \Mob(x_p)) .
\end{align}

\medskip

\item[\textnormal{(POS)}] 
\textnormal{\bf Positivity:} 
For each $\alpha \in \LP_\multii$, we have $\PartF_\alpha(x_1,\ldots ,x_p)>0$ for all $x_1 < \cdots < x_p$.

\medskip

\item[\textnormal{(LIN)}] 
\textnormal{\bf Linear independence:} 
The functions $\{\PartF_\alpha \; | \; \alpha \in \LP_\multii\}$ are linearly independent. 
\end{itemize}
\end{citedtheorem}

\subsection{Multiple chordal Schramm-Loewner evolution, $\SLE_4$}

Level lines of the GFF can be described by variants of Schramm-Loewner evolution curves, $\SLE_\kappa$ with parameter $\kappa=4$~\cite{Schramm-Sheffield:Contour_lines_of_2D_discrete_GFF, Schramm-Sheffield:A_contour_line_of_the_continuum_GFF, Wang-Wu:Level_lines_of_Gaussian_free_field_I} --- 
in particular, marginals of the curves in $\OML$ and hence frontiers of the first passage set are described precisely by such $\SLE_4$ curves. 
Let 
\begin{align*}
(\eta_1,\dots,\eta_N) = \OML(\Phi+u_\beta) = \OML_\beta 
\end{align*}
for a GFF with boundary data $u_\beta$ as in Equation~\eqref{eq:harmonic_function}, where $N$ is given by Equation~\eqref{eq: definition of N}. 
Due to the domain Markov property of the GFF, the curves $\eta_1,\dots,\eta_N$ satisfy the following resampling property 
(which we will verify in Theorem~\ref{thm:resampling_property} in Section~\ref{subsec:resampling_property}):

\begin{citedtheorem}[{Resampling property; see Theorem~\ref{thm:resampling_property}}] 
\label{thm:conditional_level_lines_SLE4_intro}
For each $j \in \{1,\dots,N\}$, let $\domain_j$ be the connected component of $\bH \setminus \bigcup_{k \neq j} \eta_k$ containing $\eta_j$. The conditional law of the curve $\eta_j$, given the other curves $\eta_1,\dots,\eta_{j-1},\eta_{j+1},\dots,\eta_N$, is that of a chordal $\SLE_4$ curve in $\domain_j$.
\end{citedtheorem} 

Moreover, the resampling property uniquely characterizes the law of the curve collection, once their topological configuration has been chosen. 
One can produce curve collections with variant configurations by convex combinations of these \quote{pure} measures. 
Unfortunately, not all topological configurations are possible to achieve by simply setting the boundary data as in~\eqref{eq:harmonic_function} in the GFF model; see the discussion in Section~\ref{subsec:discussion}.
We denote
\begin{align*}
\PLP_\multii \coloneqq \{ \alpha \in \LP_\multii \;|\; \textnormal{there exists } \beta \in \GLP_\multii \textnormal{ such that } M_{\beta, \alpha} = 1 \} ,
\end{align*}
and for each $\beta \in \GLP_\multii$, we denote $\PLP_\multii(\beta) \coloneqq \{ \alpha \in \LP_\multii \;|\; M_{\beta, \alpha} = 1 \} = \{ \alpha \in \LP_\multii \;|\; \smash{\beta \KWleq \alpha} \}$.

\begin{citedtheorem}[{Multiple chordal $\SLE_4$; see Theorem~\ref{thm:existence_fused_multiple_sle}}] \label{thm:global_SLE}
For each $\alpha \in \PLP_\multii$, there exists a unique law $\aSLE_4$ on curves $\eta_1,\dots,\eta_N \subset \overline{\bH}$ that satisfies the resampling property. 
\end{citedtheorem} 

We prove the existence part in Theorem~\ref{thm:global_SLE} in Section~\ref{subsec:resampling_property} and the uniqueness part in Section~\ref{subsec:uniqueness_proof}.
The existence will be established by directly constructing the multiple chordal $\SLE_4$ from GFF level lines.
To prove the uniqueness, we follow a very different strategy than in the known cases 
which use Markov chain arguments~\cite{Miller-Sheffield:Imaginary_geometry2, MSW:Non-simple_SLE_curves_are_not_determined_by_their_range, BPW:On_the_uniqueness_of_global_multiple_SLEs, Zhan:Existence_and_uniqueness_of_nonsimple_multiple_SLE}.
Instead, we show that any $\aSLE_4$ can be coupled as $\OML_\beta$ for any $\beta$ such that $\alpha \in \PLP_\multii(\beta)$ and deduce that the conditional law of these level lines, given the connectivity pattern, does not depend on $\beta$. From this, the uniqueness follows. Complementing the existence and uniqueness, we give an explicit Loewner equation for an $\aSLE_4$ curve with $\alpha \in \PLP_\multii$, as an $\SLE_4$ type curve with the partition function $\PartF_\alpha$ in Section~\ref{subsec:local_description_conditional_case}.

By conformal invariance, these results readily extend to any (nice) topological polygon: 
in general, these multiple chordal $\SLE_4$ curves (as per Definition~\ref{def:fused_multiple_sle})
have probability measures on spaces of (continuous unparametrized) curves indexed by a simply connected domain $\domain$, marked boundary points $x_1,\dots,x_p \in \partial \domain$, and valences $\multii$ that encode how many curves are attached to each boundary point (see Equation~\eqref{eq: definition of N}). 
For notational ease, we only state the general results in the bulk of this article --- see Section~\ref{sec:uniqueness}.

\subsection{Discussion}
\label{subsec:discussion}

Let us note that, in general, $\PLP_\multii \neq \LP_\multii$. 
Indeed, from the GFF perspective, the issue stems from the following. A level line coupled with $\Phi + u_\beta$, where $u_\beta$ is as in~\eqref{eq:harmonic_function}, 
must start at a point $x$ where %\footnote{For a function $f \colon \bR \to \bR$ and $x \in \bR$, we write $f(x^+) = \lim_{y \nearrow x} f(y)$ and $f(x^-) = \lim_{y \searrow x} f(y)$.} 
$u_\beta(x^+) - u_\beta(x^-) > 0$ and end at a point $y$ where $u_\beta(y^+) - u_\beta(y^-) < 0$. 
This puts a natural limitation on which link patterns can appear in $\PLP_\multii$ for some vector $\multii$ of valences, 
as a point of positive jump and a point of negative jump for $u_\beta$ cannot both be linked to the same other point. See for example Figure~\ref{fig:impossible_pattern}, where the left point is a point of positive jump and the right is a point of negative jump, but in which neither a positive nor negative jump for the middle point can give rise to said link pattern. (In fact, $\PLP_{(2,2,2)}$ is empty.) We expect that it is possible to construct multiple $\SLE_4$ with general link patterns, and it would be an interesting extension of our result.

\begin{figure}[h!]
\includegraphics[width=.3\textwidth]{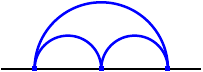}
\caption{
\label{fig:impossible_pattern}
Example of a link pattern which is not in $\PLP_\multii$ for $\multii = (2,2,2)$.
}
\end{figure}

\medskip

\noindent
\textbf{Organization of this article.}
We review preliminaries in Section~\ref{sec:preliminaries}.
Also, in Section~\ref{subsec:Frontiers} we prove that the level lines in $\OML_\beta = \OML(\Phi+u_\beta) $ are the connected components of the frontiers of the first passage sets $\smash{(\bA_{2k\lambda}^{u_\beta})_{k \in \bZ}}$. 
In Section~\ref{sec:GFF_properties}, we prove key results related to the partition functions of the GFF: Theorems~\ref{thm:corss_proba_H}~\&~\ref{thm:CFT properties}, as well as a result 
relating the partition function $\CobloF_{u_\beta}^\GFF$ of \cite{Dubedat:SLE_and_free_field} to the conformal block function $\CobloF_\beta$ (Proposition~\ref{prop:partition_functions_agree}). 
In Section~\ref{sec:uniqueness}, we establish the results related to multiple $\SLE_4$ processes. 
More precisely, in Section~\ref{subsec:resampling_property} we prove Theorem~\ref{thm:conditional_level_lines_SLE4_intro} as well as the existence part of Theorem~\ref{thm:global_SLE}. In Section~\ref{subsec:uniqueness_proof}, 
we show that $\aSLE_4$, with $\alpha \in \PLP_\multii$, can be coupled with a zero-boundary GFF as $\OML_\beta$ for any suitable $\beta$ (Proposition~\ref{prop:coupling_fused_multiple_sle}).
From this, we deduce the uniqueness part of Theorem~\ref{thm:global_SLE}.  
With the global description at hand, we turn to studying local descriptions: 
In Section~\ref{subsec:local_description_unconditional_case}, we describe the marginal law of one curve in $\OML_\beta$ and in Section~\ref{subsec:local_description_conditional_case}, we describe the conditional law of one curve in $\OML_\beta$, given its connectivity $\conn$. 
In the final Section~\ref{sec:mGFF}, using the relationship between level lines and frontiers, and the convergence of the first passage sets of metric graph GFFs to those of the continuum GFF,
we prove Theorem~\ref{thm:corss_proba_MGFF}: the metric graph GFF analogue of Theorem~\ref{thm:corss_proba_H} (a scaling limit result).

\medskip

\noindent
\textbf{Acknowledgments.}

A.K. was supported by the Academy of Finland grant number 339515. A.K. also thanks Stiftelsen för Åbo Akademi and Åbo Akademis Jubileumsfond for travel funding.

This material is part of a project that has received funding from the European Research Council (ERC) under the European Research Council (ERC) under the European Union's Horizon 2020 research and innovation programme (101042460): 
ERC Starting grant \quote{Interplay of structures in conformal and universal random geometry} (ISCoURaGe) 
and from the Academy of Finland grant number 340461 \quote{Conformal invariance in planar random geometry.} 
While carrying out this project, E.P.~has also been supported by 
the Academy of Finland Centre of Excellence Programme grant number 346315 \quote{Finnish centre of excellence in Randomness and STructures (FiRST).}
E.P.~also acknowledges earlier support from the Deutsche Forschungsgemeinschaft (DFG, German Research Foundation) under Germany's Excellence Strategy EXC-2047/1-390685813.  

L.S.~is supported by the Swedish Research Council (VR) grant 2024-05589. L.S.\ also acknowledges the support of Åbo Akademis Jubileumsfond, for funding a visit to Åbo Akademi, where part of this work was carried out.

\bigskip

\section{Preliminaries}
\label{sec:preliminaries}

In this section, we review results that are (almost) known from the literature or, in the case of Section~\ref{subsec:conformal blocks}, from the companion article~\cite{Karrila-Peltola:Boundary_double-dimer_patterns_and_CFT}.
Section~\ref{subsec:comb_notation}  focuses on the combinatorial notions which will be used throughout; 
in Section~\ref{subsec:conformal blocks} we define the conformal block functions (with $c=1$, or $\kappa=4$); 
while Sections~\ref{subsec:SLE} and~\ref{subsec:GFF} respectively summarize the basic notions related to SLE curves and GFF.
The last Section~\ref{subsec:Frontiers} introduces first passage sets (FPS) for the GFF and gathers their key properties.

\subsection{Combinatorial notation}
\label{subsec:comb_notation} 

\subsubsection{Dyck paths and link patterns}

A \emph{(generalized) Dyck path} of $p$ steps is a function 
\begin{align*}
\beta \colon  \{ 0, 1, \ldots, p\} \to \bZnn \qquad \textnormal{with } \beta_0 = \beta_p =0 . 
\end{align*}
We shall consistently visualize Dyck paths by joining the points $(j, \beta_j)$ and $(j+1, \beta_{j+1})$ with a line segment for each $j \in \{0,1,\ldots,p-1\}$. Dyck paths with no constant steps, i.e., with $|\beta_j-\beta_{j-1}| \geq 1$ for all $1 \leq j \leq p$, naturally label GFF boundary data with $p$ jumps, as in Equation~\eqref{eq:harmonic_function}.  
To such a Dyck path is associated a multiindex (valence vector) $\multii = (s_1,\ldots,s_p) \in \bZpos^p$ given by $\multii_j = |\beta_j-\beta_{j-1}|$
which reveals the number of GFF level lines adjacent to the $j$:th boundary jump point, for all $1 \leq j \leq p$. 
The collection of boundary conditions $\beta$ generating given valences $\multii$ is denoted $\GLP_\multii$; see Definition~\ref{def:Dyck_max_sloped} below.

It is natural to index the topological connectivities of GFF level lines by general \emph{$\multii$-valenced link patterns}. Given an integer $p \geq 2$ and 
a multiindex $\multii = (s_1,\ldots,s_p) \in \bZpos^p$, we denote 
\begin{align}\label{eq: definition of N}
s_0 \coloneqq 0, 
\qquad
N = \frac{|\multii|}{2} \coloneqq \frac{1}{2} \sum_{j=1}^p s_j ,
\qquad \textnormal{and} \qquad 
\summ_k \coloneqq s_0 + \cdots + s_k . %;
\end{align}
%where we also require that $\underset{1 \leq j \leq p}{\max} s_j \leq N$ for the link patterns to exist.

\begin{definition} \label{def:link_pattern}
A $\multii$-\emph{valenced \textnormal{(}planar\textnormal{)} link pattern} with $N$ links at $p$ points is a collection $\alpha$
of links $\link{a}{b}$ in $\bH$, each connecting a pair $a < b$ of indices $a,b \in \{1,2,\ldots,p \}$, 
such that
\begin{itemize}[leftmargin=*]
\item for any $j \in \{1,2, \ldots,p\}$, the index $j$ is an endpoint of 
exactly $s_j$ links, 
\item none of the links intersect in $\bH$; 
but only at their common endpoints in $\{1,2, \ldots,p\} \subset \bR$.
\end{itemize}
Note that set-theoretically, the collection of links of $\alpha$ is a multiset
\begin{align}\label{eq:alpha}  
\alpha = \big\{ \link{a_1}{b_1}, \ldots, \link{a_N}{b_N} \big\} \in \LP_\multii . 
\end{align} 
\end{definition}

When all of the valences equal one, i.e., $\multii = (1,1,\ldots,1) \eqqcolon (1)^{2N}$ (so $p = 2N$) 
the set of $N$-link patterns is denoted as $\LP_N \coloneqq \LP_{(1)^{2N}}$. 
The difference between a (non-valenced) link pattern and a valenced link pattern is that the latter allows curves to have common endpoints at the boundary, with multiplicities labeled by valences $s_j$. 
Valenced link patterns can be obtained from link patterns via a fusion procedure, provided that no links connect endpoints to be fused together\footnote{We shall see an example in the context of the scaling limit of the metric graph GFF in Section~\ref{sec:mGFF}.}.  
Conversely, one can separate valenced points into distinct points, thereby \quote{unfusing} a valenced link pattern:
\begin{align} \label{eq:iotamap}
\LP_\multii \to \LP_N,  \qquad
& \alpha \mapsto \imath(\alpha) .
\end{align}
The map $\imath$ is defined combinatorially as follows: split each index $j \in \{1,2,\ldots,p \}$ of $\alpha$ 
into $s_j$ distinct indices, and attach the $s_j$ links ending at $j$ in $\alpha$ to these new $s_j$ indices,  
so that the link possessing the leftmost endpoint gets the leftmost one of the new indices, etc. 
Then, label all indices from left to right by $1,2,\ldots,2N$ to obtain an $N$-link pattern $\imath(\alpha) \in \LP_N$.

We identify the set $\LP_\multii$ with the set of all Dyck paths over the multiindex $\multii$ 
via a bijection where, for each $1 \leq j \leq p$, if $j$ appears $\ell_j$ times as a left link endpoint in the link pattern $\beta$ and $r_j$ times as a right endpoint,
we set the $j$:th step of the Dyck path (also denoted as $\beta$) to be $\beta_j-\beta_{j-1} := \ell_j - r_j$ (see Figure~\ref{fig:walk}(left)).
This bijection requires $\multii$ to be known; for instance the link patterns $\{ \{ 1, 3\}, \{1, 2 \},    \{2, 3\}  \}$ and $\{ \{ 1, 2\}, \{1, 2\},  \{2, 3\} ,  \{2, 3\}  \}$ over the multiindices $(2, 2, 2)$ and $(2, 4, 2)$, respectively, both correspond to the Dyck path $(0, 2, 2, 0)$.

\subsubsection{The binary relation \quote{$\KWleq$}}

\begin{figure}
\includegraphics[width=.455\textwidth]{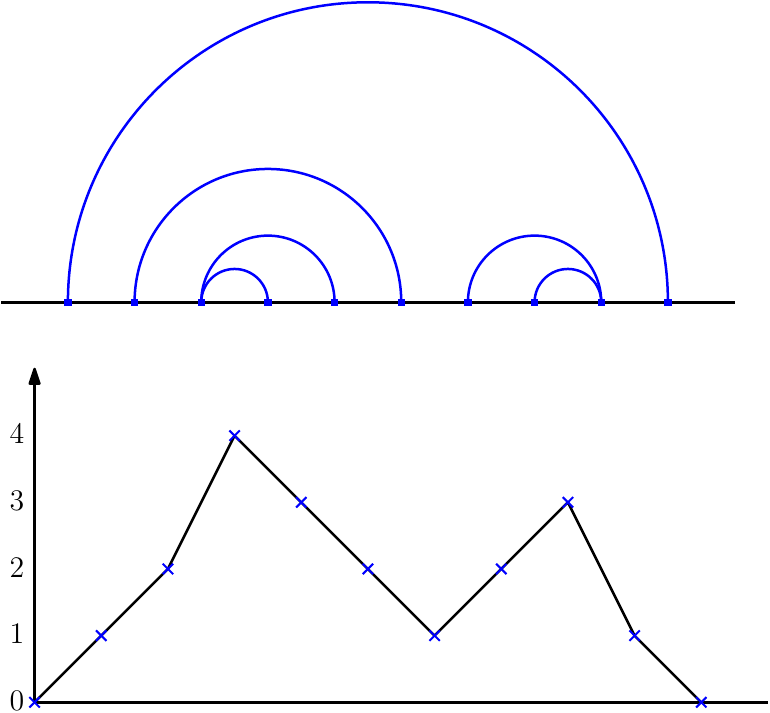} \quad
\includegraphics[width=.505\textwidth]{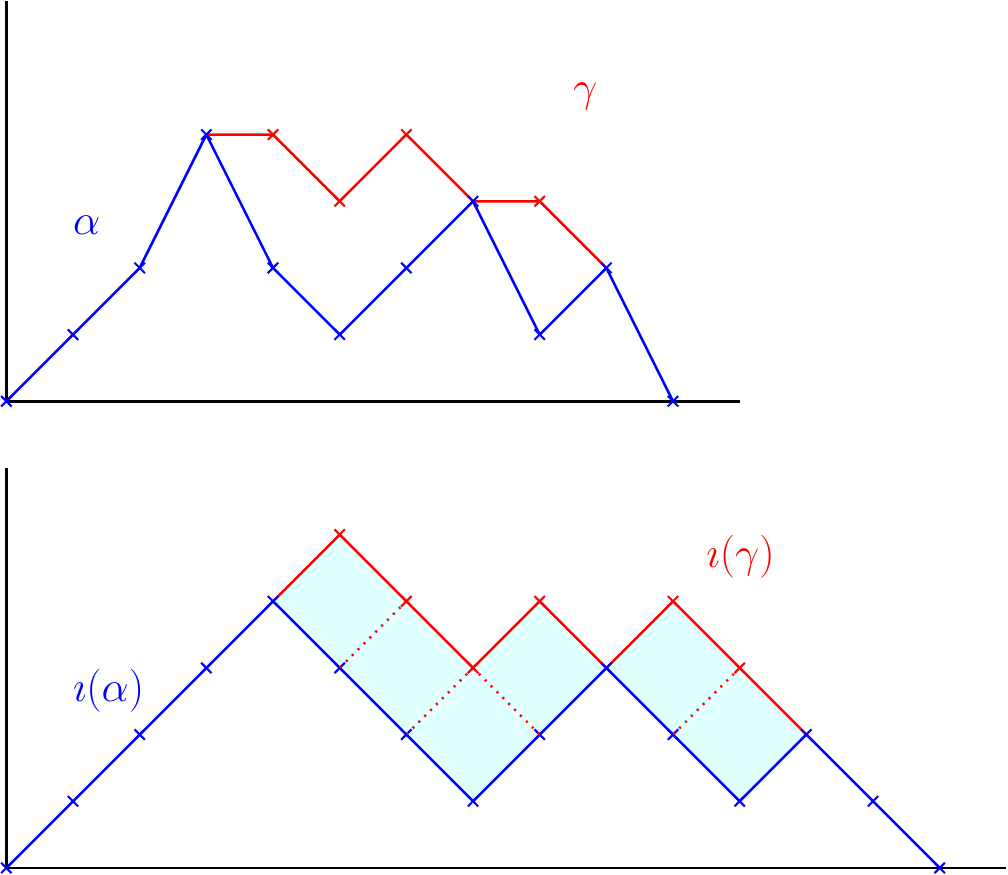}
\caption{
\label{fig:walk}
\label{fig:skew_Young_diagram}
Combinatorial objects central to the present work: 
The left panel illustrates 
a Dyck path $\beta = (0,1,2,4,3,2,1,2,3,1,0)$ over the multiindex $\multii = (1,1,2,1,1,1,1,1,2,1)$ and the corresponding $\multii$-valenced link pattern. Here, $\beta \in \GLP_\multii$. 
The right panel illustrates 
two Dyck paths $\alpha \DPleq \beta$ over the multiindex $\multii = (1,1,2,2,1,1,1,2,1,2)$ (top panel) and their unfused images $\imath (\alpha) \DPleq \imath (\beta) $
and the corresponding skew Young diagram $\alpha / \beta = \imath(\alpha) / \imath(\beta)$, highlighted in light blue color (bottom panel). 
In this example, the number of atomic squares for the skew Young diagram is $|\alpha/\beta| = 6$.  
}
\end{figure}

Let a Dyck path $\beta$ of $p$ steps and without constant steps (i.e., a GFF boundary condition) be given.
In analogy with Equation~\eqref{eq:harmonic_function}, set
\begin{align*}
U_\beta(x) \coloneqq 
\begin{cases}
0 , & \textnormal{if $x<1$ or $x>p$,}
\\ 
\beta_j ,  & \textnormal{if } x\in (j, j+1), \ j \in \{1,\dots,p-1\}.
\end{cases}
\end{align*}
We say that a valenced link pattern $\alpha \in \LP_\multii$ \emph{is a (combinatorial) level-line pattern for the boundary condition $\beta$} if
$\multii_j = |\beta_j-\beta_{j-1}|$ for all $1 \leq j \leq p$, 
and on the complement of the links of $\alpha$, there exists a function which is constant on each connected component of $\bH \setminus \alpha$ and has the boundary values $U_\beta$ on $\bR$ 
(see Figure~\ref{fig:combinatorial bijections}(right)). Lemma~\ref{lem:KWleq importance} below characterizes the possible combinatorial level-line patterns of given boundary data in terms of the binary relation \quote{$\smash{\KWleq}$.}
This will be used in the proof of Theorem~\ref{thm:CFT properties}, which also tells that the combinatorial level-line patterns indeed are exactly the possible GFF level-line patterns.

\begin{definition}[Maximally sloped, $\GLP_\multii$]\label{def:Dyck_max_sloped}
Fix $\multii = (s_1, s_2, \ldots, s_p)$. 
A Dyck path $\beta \in \LP_\multii$ over the multiindex $\multii$ is said to be \emph{maximally sloped} if it satisfies
$|\beta_{j+1} - \beta_{j}| = s_{j+1}$ for all $1 \leq j \leq p-1$.
We denote the set of all maximally sloped Dyck paths over $\multii$ by $\GLP_\multii$.
\end{definition}

Not all valences allow maximally sloped Dyck paths: for instance for $\multii = (2, 2, 2)$, there cannot exist a walk of three steps all of sizes $\pm 2$ that would start and end at height zero.
(The only possible link pattern in this case is illustrated in Figure~\ref{fig:impossible_pattern} in the introduction.)

\begin{figure}
\centering
\begin{minipage}{0.45\textwidth}
\raisebox{8.2em}{\includegraphics[width=\textwidth]{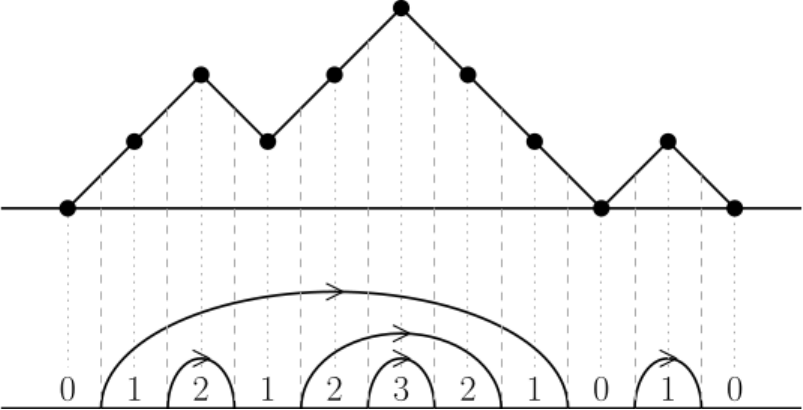}}
\end{minipage}
\begin{minipage}{0.45\textwidth}
\includegraphics[width=\textwidth]{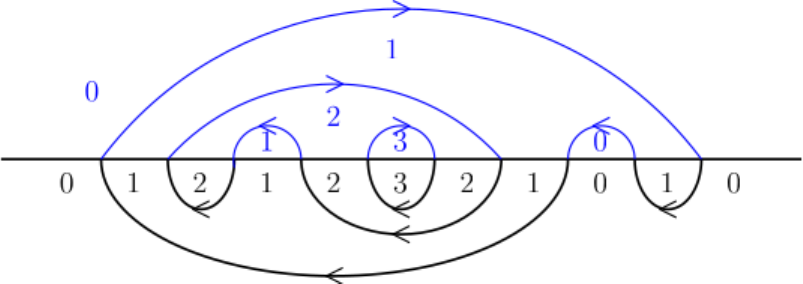}
\end{minipage}
\begin{minipage}{0.45\textwidth}
\raisebox{8.2em}{\includegraphics[width=\textwidth]{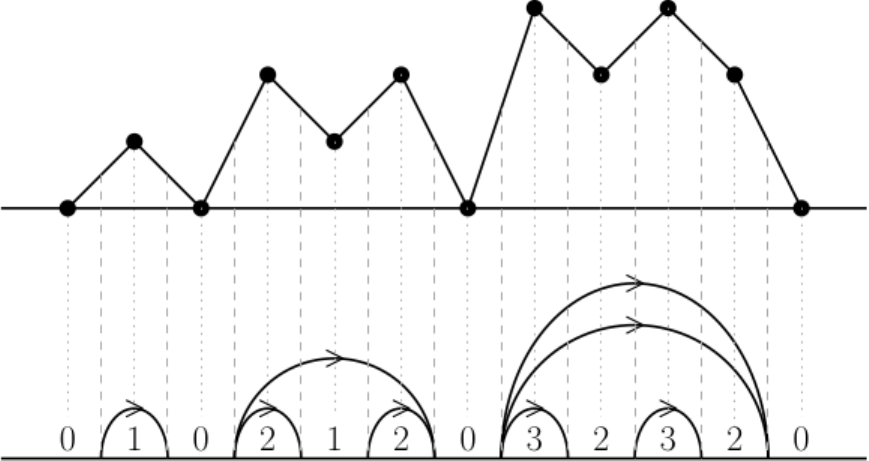}}
\end{minipage}
\begin{minipage}{0.45\textwidth}
\includegraphics[width=\textwidth]{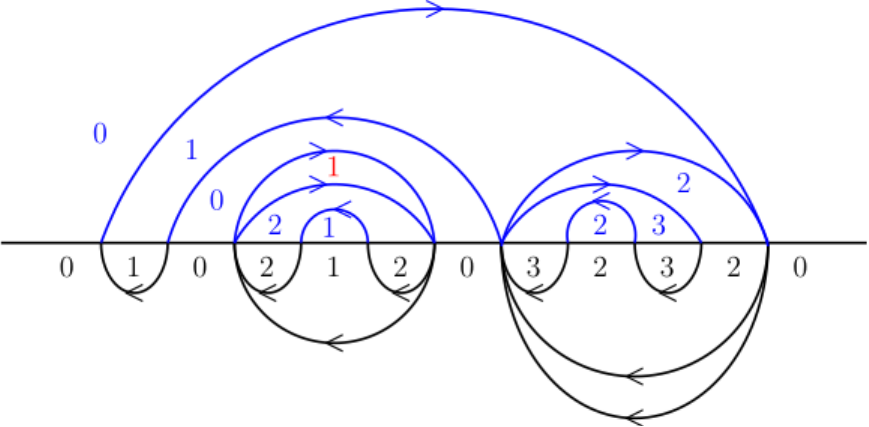}
\end{minipage}
\caption{
\label{fig:combinatorial bijections}
The left panel illustrates an unvalenced (top) and valenced (where $\beta \in \GLP_\multii$, bottom) example of the bijection between Dyck paths $\beta$ and link patterns $\beta$: up-steps correspond to left link endpoints, down-steps to right endpoints (with multiplicity in bottom panel). 
The bijection can be extended so that the Dyck path heights between the steps become the boundary condition $U_\beta$, and the link pattern yields the unique possible pattern of level lines with this boundary condition in which the oriented level lines (i.e., with larger height on the right) all flow from left to right. \\[.5em]
The right panel illustrates the relation \quote{$\smash{\KWleq}$}: the link patterns $\beta$ of the left panels are now drawn on the lower half-plane with all links oriented from right to left. Then, on the top panel, $\alpha \in \LP_N$ with $\beta \smash{\KWleq} \alpha$ (resp. $\alpha \in \LP_\multii$ with $\imath ( \beta) \smash{\KWleq} \imath (\alpha)$ on the bottom panel) is drawn in the upper half-plane. 
By~\cite[Lemma~2.5]{KKP:Boundary_correlations_in_planar_LERW_and_UST}, the relation $\smash{\KWleq}$ guarantees that the link pattern $\alpha$ can be oriented so as to form oriented loops together with $\beta$. By Lemma~\ref{lem:KWleq importance}, this also equivalently means that a locally constant function can be defined in $\bH \setminus \alpha$ so that it coincides with $U_\beta$ on $\bR$. In the valenced case, such a function is not uniquely determined by $\beta$ (the red number $1$ could be changed), but there is a unique such function with differences $\pm 1$ over the links of $\alpha$.
}
\end{figure}

We are now ready to characterize which $\multii$-valenced link patterns $\alpha \in \LP_\multii$ 
provide combinatorial level-line patterns for a given boundary condition $\beta \in \GLP_{\multii}$.

\begin{lemma}
\label{lem:KWleq importance}
Fix $\multii \in \bZpos^p$ and $\beta \in \GLP_\multii$. A valenced link pattern $\alpha \in \LP_\multii$ is a (combinatorial) level-line pattern for the boundary condition $\beta$ if and only if the relation $$\imath ( \beta) \KWleq \imath (\alpha)$$ introduced in~\cite[Definition~2.2]{KKP:Boundary_correlations_in_planar_LERW_and_UST} holds for the corresponding \quote{unfused} link patterns. 
In this case, we simply denote $\smash{\beta \KWleq \alpha}$ for $\alpha, \beta \in \LP_\multii$.
Furthermore, the corresponding locally constant function on $\bH \setminus \alpha$ can be chosen so that 
its values at neighboring connected components differ by~$\pm 1$.
\end{lemma}

\begin{remark}
\label{rem:KWleq importance}
The last part of the lemma also implies the following characterization of all $\alpha \in \LP_\multii$ that are combinatorial level-line patterns for the boundary condition $\beta \in \GLP_\multii$: for any $j \in \{1,\ldots, p\}$, with $U_\beta (j^-) < m+1/2 < U_\beta (j^+)$ for some integer $m$, add a level line of \quote{height} $m$ from $j$ to any $k \in \{1,\ldots, p\}$, with $U_\beta (k^-) > m+1/2 > U_\beta (k^+)$. 
Then, split $\bH$ into two simply connected domains, impose boundary data $m$ on the left and $(m+1)$ on the right side of the level line, and continue similarly. 
An inductive argument on $N$ in~\eqref{eq: definition of N} readily shows that this produces exactly all the link patterns satisfying the $\pm 1$ condition.
\end{remark}

\begin{proof}[Proof of Lemma~\ref{lem:KWleq importance}]
The proof is based on an extension of the bijection between Dyck paths and link patterns (see Figure~\ref{fig:combinatorial bijections}(left)): for $\beta \in \GLP_\multii$, we see that
\begin{itemize}
\item the heights of the Dyck path $\beta$ determine the boundary data $U_\beta$ between the link endpoints (dotted lines in Figure~\ref{fig:combinatorial bijections}, right panel); and
\item the links of the valenced link pattern $\beta$, when oriented from left to right, then form one possible pattern of level lines with this boundary data, with larger height on the right of each oriented link.
\end{itemize}
Now, suppose first that $\imath ( \beta) \smash{\KWleq} \imath (\alpha)$. 
By~\cite[Lemma~2.5]{KKP:Boundary_correlations_in_planar_LERW_and_UST}, this means that, drawing $\imath (\beta)$ on the lower half-plane with all links oriented from right to left, $\imath(\alpha)$ can be drawn in the upper half-plane with some orientation of its links so that $\imath(\alpha) \cup \imath(\beta)$ consists of oriented loops (see Figure~\ref{fig:combinatorial bijections}(left)). Define now a function on $\bC \setminus (\imath(\alpha) \cup \imath(\beta) )$ as zero on the infinite component, and upon moving inward across the loops of $\imath(\alpha) \cup \imath(\beta)$, by increasing (resp.~reducing) the height by one on each crossing of a clockwise (resp.~counterclockwise) loop. 
By the above observations (as this inward exploration can be performed in the lower half-plane), this function coincides with $U_{\imath (\beta)}$ on $\bR$, whence $\imath(\alpha)$ is a level-line pattern for the boundary condition $\imath (\beta)$. 
The analogous claim for $\alpha$ and $\beta$ follows by fusion.

Suppose now that $\alpha \in \LP_\multii$ is a level-line pattern for the boundary condition $\beta$. Note that $\bH \setminus \alpha$ may contain connected components that do not intersect $\bR$, so their height is not directly determined by $\beta$ (see Figure~\ref{fig:combinatorial bijections}, bottom-right panel). We however claim that, by adjusting such undetermined heights, it is possible to achieve a situation where all neighboring components of $\bH \setminus \alpha$ have a height difference of $\pm 1$. 
Assuming this claim for a moment, since $\beta \in \GLP_\multii$, \quote{unfusing} shows that $\imath ( \alpha) $ is a combinatorial level line pattern for the boundary condition $\imath ( \beta )$. By the above bullet points, also $\imath ( \beta) $ itself is a combinatorial level-line pattern for the boundary condition $\imath ( \beta )$. In particular, we can draw the links of $\imath(\alpha)$ on the upper half-plane and $\imath(\beta)$ on the lower half-plane and extend $U_\beta$ as a locally constant function to the complement of the links $\bC \setminus (\imath(\alpha) \cup \imath(\beta))$, with differences $\pm 1$ over each link. Now,  orienting the loops of $\imath(\alpha) \cup \imath(\beta)$ to yield the oriented level lines of this locally constant function on $\bC \setminus (\imath(\alpha) \cup \imath(\beta))$, we then obtain the above cited characterization of $\smash{ \imath ( \beta) \KWleq \imath (\alpha) }$ from~\cite[Lemma~2.5]{KKP:Boundary_correlations_in_planar_LERW_and_UST}.

We now prove the postponed claim by induction on $N$. 
For $N=1$, every component of $\bH \setminus \alpha$ has to intersect $\bR$, so there is nothing to prove. For general $N$, note that the infinite component of $\bH \setminus \alpha$ has height $0$, and assign to every component adjacent to it height $1$. 
Now, inside each such connected component, we have a new similar problem, but with a lower number of curves (i.e., a simply connected domain $D$, a boundary condition monotonous at fused boundary points, and some height to be adjusted on the connected components of $D \setminus \alpha$). The claim, and the entire proof, is thus completed by induction on $N$.
\end{proof}

\subsubsection{Further notation}

We now review the rest of the definitions used in Equation~\eqref{eq:PartF} and in Section~\ref{subsec:conformal blocks}.
(These concepts are also central in the companion article~\cite{Karrila-Peltola:Boundary_double-dimer_patterns_and_CFT}.)

The set $\LP_\multii$ of (generalized) Dyck paths also has a natural partial order, inherited by the valenced link patterns: for 
$\alpha = (\alpha_{0}, \alpha_{1}, \ldots, \alpha_p)$ and $\beta = (\beta_0, \beta_1, \ldots, \beta_{p})$,
we denote $\alpha \DPleq \beta$ if and only if $\alpha_j \leq \beta_j$ for all $j$. It is not difficult to show that given $\alpha, \beta \in \LP_\multii$, 
the property $\alpha \DPleq \beta$ is equivalent to $ \imath( \alpha) \DPleq \imath(\beta)$.
When $\alpha \DPleq \beta$, the area between the Dyck paths $\imath(\alpha)$ and $\imath(\beta)$ forms a \emph{skew Young diagram}, denoted by $\alpha / \beta$.
We let $|\alpha/\beta|$ denote the number of atomic square tiles in~$\alpha/\beta$, as illustrated also in Figure~\ref{fig:skew_Young_diagram}(right).

To each skew Young diagram~$\alpha/\beta$ for a pair $\alpha, \beta \in \LP_N$ of $N$-link patterns 
(or rather, Dyck paths $\alpha, \beta$, which are in bijection with $N$-link patterns), in~\cite[Definition~2.8]{KKP:Boundary_correlations_in_planar_LERW_and_UST} (following~\cite{Kenyon-Wilson:Double_dimer_pairings_and_skew_Young_diagrams})  
one associates a (non-empty) set $\CItilingsof (\alpha / \beta)$ of \emph{cover-inclusive Dyck tilings}, illustrated in~\cite[Figures~2.11~\&~2.12]{KKP:Boundary_correlations_in_planar_LERW_and_UST}.
Because we will not need the precise definition for deriving the results in the present work, we refer to the extensive overview in~\cite[Section~2]{KKP:Boundary_correlations_in_planar_LERW_and_UST} for these objects and their combinatorics.  
Slightly generalizing those references, for each pair $\alpha, \beta \in \LP_\multii$ of $\multii$-valenced link patterns (or rather, Dyck paths over $\multii$) such that $\alpha \DPleq \beta$, 
we let $\# \CItilingsof (\alpha / \beta)$ denote the number of cover-inclusive Dyck tilings of~$\imath(\alpha)/\imath(\beta)$.

Finally, we need another partial order relation on $\LP_N$. 
Given a valence vector $\multii \in \bZpos^p$ with $|\multii|=2N$, for each pair $\alpha, \beta \in \LP_N$ of $N$-link patterns (Dyck paths), 
we set $\beta \geq_\multii \alpha$ if 
\begin{align} \label{eq: partial order}
\alpha \DPleq \beta  
\qquad \textnormal{and in addition,} \quad \alpha_{\summ_k} = \; &  \beta_{\summ_k} ,
\quad \textnormal{for all } k \in \{ 0, 1, \ldots, p \} ,
\end{align} 
using the notation from Equation~\eqref{eq: definition of N}. 
For $\alpha \in \LP_\multii$, we denote 
\begin{align} \label{eq: partial order larger set}
\LP_{\geq_\multii \alpha} := \{ \beta \in \LP_N \; | \; \beta \geq_\multii \imath ( \alpha ) \} .
\end{align}

\subsection{Conformal block functions}
\label{subsec:conformal blocks}

Writing $\beta \in \LP_N$ as an $N$-link pattern with links as in Equation~\eqref{eq:alpha}, the well-known \emph{conformal block function} associated to $\beta$ reads
\begin{align} \label{eq:CobloF_nonval}
\CobloF_\beta(x_1,\ldots,x_{2N}) \coloneqq \prod_{1 \leq i < j \leq 2N} (x_j - x_i)^{\frac{1}{2} \vartheta(i,j)} 
= \frac{ \SpechtP_\beta(x_1,\ldots,x_{2N}) }{ ( \Delta (x_1,\ldots,x_{2N}) )^{1/2} } ,
\end{align}
where $\Delta (x_1,\ldots,x_{2N}) \coloneqq \underset{i < j}{\prod} \, (x_j - x_i)$ is the Vandermonde determinant, and
\begin{align*}
\vartheta(i,j) 
\coloneqq 
\begin{cases} 
+1, \quad i,j \in \{a_1,a_2,\ldots ,a_{N} \}  \; \textnormal{or} \; i,j \in \{b_1,b_2,\ldots ,b_{N}\}, \\
-1, \quad \textnormal{otherwise},
\end{cases}
\end{align*}
and $\SpechtP_\beta$ are \emph{Specht polynomials} defined as
\begin{align*}
\SpechtP_\beta(x_1,\ldots,x_{2N}) \coloneqq \; &
\prod_{1 \leq i < j \leq 2N} (x_j - x_i)^{\delta(i,j)} , \qquad \beta \in \LP_N ,
\\
\qquad \textnormal{where} \qquad 
\delta(i,j) \coloneqq \; & \delta_{\dir \beta_{j} , \dir \beta_i}
= \begin{cases}
1 , & \dir \beta_{j} = \dir \beta_i , \\
0 , & \dir \beta_{j} \neq \dir \beta_i .
\end{cases}
\end{align*}
We see that $\CobloF_\beta$ coincide with the conformal block functions defined in~\cite[Equation~(6.1)]{Peltola-Wu:Global_and_local_multiple_SLEs_and_connection_probabilities_for_level_lines_of_GFF}
and in~\cite[Definition~3.1]{LPR:Fused_Specht_polynomials_and_c_equals_1_degenerate_conformal_blocks}. 
In this context, the step directions 
\begin{align*}
\dir \beta_{j} \coloneqq \beta_{j} - \beta_{j-1} \in \{\pm 1\}, \qquad j \in \{1,\dots,2N\} ,
\end{align*}
encode jump directions in GFF boundary data as in~\cite[Sections~5-6]{Peltola-Wu:Global_and_local_multiple_SLEs_and_connection_probabilities_for_level_lines_of_GFF}.

\begin{definition} \label{def:VF}
We define 
\begin{align*}
\sV_\alpha(x_1,\ldots,x_{2N}) \coloneqq \sum_{\substack{\beta \in \LP_N \\ \alpha \leq_\multii \beta }}  (-1)^{|\alpha / \beta|} \, \CobloF_\beta(x_1,\ldots,x_{2N}) , \qquad \alpha \in \LP_\multii ,
\end{align*}
where \quote{$\leq_\multii$} is the partial order defined in Equation~\eqref{eq: partial order}, 
$|\alpha / \beta| \in \bZnn$ is the number of atomic square tiles in the skew Young diagram $\alpha / \beta = \imath(\alpha) / \beta$ (see Figure~\ref{fig:skew_Young_diagram}(right) for an illustration),
and $\CobloF_\beta$ is the conformal block function~\eqref{eq:CobloF_nonval}.
\end{definition}

\begin{lemma} \label{lem:VF_fusion}
Fix $x_1 <\cdots < x_p$ and $\beta \in \LP_\multii$. The following limit exists:
\begin{align} \label{eq:fused_conformal_block}
\CobloF_\beta (x_1,\ldots ,x_p) \coloneqq 
\lim_{\substack{ \xi_{1},\xi_{2},\ldots,\xi_{\summ_{1}} \to x_1 \\[.2em] \xi_{\summ_1+1},\xi_{\summ_1+2},\ldots,\xi_{\summ_{2}} \to x_2 \\[.1em] \vdots \\[.2em] \xi_{\summ_{p-1}+1},\xi_{\summ_{p-1}+2},\ldots,\xi_{\summ_{p}} \to x_p }} \frac{\sV_\beta(\xi_1,\ldots,\xi_{2N}) }{ \overset{p-1}{\underset{k=0}{\prod}} \; \underset{\summ_k < i < j \leq \summ_{k+1}}{\prod} (\xi_j - \xi_i)^{1/2} } . 
\end{align}
\end{lemma}

\begin{proof}
This is proven in~\cite{Karrila-Peltola:Boundary_double-dimer_patterns_and_CFT} based on the combinatorial Definition~\ref{def:VF} and~\cite[Proposition~3.14]{LPR:Fused_Specht_polynomials_and_c_equals_1_degenerate_conformal_blocks};
it is essentially a consequence of~\cite[Propositions~3.2~\&~3.14]{LPR:Fused_Specht_polynomials_and_c_equals_1_degenerate_conformal_blocks}. 
\end{proof}

\begin{definition}\label{def:CobloF}
We define the \emph{conformal block functions} $\CobloF_\beta$, for $\beta \in \LP_\multii$, by Equation~\eqref{eq:fused_conformal_block}.
\end{definition}

In the original~\cite[Definition~3.15]{LPR:Fused_Specht_polynomials_and_c_equals_1_degenerate_conformal_blocks}, 
the conformal block functions $\CobloF_\beta$ were indexed by column-strict Young tableaux.
In~\cite{Karrila-Peltola:Boundary_double-dimer_patterns_and_CFT} it is shown that these two definitions agree.

\bigskip 

The analogous fusion limit property holds for the pure partition functions:

\begin{lemma} \label{lem:PartF_fusion}
Fix $x_1 <\cdots < x_p$ and $\alpha  \in \LP_\multii$. The following limit exists:
\begin{align*}
\lim_{\substack{ \xi_{1},\xi_{2},\ldots,\xi_{\summ_{1}} \to x_1 \\[.2em] \xi_{\summ_1+1},\xi_{\summ_1+2},\ldots,\xi_{\summ_{2}} \to x_2 \\[.1em] \vdots \\[.2em] \xi_{\summ_{p-1}+1},\xi_{\summ_{p-1}+2},\ldots,\xi_{\summ_{p}} \to x_p }} \frac{\PartF_{\imath(\alpha)}(\xi_1,\ldots,\xi_{2N}) }{ \overset{p-1}{\underset{k=0}{\prod}} \; \underset{\summ_k < i < j \leq \summ_{k+1}}{\prod}  (\xi_j - \xi_i)^{1/2} } . 
\end{align*}
Moreover, this limit equals the function $\PartF_{\alpha}(x_1,\ldots,x_p)$ defined in Equation~\eqref{eq:PartF}.
\end{lemma}

\begin{proof}
The full proof is given in~\cite{Karrila-Peltola:Boundary_double-dimer_patterns_and_CFT}. The starting point is Equation~\eqref{eq:PartF} for the \quote{unfused} pure partition functions $\PartF_{\imath(\alpha)}$ (originating from \cite[Theorem~1.5]{Peltola-Wu:Global_and_local_multiple_SLEs_and_connection_probabilities_for_level_lines_of_GFF}):
\begin{align}
\label{eq:unfused partF}
\PartF_{\imath (\alpha)} (\xi_1,\ldots,\xi_{2N}) = \sum_{\substack{\beta \in \LP_N \\ \imath (\alpha) \DPleq \beta }} (-1)^{|\alpha/\beta|} \, \# \CItilingsof ( \alpha / \beta ) \, \CobloF_\beta (\xi_1,\ldots,\xi_{2N}) .
\end{align}
If $\alpha \in \LP_\multii$, then the right-hand side can be non-trivially 
transferred to linear combinations in the smaller set $\{ \sV_{\beta} \colon \beta \in \LP_\multii\}$, whose limits are then handled using Lemma~\ref{lem:VF_fusion}. 
\end{proof}

\subsection{Schramm-Loewner evolution and its variants}
\label{subsec:SLE}

We next introduce $\SLE_\kappa$ and $\SLE_\kappa(\ul{\rho})$ processes~\cite{Rohde-Schramm:Basic_properties_of_SLE,Schramm-Wilson:SLE_coordinate_changes,Schramm:ICM}. 
These are random curves which arise in the context of statistical mechanics, and their geometry and regularity have been studied extensively.

Consider the Loewner differential equation
\begin{align}\label{eq:LDE}
\begin{split}
g_0(z) = \; &  z , \\
\partial_t g_t(z) = \; &  \frac{2}{g_t(z) - W(t)}, \qquad  \; t > 0 ,
\end{split}
\end{align} 
where $t \mapsto W(t)$ is a continuous and real-valued function.  
For each $z \in \ol{\bH}$, the solution to the ODE~\eqref{eq:LDE} exists up until the swallowing time $T_z \coloneqq \inf\{ t \geq 0 \colon g_t(z) - W(t) = 0 \}$ and for each $t \geq 0$, we define the compact $\bH$-hull $K_t = \{ z \in \bH \colon T_z \leq t \}$.  
Then, $g_t$ is a conformal bijection $g_t \colon \bH \setminus K_t \to \bH$ normalized so that $|g_t(z) - z| \to 0$ as $|z| \to \infty$. 
The family of conformal maps $(g_t)_{t \geq 0}$ is called the \emph{Loewner chain} generated by $W$.  
We say that the Loewner chain $(K_t)_{t \geq 0}$ is \emph{generated by a curve} if there is a curve $\gamma \colon [0,\infty) \to \ol{\bH}$ such that $K_t$ is the complement of the unbounded connected component of $\bH \setminus \gamma[0,t]$, for each $t \geq 0$.

To construct chordal $\SLE_\kappa$, one solves the ODE~\eqref{eq:LDE} with $W = \sqrt{\kappa} B$, where $B$ is a standard Brownian motion and $\kappa \geq 0$ a diffusivity parameter\footnote{The geometry of $\eta$ depends heavily on $\kappa$ and the curves exhibit three phases: almost surely, $\eta$ is simple if $\kappa \leq 4$, self-intersecting if $\kappa \in (4,8)$ and space-filling if $\kappa \geq 8$. We will be concerned with the critical case $\kappa=4$.}.  
Then, one obtains a family of compact $\bH$-hulls $(K_t)_{t \geq 0}$, whose associated Loewner chain is generated by a curve $\eta$~\cite{Rohde-Schramm:Basic_properties_of_SLE}. 
This curve is an $\SLE_\kappa$ curve in $\bH$ from $0$ to $\infty$ and write $\eta \sim \SLE_\kappa$ in $(\bH;0,\infty)$.

$\SLE_\kappa$ curves are scale-invariant (due to the scaling property of Brownian motion),
and thus, in other domains they can be defined as conformal images of $\SLE_\kappa$ curves in $\bH$.
Indeed, if $\domain \subsetneq \bC$ is a simply connected domain and $x,y \in \partial \domain$ distinct boundary points,
then $\eta$ is an $\SLE_\kappa$ curve in $(\domain;x,y)$ if and only if $\varphi^{-1}\circ \eta$ is an $\SLE_\kappa$ curve in $(\bH;0,\infty)$, where $\varphi \colon \bH \to \domain$ is a conformal bijection such that $\varphi(0) = x$ and $\varphi(\infty) = y$. 

Moreover, $\SLE_\kappa$ curves are time-reversible in the sense that if $\eta \sim \SLE_\kappa$ in $(\domain;x,y)$ is parametrized as $\eta \colon [0,T] \to \ol{\domain}$ so that $\eta(0)=x$ and $\eta(T)=y$, 
then its time-reversal $\wt{\eta} \colon [0,T] \to \ol{\domain}$ defined as $\wt{\eta}(t) \coloneqq \eta(T-t)$ has the same law as an $\SLE_\kappa$ curve in $(\domain;y,x)$. 
For the special case of $\kappa=4$ of interest to us, this is an immediate consequence of the level-line description of the $\SLE_4$ curve. In general, this is a highly nontrivial result~\cite{Zhan:Reversibility_of_chordal_SLE,Miller-Sheffield:Imaginary_geometry3}.

For more details on $\SLE$s, see, e.g.,~\cite{Lawler:SLE, Kemppainen:SLE_book}, and for their interaction with the GFF, see, e.g.,~\cite{Dubedat:SLE_and_free_field, Powell-Werner:Lecture_notes_on_the_GFF} and the discussion and references in the next Section~\ref{subsec:GFF}.

\bigskip

Fix $x_{\ell}^{\Left} < x_{\ell-1}^{\Left} < \cdots < x_{1}^{\Left} \leq 0 \leq x_{1}^{\Right} < \cdots < x_{r-1}^{R} < x_{r}^{R}$, and write $\ul{x}^\Left = (x_{\ell}^{\Left},\dots,x_{1}^{\Left})$, 
and $\ul{x}^\Right = (x_{1}^{\Right},\dots,x_{r}^{\Right})$,  
and $\ul{\rho}^\Left = (\rho_{1}^{\Left},\dots,\rho_{\ell}^{\Left})$, 
and $\ul{\rho}^\Right = (\rho_{1}^{\Right},\dots,\rho_{r}^{\Right})$, 
where $\rho_{j}^{\bullet} \in \bR$, for $\bullet \in \{ \Left,\Right\}$, are the \emph{weights} of the force points $x_{j}^{\bullet}$.  
Let $W$ be the solution to the SDEs
\begin{align}\label{eq:sle_kappa_rho_driver_sde}
\ud W(t) \; = \; \; & \sqrt{\kappa} \, \ud B(t)
\; + \;  \sum_{j = 1}^\ell \frac{\rho_{j}^{\Left} \, \ud t }{W(t) - V_{j}^{\Left}(t)}
\; + \;  \sum_{j=1}^r \frac{\rho_{j}^{\Right} \, \ud t}{W(t) - V_{j}^{\Right}(t)} , \\
\ud V_{j}^{\bullet}(t) \; = \; \; & \frac{2 \, \ud t}{V_{j}^{\bullet}(t) - W(t)} , \quad V_{j}^{\bullet}(0) = x_{j}^{\bullet}, \quad j\in \{1,\dots,N^\bullet\}, \; \bullet \in \{\Left,\Right\}, \nonumber
\end{align}
where $N^\Left = \ell$ and $N^\Right = r$ are the numbers of force points on the left and right side of the origin, respectively. 
By solving the ODE~\eqref{eq:LDE} with $W$ (for the maximal time interval where a solution makes sense, see below), 
we obtain a Loewner chain $(g_t)_{t \geq 0}$ and an associated family of compact $\bH$-hulls $(K_t)_{t \geq 0}$, which is also generated by a curve; namely
an $\SLE_\kappa(\ul{\rho}^\Left;\ul{\rho}^\Right)$ curve in $(\bH;0,\infty)$ with force points $\ul{x}^\Left, \ul{x}^\Right$.  
It is a natural generalization of $\SLE_\kappa$ involving boundary points $x_{j}^{\bullet}$ which either attract (if $\rho_{j}^{\bullet} < 0$) or repel (if $\rho_{j}^{\bullet} > 0$) the curve --- see Lemma~\ref{lem:force_points}. 
These $\SLE_\kappa(\ul{\rho}^\Left;\ul{\rho}^\Right)$ curves are absolutely continuous with respect to $\SLE_\kappa$ curves away from the boundary and, as such, exhibit the same dependence on $\kappa$.

As hinted above, the above SDEs might not necessarily make sense for infinite time and thus the $\SLE_\kappa(\ul{\rho}^\Left;\ul{\rho}^\Right)$ will not reach the target point at $\infty$. 
(For example, the points $V_{j}^{\bullet}(t)$ can collide with each other, and upon doing so, merge and never separate, as they follow the same ODE.) 
Indeed, the solution to Equation~\eqref{eq:sle_kappa_rho_driver_sde} exists until the \emph{continuation threshold}
\begin{align}\label{eq:continuation_threshold}
T \coloneqq \inf\Big\{ t \geq 0 \colon \textnormal{ either } \rhoSum{m(t)}^{\Left} \leq -2 
\; \textnormal{ or } \rhoSum{n(t)}^{\Right} \leq -2 \Big\} ,
\end{align}
where $m(t) \coloneqq \max\{j \colon V_{j}^{\Left}(t) = W(t)\}$ and $n(t) \coloneqq \max\{j \colon V_{j}^{\Right}(t) = W(t)\}$, and
\begin{align*} %\label{eq:sum_of_weights}
\rhoSum{j}^{\bullet} \coloneqq \sum_{k = 0}^j \rho_{k}^{\bullet}, \quad j\in \{1,\dots,N^\bullet\},\ \bullet \in \{ \Left,\Right\} ,
\end{align*} 
with the convention that $\rho_{0}^{\Left} = \rho_{0}^{\Right} = 0$, $x_{0}^{\Left} = 0^-$, $x_{0}^{\Right} = 0^+$, $x_{\ell+1}^{\Left} = -\infty$, $x_{r+1}^{\Right}= +\infty$. 
In other words, we may and do extend the solutions of the SDEs~\eqref{eq:sle_kappa_rho_driver_sde} over times when $V_{j}^{\bullet}(t) - W(t)$ hits zero, as long as $\rhoSum{m(t)}^{\Left} $ and $ \rhoSum{n(t)}^{\Right}$ remain at most $-2$. Geometrically, at such a time $t$, the curve $\eta$ hits the interval
$(x_{m(t)+1}^{\Left},x_{m(t)}^{\Left})$~or $(x_{n(t)}^{\Right},x_{n(t)+1}^{\Right})$, and for all times after such a boundary-hitting time, the force points swallowed by the curve remain coalesced: $V_{1}^{L}(t+s) = \cdots  = V_{m(t)}^{L}(t+s)$ or $V_{1}^{R}(t+s) = \cdots  = V_{n(t)}^{R}(t+s)$ for $s \geq 0$.

Like chordal $\SLE_\kappa$, we can define $\SLE_\kappa(\ul{\rho}^\Left;\ul{\rho}^\Right)$ in other domains by conformal invariance. 
If $\domain \subsetneq \bC$ is a simply connected domain, $x,y \in \partial \domain$ distinct boundary points,
and $\ul{x}^\Left = (x_{\ell}^{\Left},\dots,x_{1}^{\Left})$ and $\ul{x}^\Right = (x_{1}^{\Right},\dots,x_{r}^{\Right})$
respectively points on the clockwise and counterclockwise boundary arcs from $x$ to $y$, 
then $\eta \sim \SLE_\kappa(\ul{\rho}^\Left;\ul{\rho}^\Right)$ in $(\domain;x,y)$ with force points $(\ul{x}^\Left;\ul{x}^\Right)$ 
if and only if  
$\varphi^{-1}\circ \eta \sim \SLE_\kappa(\ul{\rho}^\Left;\ul{\rho}^\Right)$ in $(\bH;0,\infty)$, 
where $\varphi \colon \bH \to \domain$ is a conformal bijection
such that $\varphi(0) = x$ and $\varphi(\infty) = y$,
and the force points of the curve $\varphi^{-1}\circ \eta$ are given by $\varphi^{-1}(x_{\ell}^{\Left}),\dots,\varphi^{-1}(x_{1}^{\Left})$
and $\varphi^{-1}(x_{1}^{\Right}),\dots,\varphi^{-1}(x_{r}^{\Right})$.
Note that if there is more than one force point, then we cannot control the positions of the force points under the conformal map --- the definition of $\SLE_\kappa(\ul{\rho}^\Left;\ul{\rho}^\Right)$ involves nontrivial conformal moduli, which are encoded in the evolution of the force points in Equation~\eqref{eq:sle_kappa_rho_driver_sde}, 
or more precisely, to the $\SLE_\kappa$ \emph{partition function}~\cite{Schramm-Wilson:SLE_coordinate_changes,Kytola:On_CFT_of_SLE_kappa_rho,Lawler:Partition_functions_loop_measure_and_versions_of_SLE}
(compare with Equation~\eqref{eq:CobloF_ms} when $\kappa=4$)
\begin{align*}
\prod_{1 \leq i < j \leq \ell} (x_{i}^{\Left} - x_{j}^{\Left})^{\frac{\rho_{i}^{\Left} \rho_{j}^{\Left}}{2 \kappa}}
\prod_{1 \leq i < j \leq r} (x_{j}^{\Right} - x_{i}^{\Right})^{\frac{\rho_{i}^{\Right} \rho_{j}^{\Right}}{2\kappa}}
\prod_{\substack{1 \leq i \leq \ell \\ 1 \leq j \leq r}} (x_{j}^{\Right} - x_{i}^{\Left})^{\frac{\rho_{j}^{\Right} \rho_{i}^{\Left} }{2\kappa}}
\prod_{1 \leq i \leq \ell} (x_{i}^{\Left})^{\frac{\rho_{i}^{\Left}}{\kappa}}
\prod_{1 \leq j \leq r} (x_{j}^{\Right})^{\frac{\rho_{j}^{\Right}}{\kappa}} .
\end{align*}
Here, one could interpret the factors in the last two products as the distances of the force points to the starting point $0$ of the curve, 
which has weight $\rho_0 \coloneqq 2$. 

\bigskip

From now on, we shall focus on the case of $\kappa=4$, which is of primary interest from the GFF point of view.
By the results in~\cite[Lemma~2.1]{Miller-Wu:Intersections_of_SLE_paths:_the_double_and_cut_point_dimension_of_SLE} (see also~\cite[Lemma~15]{Dubedat:Duality_of_SLE} and~\cite[Remark~5.3~\&~Theorem~1.3]{Miller-Sheffield:Imaginary_geometry1}), we can describe precisely the interaction of $\SLE_4(\ul{\rho}^\Left;\ul{\rho}^\Right)$ curves with the boundary in terms of the weights of the force points.

\begin{lemma}\label{lem:force_points}
Consider $\eta \sim \SLE_4(\ul{\rho}^\Left;\ul{\rho}^\Right)$ in $(\bH;0,\infty)$ with force points $(\ul{x}^\Left;\ul{x}^\Right)$. Then,
\begin{enumerate}[label=(\roman{*}), ref=(\roman{*})]
\item if $\rhoSum{j}^{\Right} \geq 0$, then $\eta$ almost surely does not hit $(x_{j}^{\Right},x_{j+1}^{\Right})$;

\smallskip

\item if $\rhoSum{j}^{\Right} \in (-2,0)$, then $\eta$ can hit and bounce off of $(x_{j}^{\Right},x_{j+1}^{\Right})$, while $\eta \cap (x_{j}^{\Right},x_{j+1}^{\Right})$ still has empty interior, almost surely.
\end{enumerate}
The same holds if we replace \quote{$\Right$} by \quote{$\Left$} and consider $(x_{j+1}^{\Left},x_{j}^{\Left})$.
\end{lemma}

\subsection{Gaussian free field and local sets}
\label{subsec:GFF}

Fix a domain $\domain \subsetneq \bC$. 
Let $C_0^\infty(\domain)$ denote the space of smooth functions on $\domain$ with compact support. 
The \emph{Dirichlet inner product} 
is defined by\footnote{We note that for functions that vanish on the boundary $\partial \domain$, this definition agrees with the regularized Dirichlet norm of Equation~\eqref{eq: regularized Dirichlet norm}, and hence there is no confusion in having the same notation.} 
\begin{align*}
(f,g)_\nabla \coloneqq \int_\bC \nabla f(z) \cdot \nabla g(z) \, \ud z, \qquad f,g \in C_0^\infty(\domain), 
\end{align*}
where $\ud z$ is the two-dimensional Lebesgue measure.  
We denote by $H_0(\domain)$ the Hilbert space closure of $C_0^\infty(\domain)$ with respect to $(\cdot,\cdot)_\nabla$, let $(e_n)_{n \geq 1}$ be a $(\cdot,\cdot)_\nabla$-orthonormal basis of $H_0(\domain)$ and $(\zeta_n)_{n \geq 1}$ a sequence of independent 
standard normal random variables. Then, the zero-boundary \emph{Gaussian free field} (GFF) $\Phi$ on $\domain$ is defined as 
\begin{align} \label{eq:GFF:sum}
\Phi \coloneqq \sum_{n \geq 1} \zeta_n e_n.
\end{align}
The partial sums of~\eqref{eq:GFF:sum} converge in the Sobolev space $H^s(\bC)$ for each $s < 0$, and hence $\Phi$ is realized as a random distribution. Its law is independent of the choice of the basis. % $(e_n)_{n \geq 1}$.

It follows immediately from the conformal invariance of $(\cdot,\cdot)_\nabla$ that if $\varphi \colon \domain \to \wt{\domain}$ is a conformal bijection, then $\Phi \circ \varphi^{-1}$ is a zero-boundary GFF in $\wt{\domain}$, that is, the zero-boundary GFF is conformally invariant.
Another important property of the zero-boundary GFF is the \emph{domain Markov property}. That is, if we fix some open $O \subset \domain$, then we can write $\Phi = \Phi^{\domain \setminus O} + \Phi_{\domain \setminus O}$, where $\Phi^{\domain \setminus O}$ is a zero-boundary GFF on $O$, and $\Phi_{\domain \setminus O}$ is a distribution which agrees with $\Phi$ on $\domain \setminus O$ and is harmonic on $O$ and is independent of $\Phi^{\domain \setminus O}$.

Lastly, consider a real-valued function $f \colon \partial \domain \to \bR$ and let $u$ be the harmonic function on $\domain$ such that $u|_{\partial \domain} = f$. Then, the \emph{GFF with boundary data} $f$ is defined to be $\Phi + u$.

For more details on the GFF, see, e.g.,~\cite{Sheffield:GFF_for_mathematicians, Berestycki-Powell:Gaussian_free_field_and_Liouville_quantum_gravity}.

\subsubsection{Local sets}

We now turn our attention to \quote{local sets} of the GFF: 
random subsets of the domain of the field which are coupled with the field in a Markovian way, and have been studied extensively --- see e.g.~\cite{Aru-Sepulveda:Two-valued_local_sets_of_2D_continuum_GFF_connectivity_labels_and_induced_metrics,
ASW:BTLS,
ALS:First_passage_sets_of_the_2D_continuum_GFF,
ALS:First_passage_sets_of_2D_GFF}.

For a domain $\domain \subset \bC$ we denote by $\sK(\domain)$ the family of closed subsets of $\ol{\domain}$.
We shall view $(\sK(\domain), \disthaus)$ as a metric space endowed with the Hausdorff distance 
\begin{align}\label{eq:hausdorff_distance}
\disthaus(A,B) = \inf\{ \epsilon > 0 \colon B \subset A(\epsilon), \ A \subset B(\epsilon) \} , \qquad A,B \in \sK(\domain) ,
\end{align}
where $A(\epsilon)$ and $B(\epsilon)$ denote the Euclidean $\epsilon$-neighborhoods of $A$ and $B$.

Let $\Phi$ be a zero-boundary GFF in a domain $\domain$. A local set of $\Phi$ is a random variable which takes values in $\sK(\domain)$ with respect to which $\Phi$ satisfies a strong Markov property:

\begin{definition}\label{def:local_set}
A random set $A \in \sK(\domain)$ is a \emph{local set} of $\Phi$ if there is a coupling of $\Phi$, $A$, and a field $\Phi_A$ such that $\Phi_A$ restricts to a harmonic function on $\domain \setminus A$ and, conditionally on $(\Phi_A,A)$, 
the field $\Phi^A \coloneqq \Phi - \Phi_A$ is a zero-boundary GFF in $\domain \setminus A$.
\end{definition}

Given a local set coupling $(\Phi,A,\Phi_A)$, we denote by $h_A$ the function that is harmonic and equal to $\Phi_A$ in $\domain \setminus A$, and $0$ on $A$. We also denote by $\mathscr{F}_A$ the $\sigma$-algebra generated by $(A,\Phi_A)$.

\begin{lemma}[{See, e.g.,~\cite[Lemma~2.2]{Aru-Sepulveda:Two-valued_local_sets_of_2D_continuum_GFF_connectivity_labels_and_induced_metrics}}]\label{lem:local_set_properties}
Let $\Phi$ be a zero-boundary GFF on a domain $\domain \subsetneq \bC$.
\begin{enumerate}
\item\label{it:uniqueness_local_set} 
Any local set can be coupled in a unique way with a GFF: Let $(\Phi,A,\Phi_A,\wh{\Phi}_A)$ be a coupling where $(\Phi,A,\Phi_A)$ and $(\Phi,A,\wh{\Phi}_A)$ satisfy Definition~\ref{def:local_set}. Then, almost surely, $\Phi_A = \wh{\Phi}_A$.

\smallskip

\item\label{it:conditionally_independent_union} 
If $A$ and $B$ are local sets coupled with the same GFF $\Phi$ and $(A,\Phi_A)$ and $(B,\Phi_B)$ are conditionally independent given $\Phi$, then $A \cup B$ is also a local set coupled with $\Phi$. Additionally, $B \setminus A$ is a local set of $\Phi^A$ with $(\Phi^A)_{B \setminus A} = \Phi_{A \cup B} - \Phi_A$.

\smallskip

\item\label{it:increasing_local_sets} 
Suppose $(\Phi,A_n)$ is a local set coupling for every $n \in \bN$, and for some $k \in \bN$ independent of $n$, almost surely $A_n \cup \partial \domain$ has no more than $k$ connected components. 
Then $(\Phi,A_n,\Phi_{A_n})$ is tight and any subsequential limit is a local set coupling. Moreover, if in addition 
the sets $A_n$ are non-decreasing in $n$, then $A \coloneqq \ol{\bigcup_n A_n}$ is also a local set and $\Phi_{A_n} \to \Phi_{A}$ in probability as $n \to \infty$, in $H^{-1-\epsilon}(\domain)$ for bounded $\domain$ and in $H_\loc^{-1-\epsilon}(\domain)$ otherwise, for any $\epsilon>0$.

\smallskip

\item\label{it:decreasing_local_sets} 
Suppose that $(\Phi,A_n)$ is a local set coupling for every $n \in \bN$, and the sets $A_n$ are decreasing in $n$. 
Then $A \coloneqq \bigcap_n A_n$ is also a local set and $\Phi_{A_n} \to \Phi_{A}$ almost surely as $n \to \infty$. 
\end{enumerate}
\end{lemma}

\begin{remark}\label{rmk:union}
The local sets that we consider will always be determined by the GFF $\Phi$. 
Namely, any two such local sets will be conditionally independent given $\Phi$, and hence their union will be a local set of $\Phi$ (by Item~\ref{it:conditionally_independent_union} of Lemma~\ref{lem:local_set_properties}).
\end{remark}

\begin{remark}\label{rmk:boundary_data}
For a local set $A$ of $\Phi$, we say that $\Phi+u$ has boundary data $g \colon \partial A \to \bR$ on $\partial A$ if $h_A + u$ has boundary values given by $g$ on $\partial A$. We further note that Item~\ref{it:conditionally_independent_union} of  Lemma~\ref{lem:local_set_properties} implies that if two local sets $A$ and $B$ are conditionally independent, then the boundary values of $h_{A \cup B} + u$ agree with those of $h_A + u$ on $A \setminus B$ and with those of $h_B + u$ on $B \setminus A$.
\end{remark}

\subsubsection{Level lines}\label{subsubsec:level_lines}

An important class of local sets is that of \emph{level lines} of the GFF. The coupling of level lines is detailed in the following result, which 
uses the theory of $\SLE_4$ curves with force points from Section~\ref{subsec:SLE} (and we refer to the notation there, especially Equation~\eqref{eq:sle_kappa_rho_driver_sde}).

\begin{theorem}[{\cite{Wang-Wu:Level_lines_of_Gaussian_free_field_I}}] \label{thm:level_lines_coupling}
Fix a vector $(\ul{\rho}^\Left;\ul{\rho}^\Right)$ of weights, and let $(K_t)_{t \geq 0}$  be the hulls and $(f_t \coloneqq g_t - W(t))_{t \geq 0}$ the centered Loewner chain 
of an $\SLE_4(\ul{\rho}^\Left;\ul{\rho}^\Right)$ in $(\bH;0,\infty)$ with force points $(\ul{x}^\Left;\ul{x}^\Right)$. 
Let $\Phi$ be a zero-boundary GFF on $\bH$ and $u_t$ the bounded harmonic function on $\bH$ with boundary values 
\begin{align*}
u_t(x) \coloneqq 
\begin{cases}
-\lambda(1+\rhoSum{j}^{\Left}) , & x \in \big[f_t(x_{j+1}^{\Left}),f_t(x_j^{\Left})\big) , \\
\phantom{-}\lambda(1 + \rhoSum{j}^{\Right}) ,  & x \in \big(f_t(x_j^{\Right}),f_t(x_{j+1}^{\Right})\big],
\end{cases}
\end{align*}
with the convention that $\rho_0^{\Left} = \rho_0^{\Right} = 0$, $x_0^{\Left} = 0^-$, $x_{\ell+1}^{\Left} = -\infty$, $x_0^{\Right} = 0^+$ and $x_{r+1}^{\Right} = +\infty$.

If $\tau$ is any stopping time for the $\SLE_4(\ul{\rho}^\Left;\ul{\rho}^\Right)$ process which almost surely occurs before the continuation threshold (see Equation~\eqref{eq:continuation_threshold}), then $K_\tau$ is a local set of $\Phi$ and the conditional law of $(\Phi + u_0)|_{\bH \setminus K_\tau}$ given $K_\tau$ is equal to the law of $(\Phi + u_\tau) \circ f_\tau$.

In particular, one can couple $\eta \sim \SLE_4(\ul{\rho}^\Left;\ul{\rho}^\Right)$ and $\Phi$ such that $\eta[0,\tau]$ is a local set of $\Phi$, such that the boundary data of $(\Phi + u_0)$ on the left \textnormal{(}resp.~right\textnormal{)} side of $\eta[0,\tau]$ equals $-\lambda$ \textnormal{(}resp.~$\lambda$\textnormal{)}.
\end{theorem}

\begin{definition}\label{def:level_line}
We say that an $\SLE_4(\ul{\rho}^\Left;\ul{\rho}^\Right)$ process $\eta$, coupled with $\Phi + u$ as in Theorem~\ref{thm:level_lines_coupling} is a \emph{level line} of $\Phi + u$ starting from $0$.  
Moreover, for $q \in \bR$, we say that $\eta$ is the level line of $\Phi + u$ of height $q$ if $\eta$ is coupled as a level line of the field $\Phi + u - q$. 
Note that this is natural, because then the boundary data of $\eta$ with respect to $\Phi + u$ is $\pm \lambda + q$.
\end{definition}

\begin{remark}\label{rmk:height_definition_difference}
The definition of a level line of height $q$ varies in the literature. 
For instance, in~\cite{Wang-Wu:Level_lines_of_Gaussian_free_field_I} the level line of $\Phi + u$ of height $q$ is the level line of $\Phi + u + q$.
\end{remark}

By conformal invariance, it is possible to couple $\SLE_4(\ul{\rho}^\Left;\ul{\rho}^\Right)$ curves with GFF in other domains, as level lines. 
Indeed, consider a simply connected domain $\domain$ with distinct boundary points $x,y \in \partial \domain$. 
Write $x_0^{\Left} = x_0^{\Right} = x$ and $x_{\ell+1}^{\Left} = x_{r+1}^{\Right} = y$, and let $\ul{x}^\Left = (x_1^{\Left},\dots,x_{\ell}^{\Left})$ (resp.~$\ul{x}^\Right = (x_1^{\Right},\dots,x_r^{\Right})$) be a vector of points ordered in the clockwise (resp.~counterclockwise) order on the clockwise (resp.~counterclockwise) boundary arc of $\partial \domain$ from $x$ to $y$.
Also, let $\ul{\rho}^\Left = (\rho_1^{\Left},\dots,\rho_{\ell}^{\Left})$ and $\ul{\rho}^\Right = (\rho_{1}^{\Right},\dots,\rho_r^{\Right})$ be vectors of real numbers.  
Let $\Phi$ be a zero-boundary GFF on $\domain$ and let $u$ be a harmonic function on $\domain$ with boundary values
\begin{align*}
u(x) \coloneqq 
\begin{cases}
-\lambda(1+\rhoSum{j}^{\Left}) , & x \in [x_{j+1}^{\Left},x_j^{\Left}) , \\
\phantom{-}\lambda(1 + \rhoSum{j}^{\Right}) ,  & x \in (x_j^{\Right},x_{j+1}^{\Right}] .
\end{cases}
\end{align*} 
Then, $\eta \sim \SLE_4(\ul{\rho}^\Left;\ul{\rho}^\Right)$ in $(\domain;x,y)$ is a level line of $\Phi+ u$ 
if $\varphi(\eta)$ is a level line of the GFF $(\Phi + u) \circ \varphi^{-1}$,
where $\varphi \colon \domain \to \bH$ is a conformal bijection such that $\varphi(x) = 0$ and $\varphi(y) = \infty$.

Importantly, level lines coupled with $\Phi$ as in Theorem~\ref{thm:level_lines_coupling} are almost surely determined by $\Phi$, i.e., they are measurable with respect to $\Phi$~\cite[Theorem~1.1.2]{Wang-Wu:Level_lines_of_Gaussian_free_field_I}. 
Moreover, any level line is continuous up to and including the continuation threshold~\eqref{eq:continuation_threshold}, 
by~\cite[Theorem~1.1.3]{Wang-Wu:Level_lines_of_Gaussian_free_field_I}. 
Let $\eta_q^x$ denote the (unique) level line of $\Phi + u$ of height $q \in \bR$ started at $x \in \bR$.
Two level lines started at points $x_1 \leq x_2$ interact nontrivially~\cite[Theorem~1.1.4]{Wang-Wu:Level_lines_of_Gaussian_free_field_I}:
\begin{itemize}[leftmargin=*]
\item If $q_1 < q_2$, then $\eta_{q_1}^{x_1}$ almost surely stays to the left of $\eta_{q_2}^{x_2}$.
\item If $q_1 = q_2$, then $\eta_{q_2}^{x_2}$ can intersect $\eta_{q_1}^{x_1}$ and upon doing so, they merge and never separate. 
\end{itemize}

\begin{remark}\label{rmk:conditional_law_level_lines}
Since level lines are determined by $\Phi$, it follows that they are conditionally independent given $\Phi$.
Therefore, for any collection $\eta_{q_1}^{x_1},\dots,\eta_{q_n}^{x_n}$ of level lines, it follows from Lemma~\ref{lem:local_set_properties} and Theorem~\ref{thm:level_lines_coupling} that the conditional law of $\eta_{q_1}^{x_1}$ given $\eta_{q_2}^{x_2},\dots,\eta_{q_n}^{x_n}$, which we identify with $A \coloneqq \eta_{q_2}^{x_2} \cup \cdots \cup \eta_{q_n}^{x_n}$, 
is that of a level line of $\Phi + u$ of height $q_1$ from $x_1$ of the GFF $\Phi^{A} + h_{A} + u$.
We will use this fact repeatedly in the present work. 
\end{remark}

Finally, we note that there is a natural way to explore a level line from the endpoint to the starting point, that is, the time-reversed level line, as detailed in~\cite[Theorem~1.1.6]{Wang-Wu:Level_lines_of_Gaussian_free_field_I}. 
Namely, let $\eta$ be the level line of $\Phi + u$ from $0$ to $\infty$ and let $\wt{\eta}$ be the level line of $-(\Phi + u)$ from $\infty$ to $0$.
Then, on the event that the two curves do not hit the continuation threshold before reaching the target point, the two curves $\eta$ and $\wt{\eta}$ are equal, as sets.

\subsubsection{First passage sets}

Another important class of local sets is that of the first passage sets of $\Phi$. 
Intuitively, for $a \in \bR$, the first passage set $\bA_{-a}^u(\Phi) = \bA_{-a}(\Phi+u)$ of $\Phi$ comprises 
points which can be reached by travelling only along points where $\Phi + u$ take values at least $-a$. 

\begin{definition}\label{def:fps}
The (upper) first passage set (FPS) of level $-a < 0$ of $\Phi + u$ is the local set $\bA = \bA_{-a}^u(\Phi) = \bA_{-a}(\Phi+u)$ which contains $\partial \domain$ and satisfies the following properties.
\begin{enumerate}[label=(\roman{*}), ref=(\roman{*})]
\item Inside each component $O$ of $\domain \setminus \bA$, the harmonic function $h_{\bA} + u$ equals $-a$ on $\partial O \setminus \partial \domain$ and $u$ on $\partial O \cap \partial \domain$, in such a way that $h_{\bA} + u \leq -a$ on $O$. 

\smallskip

\item $\Phi_{\bA} - h_{\bA} \geq 0$, i.e., for any positive test function $f \in C_0^\infty(\domain)$, we have $(\Phi_{\bA} - h_{\bA},f) \geq 0$.

\smallskip

\item For any connected component $O$ of $\domain \setminus \bA$, for any $\epsilon > 0$, and for any $z \in \partial O \cap \partial \domain$ and all sufficiently small $\delta > 0$, almost surely, we have 
\begin{align*}
h_{\bA}(z) + u(z) \geq \min\Big\{ \!\!-a, \inf_{w \in B(z,\delta) \cap \ol{O}} u(z) \Big\} - \epsilon.
\end{align*}

\smallskip
	
\item Almost surely, $\bA$ contains no isolated points and each connected component of $\bA$ that does not intersect $\partial \domain$ has a neighborhood that does not intersect any other connected component of $\bA$. 
\qedhere 
\end{enumerate}
\end{definition}

\begin{definition}\label{def:frontier}
Let $O_1,O_2,\dots$ be the connected components of $\domain \setminus \bA_{-a}^u$ sharing a boundary segment with $\partial \domain$. Then, we say that the \emph{frontier} of the FPS $\bA_{-a}^u$ is the set 
\begin{align*}
\bigcup_{j \geq 1} \partial O_j \setminus \partial \domain .
\end{align*}
\end{definition}

\begin{theorem}[{\cite[Theorem~4.3]{ALS:First_passage_sets_of_the_2D_continuum_GFF}}] \label{thm:fps_existence_uniqueness}
The FPS $\bA_{-a}^u(\Phi)$ of level $-a$ exists and is unique in the sense that if $(\Phi,A)$ is a local set coupling such that $A$ is a FPS of level $-a$ for $\Phi + u$, then $A = \bA_{-a}^u(\Phi)$ almost surely.
\end{theorem}

\begin{lemma}[Monotonicity]\label{lem:fps_monotonicity}
If $a \leq \wt{a}$ and $u \leq \wt{u}$, then $\bA_{-a}^u \subset \bA_{-\wt{a}}^{\wt{u}}$ almost surely.
\end{lemma}

\subsection{Frontiers of the GFF first passage sets (FPS)}\label{subsec:Frontiers}

\begin{figure}
\includegraphics[width=.6\textwidth]{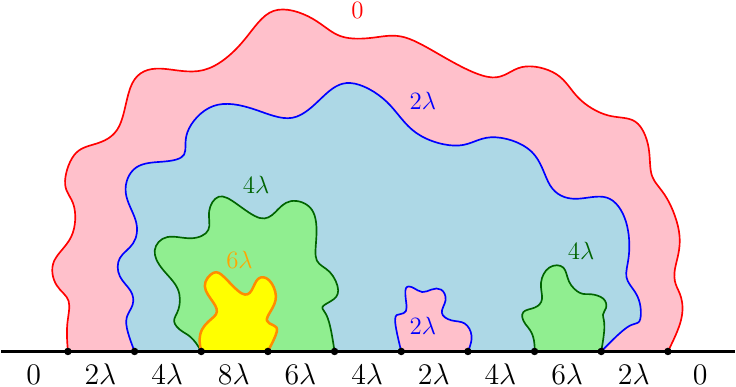}
\caption{Illustration of the first passage sets $\bA_0^u,\dots,\bA_{6\lambda}^u$, where $u = u_\beta$ and $\beta = (0,1,2,4,3,2,1,2,3,1,0)$. 
The filling of the following sets is given in the following color: $\bA_{6\lambda}^u$ in yellow, $\bA_{4\lambda}^u \setminus \bA_{6\lambda}^u$ in green, $\bA_{2\lambda}^u \setminus \bA_{4\lambda}^u$ in blue and $\bA_0^u \setminus \bA_{2\lambda}^u$ in red. 
The connected components of the frontiers, i.e., the curves $\curve_j^k$, are the darker boundaries of the filled sets.}
\label{fig:exploration}
\end{figure}

Let $\Phi$ be a zero-boundary GFF on $\bH$, and fix $\beta \in \GLP_\multii$ and boundary points 
\begin{align*}
-\infty = x_0 < x_1 < x_2 < \cdots < x_p < x_{p+1} = +\infty . 
\end{align*} 
Let $u \colon \ol{\bH} \to \bR$ be the bounded harmonic function in $\bH$ with piecewise constant boundary values obtained from height gaps being even multiples of $\lambda \coloneqq \sqrt{\pi/8}$, 
encoded in $\beta$ as
\begin{align}\label{eq:harmonic_function2}
u(x) = u_\beta(x) \coloneqq 
\begin{cases}
0 , & x \in (-\infty,x_1) \cup (x_p,+\infty), \\
2\beta_j \lambda ,  & x\in (x_j,x_{j+1}), \ j \in \{1,\dots,p\}. 
\end{cases}
\end{align}
The first passage set $\bA_{2k\lambda}^u$ is non-trivial only for $0 \leq k \leq \maxu-1$, where $\maxu = \maxu_\beta = \tfrac{1}{2\lambda} \max u$, and if $N_k$ is the number of connected components of the set $u^{-1}(\{x \in \bR \; | \; u(x) > 2k\lambda \})$, then the frontier of $\bA_{2k\lambda}^u$ consists of $N_k$ curves, which we denote by $\curve_1^k,\dots,\curve_{N_k}^k$. By the monotonicity of first passage sets (Lemma~\ref{lem:fps_monotonicity}), 
different curves $\smash{\curve_{j_1}^{k_1}}$ and $\smash{\curve_{j_2}^{k_2}}$ cannot cross each other --- 
from~\cite[Lemma~3.8]{Aru-Sepulveda:Two-valued_local_sets_of_2D_continuum_GFF_connectivity_labels_and_induced_metrics} 
it thus follows that they can only intersect at their starting points or endpoints if $k_1 \neq k_2$ 
(since the height difference of two such FPS is at least $2\lambda$).  
Moreover, we have $N_0+\cdots+N_{\maxu - 1} = N$ (where $N$ is given by~\eqref{eq: definition of N}). 
We let 
\begin{align}\label{eq:cc_of_FPS_front}
S := \{ \curve_j^k \; | \; k \in \{0, 1,\ldots,\maxu - 1\}, \ j \in \{1,\ldots,N_k\} \} 
\end{align}
denote the collection of the connected components of the frontiers of the first passage sets $(\bA_{2k\lambda}^u)_{k \in \bZ}$. 
If $\maxu =1$, the curves $(\curve_j^0)_{j=1,\ldots,N_0}$ are multiple $\SLE_4$ curves and form the frontier of the first passage set $\smash{\bA_0^u}$.  
This process is well-known, and results analogous to ours were obtained in~\cite{Peltola-Wu:Global_and_local_multiple_SLEs_and_connection_probabilities_for_level_lines_of_GFF}. 
We shall now see that, for each $k \in \{0,\ldots,\maxu - 1\}$, the curves $(\curve_j^k)_{j=1,\dots,N_k}$ are actually the level lines of $\Phi + u_\beta$ with height $(2k+1)\lambda$. 
This is essentially the content of~\cite[Proposition~5.12]{ALS:First_passage_sets_of_2D_GFF}, but we provide a proof in the continuum.

In what follows, we will use the following notation\footnote{For a function $f \colon \bR \to \bR$ and $x \in \bR$, we write $f(x^+) = \lim_{y \nearrow x} f(y)$ and $f(x^-) = \lim_{y \searrow x} f(y)$.} 
(see Figure~\ref{fig:ABnotation}):
\begin{itemize}[leftmargin=2em]
\item $\DC(u) = \{x_1,\dots,x_p \}$ is the set of discontinuity points of $u = u_\beta$ as in Equation~\eqref{eq:harmonic_function2};
\smallskip
\item $\smash{\PJ_m = \PJ_m(u)= \{x_1^m < \cdots < x_{n_m}^m\} = \{ x \in \DC(u) \, | \, u(x^-) \leq 2(m-1)\lambda < u(x^+) \}}$;
\smallskip
\item $\smash{\NJ_m = \NJ_m(u)= \{y_1^m <\cdots < y_{n_m}^m\} = \{ y \in \DC(u) \, | \, u(y^+) < 2m\lambda \leq u(y^-) \}}$;
\smallskip
\item $\smash{\PJ = \PJ(u) = \bigcup_{m=1}^{\maxu} \PJ_m}$ and $\smash{\NJ = \NJ(u) = \bigcup_{m=1}^{\maxu} \NJ_m}$, where $\maxu =\maxu_\beta = \tfrac{1}{2\lambda} \max u_\beta$.
\end{itemize}

\begin{figure}[h!]
\includegraphics[width=.6\textwidth]{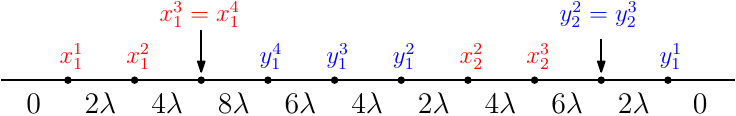}
\\[2em]
\caption{Illustration of the sets $\PJ$ and $\NJ$. The points in $\PJ$ (resp.~$\NJ$) are colored red (resp.~blue).}
\label{fig:ABnotation}
\end{figure}

In particular, $\PJ$ (resp.~$\NJ$) comprises the points at which positive (resp.~negative) jumps of the boundary data of $u$ occur. Observe that for $\Phi + u$, the points in $\PJ_m$ are exactly the points from which we can grow a level line of height $(2m-1)\lambda$, for the boundary data to the left (resp.~right) is at most $2(m-1)\lambda$ (resp.~at most $2m\lambda$). 
Indeed, the boundary values are too large or small in other points, so the weights of the force points immediately to the left or right of the starting point are at most $-2$, 
whence the continuation threshold is immediately hit (recall Theorem~\ref{thm:level_lines_coupling}~\&~Lemma~\ref{lem:force_points}). Moreover, by the same results, a level line of height $(2m-1)\lambda$ necessarily terminates at some point in $\NJ_m$ (cf.~\cite[Lemma~15]{Dubedat:Duality_of_SLE}).

Combinatorially, writing $\beta \in \GLP_\multii$ in its valenced link pattern representation~\eqref{eq:alpha}, 
\begin{align}\label{eq:beta}
\beta = \big\{ \link{a_1}{b_1}, \ldots, \link{a_N}{b_N} \big\}  ,
\end{align} 
we see that each pair $\{a_j,b_j\}$ corresponds to one point in $\PJ$ and one point in $\NJ$. 
Observe that which point belongs to which set may be different for each $j$, as the links $\link{a_j}{b_j}$ in $\beta$
in~\eqref{eq:beta} follow a left-to-right ordering, whereas the points in $\PJ$ (resp.~$\NJ$) are the starting points (resp.~endpoints) of the level lines in their natural orientation. 

Consider the set of level lines of $\Psi = \Phi + u_\beta$ given in Equation~\eqref{eq:OML}, 
\begin{align*}
\OML_\beta = \OML_\beta(\Phi) \coloneqq \big\{ \eta \;|\; \eta \textnormal{ is a level line of } \Phi + u_\beta \textnormal{ of height } (2k+1)\lambda \textnormal{ for some } k \in \bZ \big\} ,
\end{align*} 
whose heights are odd integer multiples of $\lambda$. 

\begin{lemma}\label{lem:frontiers_are_level_lines}
If $O \subset \bH \setminus \bA_{2(m-1)\lambda}(\Phi+u) = \bH \setminus \bA_{-2(\maxu -m+1)\lambda}(\Phi+u - 2\maxu \lambda)$ is a connected component such that $\partial O \cap \bH \neq \emptyset$, then each connected component of $\partial O \setminus \bR$ is a level line of $\Phi$ of height $(2m-1)\lambda$, started from some point in $\PJ_m$, for each $m \in \{1,2,\ldots, \maxu \}$. 
In~particular, $\OML_\beta$ is exactly the collection of connected components of frontiers of the FPS $(\bA_{2(m-1)\lambda}^u)_{m \in \bZ}$.
\end{lemma}

\begin{proof}
Started at each $x_a^m \in \PJ_m$, sample the level line $\eta^{a,m} \coloneqq \eta_{(2m-1)\lambda}^{x_a^m}$
of $\Phi + u$ of height $(2m-1)\lambda$  (that is, so that the boundary values on the level line belong to $\{2(m-1)\lambda,2m\lambda\}$), as in Definition~\ref{def:level_line}~\&~Theorem~\ref{thm:level_lines_coupling}.
As discussed above, $\eta^{a,m}$ terminates in a point in $\NJ_m$.

Now, $\wh{A}_m = \bigcup_{a = 1}^{n_m} \eta^{a,m}$ is a local set of $\Phi$ and $\bH \setminus \wh{A}_m = \bigcup_j O_j^m$ is a union of simply connected domains $O_j^m$ in which 
the harmonic function $h_{\wh{A}_m} + u$ has boundary values either at most $2(m-1)\lambda$ everywhere (so that $h_{\wh{A}_0} + u \leq 2(m-1)\lambda$ in $O_j^m$), 
or at least $2m \lambda$ everywhere (so that $h_{\wh{A}_0} + u \geq 2m\lambda$ in $O_j^m$). 
Let $O_{j_1}^m,\dots,O_{j_l}^m$ be the connected components of $\bH \setminus \wh{A}_m$ in which $h_{\wh{A}_m} + u \leq 2(m-1)\lambda$ and set $A^m \coloneqq \wh{A}_m \cup (\bigcup_l \wh{A}_l^m)$, where
\begin{align*}
\wh{A}_l^m \coloneqq \bA_{2(m-1)\lambda}((\Phi+u)|_{O_{\ell_l}^m}) = \bA_{-2(\maxu -m+1)\lambda}((\Phi + u -2\maxu \lambda)|_{O_{\ell_l}^m}).
\end{align*}
Then, $A^m$ is a first passage set of level $2(m-1)\lambda$ (all four conditions of Definition~\ref{def:fps} are satisfied by construction). 
Moreover, we have $A^m = \bA_{-2(\maxu-m+1)\lambda}(\Phi + u - 2\maxu\lambda)$ by uniqueness (Theorem~\ref{thm:fps_existence_uniqueness}). 
In particular, each outer boundary of $\bA_{-2(\maxu-m+1)\lambda}(\Phi + u - 2\maxu \lambda)$ which is not a subset of $\bR$ is a level line of $\Phi + u$ of height $(2m-1)\lambda$.
\end{proof}

\bigskip

\section{GFF partition functions, connection probabilities, and $c=1$ conformal blocks} 
\label{sec:GFF_properties}

In this section, we relate conformal block functions $\CobloF_\beta$, connection probabilities, and GFF partition functions. 
In Section~\ref{subsec:convergence_connection_probabilities} we prove our main result, Theorem~\ref{thm:corss_proba_H}, establishing the connection probabilities for the level lines of a GFF with more general boundary data. 
Then, in Proposition~\ref{prop:partition_functions_agree} in Section~\ref{subsec:partition_functions_conformal_blocks} we prove a relationship between the GFF partition functions \`a la Dub\'{e}dat~\cite{Dubedat:SLE_and_free_field} 
and the $c = 1$ degenerate conformal blocks \`a la~\cite{LPR:Fused_Specht_polynomials_and_c_equals_1_degenerate_conformal_blocks}, 
before wrapping up with Section~\ref{subsec:properties_pure_partition_functions}, where we collect the proof of the CFT properties of the pure partition functions $\PartF_\alpha$ stated in Theorem~\ref{thm:CFT properties}. 
We begin with a warm-up Section~\ref{subsec:connection_probabilities_disjoint}.

\subsection{Connection probabilities of disjoint multichordal $\SLE_4$ curves}\label{subsec:connection_probabilities_disjoint}

As a warm-up, let us consider the GFF with alternating boundary data.
Let $\Phi$ be a zero-boundary GFF on $\bH$, fix $-\infty = x_0 <x_1 < \cdots < x_{2N} < x_{2N+1} = \infty$, 
and let $u \colon \ol{\bH} \to \bR$ be the bounded harmonic function with piecewise constant boundary values that are even multiples of $\lambda \coloneqq \sqrt{\pi/8}$: 
\begin{align*}
u(x) = u_{\unnested}(x) \coloneqq 
\begin{cases}
0 , & x \in (x_{2j},x_{2j+1}), \ j \in \{0,\dots,N\}, \\
2\lambda ,  & x\in (x_{2j-1},x_{2j}), \ j \in \{1,\dots,N\} .
\end{cases}
\end{align*} 
Fix an $N$-link pattern $\alpha \in \LP_N$. 
In this case, the collection $S$ of frontiers of the first passage sets~\eqref{eq:cc_of_FPS_front} 
is described by an ordinary multichordal $\SLE_4$ process~\cite{Schramm-Sheffield:A_contour_line_of_the_continuum_GFF,Wang-Wu:Level_lines_of_Gaussian_free_field_I, Peltola-Wu:Global_and_local_multiple_SLEs_and_connection_probabilities_for_level_lines_of_GFF}. 
The probabilities of their connectivity patterns were solved in~\cite[Theorem~1.4]{Peltola-Wu:Global_and_local_multiple_SLEs_and_connection_probabilities_for_level_lines_of_GFF}:
\begin{align*} 
\PR_{\unnested} [\conn = \alpha] 
= \; & \frac{\PartF_{\alpha}(x_1, \ldots, x_{2N})}{\CobloF_{\unnested}(x_1, \ldots, x_{2N})} ,
\end{align*}
which is a special case of Equation~\eqref{eq:corss_proba_H} with the well-known partition function
\begin{align*}
\CobloF_{\unnested} (x_1,\ldots,x_{2N}) \coloneqq \prod_{1 \leq i < j \leq {2N}} (x_j - x_i)^{\frac{1}{2} (-1)^{j-i}} .
\end{align*}
More generally, given an $N$-link pattern $\beta \in \LP_N$, let the boundary data $u = u_\beta$ be chosen as in Equation~\eqref{eq:harmonic_function}, where the step directions 
\begin{align*}
\dir \beta_{j} \coloneqq \beta_{j} - \beta_{j-1} \in \{\pm 1\}, \qquad j \in \{1,\dots,2N\} ,
\end{align*}
encode the boundary jump directions as in~\cite[Sections~5-6]{Peltola-Wu:Global_and_local_multiple_SLEs_and_connection_probabilities_for_level_lines_of_GFF} and~\cite[Section~4]{Liu-Wu:Scaling_limits_of_crossing_probabilities_in_metric_graph_GFF}. 
The conformal block function associated to this boundary data is~\cite[Equation~(6.1)]{Peltola-Wu:Global_and_local_multiple_SLEs_and_connection_probabilities_for_level_lines_of_GFF}:
\begin{align*} 
\CobloF_\beta(x_1,\ldots,x_{2N}) \coloneqq \; & \prod_{1 \leq i < j \leq 2N} (x_j - x_i)^{\frac{1}{2} \vartheta(i,j)} 
, \qquad \beta \in \LP_N ,
\\
\vartheta(i,j) 
= \; &
\begin{cases} 
+1, \quad i,j \in \{a_1,a_2,\ldots ,a_{N} \}  \; \textnormal{or} \; i,j \in \{b_1,b_2,\ldots ,b_{N}\}, \\
-1, \quad \textnormal{otherwise}.
\end{cases}
\end{align*}
(See Equation~\eqref{eq:CobloF_nonval} for an alternative expression in terms of Specht polynomials~\cite{LPR:Fused_Specht_polynomials_and_c_equals_1_degenerate_conformal_blocks}.)
Similarly to the case of $\unnested$, the connection probabilities for $\beta \in \LP_N$ take the form of the ratio of a pure partition function and a conformal block function:
\begin{align}\label{eq:connection_probability_unfused}
\PR_\beta [\conn = \alpha] 
= \; & M_{\beta, \alpha} \, \frac{\PartF_{\alpha}(x_1, \ldots, x_{2N})}{\CobloF_\beta(x_1, \ldots, x_{2N})} 
\end{align}
where $M_{\beta, \alpha} = \one\{ \beta \KWleq \alpha \}$ 
(see \cite[Theorem~4.1]{Liu-Wu:Scaling_limits_of_crossing_probabilities_in_metric_graph_GFF} and~\cite[Section~2.3]{Karrila:Computation_of_pairing_probabilities_in_multiple-curve_models}).

\subsection{Connection probabilities of GFF level lines and FPS frontiers (Theorem~\ref{thm:corss_proba_H})} 
\label{subsec:convergence_connection_probabilities}

\begin{proof}[Proof of Theorem~\ref{thm:corss_proba_H}]
Fix $\beta \in \GLP_\multii$ and valences $\multii = (s_1,\ldots,s_p) \in \bZpos^p$. 
Consider the collection $S$ of frontiers of the first passage sets of the GFF $\Phi$ 
defined in Equation~\eqref{eq:cc_of_FPS_front}, with boundary data $u_\beta$ as in Equation~\eqref{eq:harmonic_function2}.
Fix $\alpha \in \PLP_\multii$, so that $M_{\alpha, \beta}=1$ (i.e., $\smash{\beta \KWleq \alpha}$).

We will use the sets $\PJ_m = \{x_1^m,\dots,x_{n_m}^m\}$ and $\NJ_m = \{ y_1^m,\dots,y_{n_m}^m\}$ illustrated in Figure~\ref{fig:ABnotation}.
We denote by $\eta^{a,m}$ the level line of $\Phi + u$ of height $(2m+1)\lambda$ emanating from $x_a^m$,
and by $\wt{\eta}^{b,m}$ the level line of $-(\Phi + u)$ of height $-(2m+1)\lambda$ emanating from $y_b^m$.
These curves form a random topological configuration that can be combinatorially encoded into a random valenced link pattern $\conn = \conn(\bs{\eta})$ in $\LP_\multii$, writing $\bs\eta = (\eta_1,\dots,\eta_N)$.

For $0 \leq m \leq \maxu-1$ (where $2 \maxu \lambda$ is the maximal value of $u_\beta$) and $1 \leq a,b \leq n_m$, 
we set 
\begin{align*}
\tau_\epsilon^{a,m} \coloneqq \; & \inf\big\{t \geq 0 \colon \dist(\eta^{a,m}(t),x_a^m) \geq \epsilon \big\} , \\
\wt{\tau}_\epsilon^{b,m} \coloneqq \; & \inf\big\{t \geq 0 \colon \dist(\wt{\eta}^{b,m}(t),y_b^m) \geq \epsilon \big\}.
\end{align*}
That is, $\eta^{a,m}[0,\tau_\epsilon^{a,m}]$ (resp.~$\wt{\eta}^{b,m}[0,\wt{\tau}_\epsilon^{b,m}]$) is the segment of $\eta^{a,m}$ (resp.~$\wt{\eta}^{b,m}$) grown until it first exits the $\epsilon$-neighborhood of its starting point.
Note that each $\eta^{a,m}[0,\tau_\epsilon^{a,m}]$ and $\wt{\eta}^{b,m}[0,\wt{\tau}_\epsilon^{b,m}]$ is a local set of $\Phi$. 
Let 
\begin{align*}
\Gamma_\epsilon \coloneqq \bigcup_{a,b,m} \big(\eta^{a,m}[0,\tau_\epsilon^{a,m}] \cup \wt{\eta}^{b,m}[0,\wt{\tau}_\epsilon^{b,m}]\big) .
\end{align*} 
Recall that $\sF_{\Gamma_\epsilon} = \sigma(\Gamma_\epsilon,\Phi_{\Gamma_\epsilon})$, and let $g_{\Gamma_\epsilon} \colon \bH \setminus \Gamma_\epsilon \to \bH$ be the unique conformal map
normalized so that $|g_{\Gamma_\epsilon}(z) - z| \to 0$ as $|z| \to \infty$. Then, we have
\begin{align*}
\PR_\beta \big[ \conn(\bs{\eta}) = \alpha \big] 
= \EX_\beta \big[ \one\{\conn(\bs{\eta}) = \alpha \} \big] 
= \EX_\beta \big[ \EX_\beta \big[ \one\{\conn(\bs{\eta}) = \alpha \} \giv \sF_{\Gamma_\epsilon} \big] \big].
\end{align*}
Moreover, the conditional law of the level lines/first passage set frontiers given $\Gamma_\epsilon$ is that of the corresponding level lines of the field in $\bH \setminus \Gamma_\epsilon$, i.e., where the curves start from distinct points. In this case, by conformally mapping onto $\bH$, we know the probabilities for the link patterns explicitly, by~\eqref{eq:connection_probability_unfused}. 
In particular, by conformal invariance, the following holds. Let 
\begin{align*}
\hat{\PJ}_m(\epsilon) = \; & \{ \hat{x}_1^m(\epsilon),\dots,\hat{x}_{n_m}^m(\epsilon)\} \coloneqq \PJ_m((u+h_{\Gamma_\epsilon}) \circ g_{\Gamma_\epsilon}^{-1}) , \\ 
\hat{\NJ}_m(\epsilon) = \; & \{ \hat{y}_1^m(\epsilon),\dots,\hat{y}_{n_m}^m(\epsilon) \} \coloneqq \NJ_m((u+h_{\Gamma_\epsilon}) \circ g_{\Gamma_\epsilon}^{-1}) , \qquad m \in \{ 0,\dots,\maxu-1 \} ,
\end{align*}
denote the sets of points illustrated in Figure~\ref{fig:ABnotation}, but instead of $u$, 
for the harmonic function $(u+h_{\Gamma_\epsilon}) \circ g_{\Gamma_\epsilon}^{-1}$.
Then, $\hat{x}_a^m(\epsilon) = g_{\Gamma_\epsilon}(\eta^{a,m}(\tau_\epsilon^{a,m}))$ and $\hat{y}_b^m(\epsilon) = g_{\Gamma_\epsilon}(\wt{\eta}(\wt{\tau}_\epsilon^{b,m}))$, and letting $\xi_1^\epsilon,\dots,\xi_{2N}^\epsilon$ denote the left-to-right ordering of the points $(\hat{x}_a^m(\epsilon),\hat{y}_b^m(\epsilon))_{a,b,m}$, we have 
\begin{align*}
\EX_\beta \big[ \one\{\conn(\bs\eta) = \alpha \} \giv \sF_{\Gamma_\epsilon} \big] 
= \; & \PR_{\imath(\beta)} \big[\conn(\hat{\eta}_1^\epsilon,\dots, \hat{\eta}_{2N}^\epsilon) = \imath(\alpha) \big] 
&& \textnormal{[by~\eqref{eq:connection_probability_unfused}]}
\\
= \; &  \frac{\PartF_{\imath(\alpha)}(\xi_1^\epsilon,\dots,\xi_{2N}^\epsilon)}{\CobloF_{\imath(\beta)}(\xi_1^\epsilon,\dots,\xi_{2N}^\epsilon)},
&& \textnormal{[since $\smash{\beta \KWleq \alpha}$]}
\end{align*}
where $\hat{\eta}_j^\epsilon$ is the image in $\bH$ of $\eta_j$ under $g_{\Gamma_\epsilon}$ for $j \in \{1,\dots,N\}$, 
and $\imath \colon \LP_\multii \to \LP_N$ is the \quote{unfusing} map~\eqref{eq:iotamap}. 
Next, the bounded convergence theorem yields
\begin{align*}
\PR_\beta \big[ \conn(\bs{\eta}) = \alpha \big]  = \EX_\beta \Bigg[ \lim_{\epsilon \downarrow 0} \frac{\PartF_{\imath(\alpha)}(\xi_1^\epsilon,\dots,\xi_{2N}^\epsilon)}{\CobloF_{\imath(\beta)}(\xi_1^\epsilon,\dots,\xi_{2N}^\epsilon)} \Bigg].
\end{align*}
A standard harmonic measure argument then shows that, as $\epsilon \to 0$, 
all the desired groups of points tend together:
$\xi^\epsilon_{\summ_{j-1} + 1},\ldots,\xi^{\epsilon}_{\summ_{j}} \to x_j$ for all $1 \leq j \leq p $. 
Thus, by Lemmas~\ref{lem:VF_fusion}~\&~\ref{lem:PartF_fusion}, 
\begin{align*}
 \lim_{\epsilon \downarrow 0} \frac{\PartF_{\imath(\alpha)}(\xi_1^\epsilon,\dots,\xi_{2N}^\epsilon)}{\CobloF_{\imath(\beta)}(\xi_1^\epsilon,\dots,\xi_{2N}^\epsilon)}  = \frac{\PartF_{\alpha}(x_1, \ldots, x_p)}{\CobloF_{\beta}(x_1, \ldots, x_p)} 
\end{align*}
which yields the asserted identity~\eqref{eq:corss_proba_H}. 
\end{proof}

\subsection{GFF partition functions and $c=1$ conformal blocks}\label{subsec:partition_functions_conformal_blocks}

Recall that the \emph{partition function} of the GFF with boundary data $u_\beta$ is defined as~\cite{Dubedat:SLE_and_free_field} 
\begin{align*} 
\CobloF^{\GFF}_\beta \coloneqq \frac{ e^{- \frac{1}{2} \, ( u_\beta,u_\beta )_{\nabla}} }{(\det_\zeta (-\Delta))^{1/2}} ,
\end{align*}
where $\det_\zeta (-\Delta)$ is the $\zeta$-regularized determinant of the positive Dirichlet Laplacian~\cite{OPS:Extremals_of_determinants_of_Laplacians} (which we will not need in the present work) 
and the regularized Dirichlet norm
\begin{align*} 
( u_\beta,u_\beta )_{\nabla} \coloneqq \; \DirReg{u_\beta} ,
\end{align*}
of the mean of the GFF depends on the Dyck path $\beta$.
We now show that the norm in Equation~\eqref{eq: regularized Dirichlet norm} with a suitable regularization indicated as \quote{$\bs{\colon} \!\! \cdot \! \bs{\colon}$} 
is finite and yields a particular $\SLE_4$ partition function, namely the \quote{conformal block function} $\CobloF_\beta$
defined in Equation~\eqref{eq:CobloF_ms} (introduced in~\cite{LPR:Fused_Specht_polynomials_and_c_equals_1_degenerate_conformal_blocks}) 
normalized by the square root of the $\zeta$-regularized determinant of the Dirichlet Laplacian. 

\begin{proposition}\label{prop:partition_functions_agree}
Fix valences $\multii \in \bZpos^p$ and a valenced link pattern $\beta \in \GLP_\multii$. We have 
\begin{align*}
\CobloF_\beta^{\GFF}(x_1,\dots,x_p) = \frac{\CobloF_\beta(x_1,\dots,x_p)}{(\det_\zeta(-\Delta))^{1/2}}.
\end{align*}
\end{proposition}

\begin{proof}
The proof boils down to computing
\begin{align*}
%( u_\beta,u_\beta )_{\nabla} \coloneqq 
\DirReg{u_\beta} 
= \lim_{\epsilon \to 0} \bigg( \int_{\bH \setminus ( B(x_1,\epsilon) \cup \cdots \cup B(x_{p},\epsilon))} | \nabla u_\beta(z) |^2 \, \ud z + \sum_{j=1}^{p} \frac{(2\lambda \dir \beta_j)^2}{\pi} \log \epsilon \bigg),
\end{align*}
where $2\lambda \dir \beta_j$ is the jump in the boundary values of $u_\beta$ at $x_j$.  
In particular, according to Equation~\eqref{eq:partition_function_dubedat}, 
we have to prove that $( u_\beta,u_\beta )_{\nabla} \coloneqq  \DirReg{u_\beta}  = -2\log \CobloF_\beta(x_1,\dots,x_p)$.

Since $u_\beta$ is a bounded harmonic function on $\bH$ with piecewise constant boundary data, 
we can conveniently write it as the following sum of the bounded harmonic functions: 
\begin{align*}
u_\beta(z) = \sum_{j=1}^p v_j(z) , \quad
\textnormal{where} \quad
v_j(z) \coloneqq 2 \lambda \dir \beta_j  \Big(  1 - \tfrac{1}{\pi} \, \im (\log (z-x_j)) \Big).
\end{align*}
Let us consider the integral
\begin{align*}
I_\epsilon \coloneqq \; & \int_{\bH \setminus ( B(x_1,\epsilon) \cup \cdots \cup B(x_{p},\epsilon))} | \nabla u_\beta(z) |^2 \, \ud z 
\\
= \; &
\sum_{j=1}^p \int_{\bH \setminus ( B(x_1,\epsilon) \cup \cdots \cup B(x_{p},\epsilon))} \nabla u_\beta(z) \cdot \nabla v_j (z) \, \ud z \\
= \; &
\sum_{j=1}^p \int_{\bH \setminus  B(x_j,\epsilon)} \nabla u_\beta(z) \cdot \nabla v_j (z) \, \ud z + o(1) , \qquad \epsilon \to 0.
\end{align*}
Using the identity $\nabla \cdot (u_\beta(z) \nabla v_j (z)) = \nabla  u_\beta(z) \cdot \nabla v_j (z) + u_\beta(z) \Delta v_j  (z) = \nabla  u_\beta(z) \cdot \nabla v_j (z)$ and the divergence theorem\footnote{Strictly speaking, the divergence theorem addresses finite domains; integrals over suitable approximating domains can be seen to converge, e.g., by Equation~\eqref{eq:gradient computation}.}, 
we find that
\begin{align*}
I_\epsilon = \; & \sum_{j=1}^p \int_{\partial (\bH \setminus  B(x_j,\epsilon) ) } u_\beta(z) (\nabla v_j (z) \cdot \ownvec_z) \, |\ud z| + o(1), \qquad \epsilon \to 0 ,
\end{align*}
where $\ownvec_z$ is the outward normal vector and $|\ud z|$ the length measure. 
A computation yields
\begin{align*}
\nabla v_j (z) = \frac{2 \lambda \dir \beta_j}{\pi}\bigg( \frac{y}{(x-x_j)^2 + y^2},  -\frac{x-x_j}{(x-x_j)^2 + y^2}\bigg), \quad \textnormal{for} \quad z = x + \ii y.   
\end{align*}
In particular, we note that 
\begin{align}
\label{eq:gradient computation}
\nabla v_j(z) \cdot \ownvec_z = 
\begin{cases}
\displaystyle \frac{2 \lambda \dir \beta_j}{\pi} \frac{1}{x-x_j} , & x \in (-\infty,x_j-\epsilon) \cup (x_j+\epsilon,\infty) \\
0 , & x \in \bH \cap \partial B(x_j,\epsilon), \textnormal{for any } \epsilon>0 .
\end{cases}
\end{align}
Moreover, since $\dir \beta_k = \beta_k-\beta_{k-1}$ and $u_\beta(x) = 2 \lambda \beta_k$ for $x \in (x_k, x_{k+1})$, we have 
\begin{align*}
I_\epsilon = \; &  \sum_{j=1}^p \frac{ 4 \lambda^2 \dir \beta_j}{\pi} \bigg[ \sum_{k=1}^{j-2} \beta_{k}\int_{x_k }^{x_{k+1}} \frac{\ud x}{x-x_j} + 
\beta_{j-1} \int_{x_{j-1} }^{x_{j}- \epsilon } \frac{\ud x}{x-x_j} \\
&\qquad \qquad \qquad \qquad + 
\beta_j \int_{x_{j} + \epsilon }^{x_{j+1} } \frac{\ud x}{x-x_j} + \sum_{k=j+1}^{p-1} \beta_k \int_{x_k }^{x_{k+1}} \frac{\ud x}{x-x_j} \bigg] + o(1) \\
= \; &  \sum_{j=1}^p \frac{ 4 \lambda^2 \dir \beta_j}{\pi} \bigg( \sum_{k=1}^{j-2} \beta_k 
\log \Big| \frac{x_{k+1}-x_j}{x_{k}-x_j} \Big| + \beta_{j-1} \log \Big| \frac{\epsilon}{x_{j-1}-x_j} \Big|
\\
&\qquad \qquad \qquad \qquad+ 
\beta_j \log \Big| \frac{x_{j+1}-x_j}{\epsilon} \Big|
+ \sum_{k=j+1}^{p-1} \beta_k 
\log \Big| \frac{x_{k+1}-x_j}{x_{k}-x_j} \Big| \bigg) + o(1) \\
= \; & - \sum_{j=1}^p \frac{ 4 \lambda^2 \dir \beta_j}{\pi} \bigg( 
\sum_{k=1}^{j-1} \dir \beta_k \log |x_{k}-x_j |
+ \dir \beta_j  \log \epsilon + \sum_{k=j+1}^{p} \dir \beta_k  \log |x_{k}-x_j | \bigg) + o(1) \\
= \; &  -\sum_{j=1}^p \frac{ 4 \lambda^2 (\dir \beta_j)^2}{\pi} \log \epsilon 
- \sum_{j=1}^p \frac{ 4 \lambda^2 \dir \beta_j}{\pi} \bigg( \sum_{k \neq j} \dir \beta_k \log |x_{k}-x_j | \bigg) + o(1) , \qquad \textnormal{as} \quad \epsilon \to 0 .
\end{align*}
Thus, since the logarithmic blowup in $I_\epsilon$ is canceled by the regularization in the Dirichlet norm $\DirReg{u_\beta}$, we find (using the expression in Equation~\eqref{eq:CobloF_ms} for $\CobloF_\beta$) that
\begin{align*}
( u_\beta,u_\beta )_{\nabla} = \; & \DirReg{u} 
= \;  \lim_{\epsilon \to 0} \bigg( I_\epsilon + \sum_{j=1}^{p} \frac{(2 \lambda \dir \beta_j)^2}{\pi} \log \epsilon \bigg) \\
= \; &  - \sum_{j=1}^p \frac{ 4 \lambda^2 \dir \beta_j}{\pi} \bigg( \sum_{k \neq j} \dir \beta_k \log |x_{k}-x_j | \bigg) \\
= \; &  - \log \prod_{1 \leq j < k \leq p} (x_k-x_j)^{\dir \beta_j \dir \beta_k} 
= \; - 2 \log \mathcal{U}_\beta (x_1,\dots,x_p),
\end{align*}
as was to be shown.
\end{proof}

\begin{remark}
Let us point out that if we use the other common normalization so that the Dirichlet Green's function blows up like $-\log|z-w|$ along the diagonal instead, 
then we should define $( u_\beta,u_\beta )_{\nabla} \coloneqq \tfrac{1}{2\pi} \DirReg{u}$ and $\lambda = \pi/2$. Then, we have
\begin{align*}
\CobloF_\beta^{\GFF} \coloneqq \frac{ e^{- \pi ( u_\beta,u_\beta )_{\nabla}} }{(\det_\zeta (-\Delta))^{1/2}} 
\qquad \textnormal{and} \qquad 
\CobloF_\beta^{\GFF}(x_1,\dots,x_p) = \frac{\big( \CobloF_\beta(x_1,\dots,x_p) \big)^{2\pi}}{(\det_\zeta(-\Delta))^{1/2}} ,
\end{align*}
which is the appropriate analogue of the statement in Proposition~\ref{prop:partition_functions_agree}.
\end{remark}

\subsection{CFT interpretation (Theorem~\ref{thm:CFT properties})}
\label{subsec:properties_pure_partition_functions}

\begin{proof}[Proof of Theorem~\ref{thm:CFT properties}]
The BPZ PDEs~\eqref{eq: BPZ PDE at kappa equals 4} and the M\"obius covariance~\eqref{eq: COV general at kappa equals 2}  
follow immediately from linearity, the fact that each $\PartF_\alpha$ is defined in Equation~\eqref{eq:PartF}
as a linear combination of the conformal block functions $\CobloF_\beta$,
and the fact that by~\cite[Theorem~3.24]{LPR:Fused_Specht_polynomials_and_c_equals_1_degenerate_conformal_blocks},
the latter solve the BPZ PDEs~\eqref{eq: BPZ PDE at kappa equals 4} and satisfy the M\"obius covariance~\eqref{eq: COV general at kappa equals 2}.
In light of these two properties, the partition functions $\PartF_\alpha$ could be thought of as boundary correlation functions of primary fields as in Equation~\eqref{eq:fusion_corrf} in a CFT with central charge $c=1$. 

By Theorem~\ref{thm:corss_proba_H}, in order to prove that $\PartF_\alpha(x_1,\ldots ,x_p)>0$ for all $x_1 < \cdots < x_p$,
i.e., that \textnormal{(POS)} holds, 
it suffices to show that $\PR_\beta[\conn = \alpha] > 0$ whenever 
$\smash{ \beta \KWleq \alpha }$.
On the combinatorial side, recall from Lemma~\ref{lem:KWleq importance} that $\smash{ \beta \KWleq \alpha }$ equivalently means that $\alpha$ is a combinatorial level-line pattern for the boundary condition $\beta$ (Figure~\ref{fig:combinatorial bijections}, right panel). 
On the other hand, Remark~\ref{rem:KWleq importance} gives a simple \quote{iterated curve growth description} of the collection of such $\alpha \in \LP_\multii$. With this description, and GFF level lines being local sets, it suffices to prove that for any two boundary condition jump points $x_i$ and $x_j$ with $u(x_i^-)<(2m+1)\lambda < u(x_i^+)$ and $u(x_j^-)>(2m+1)\lambda > u(x_j^+)$, the level line of height $(2m+1)\lambda $ from $x_i$ terminates at $x_j$ with a positive probability. 
Indeed, with this property at hand, growing the level lines iteratively, by Remark~\ref{rem:KWleq importance}, any $\alpha \in \LP_\multii$ with  $\smash{ \beta \KWleq \alpha }$ will then appear as $\conn$ with a positive probability.
In turn, the desired statement about level line termination points is a direct consequence of~\cite[Lemma~2.5]{Miller-Wu:Intersections_of_SLE_paths:_the_double_and_cut_point_dimension_of_SLE} 
(that proof uses the imaginary geometry coupling for $\kappa \neq 4$, the same proof works for level lines).

For the linear independence \textnormal{(LIN)}, 
on the one hand, the conformal block functions $\{ \CobloF_\beta \colon \beta \in \LP_\multii \}$ are linearly independent by~\cite[Proposition~3.17]{LPR:Fused_Specht_polynomials_and_c_equals_1_degenerate_conformal_blocks}. 
On the other hand, Equation~\eqref{eq:PartF} can be written as the following system of  $|\LP_\multii|$ linear equations:
\begin{align*}
\PartF_\alpha = \sum_{\beta \in \LP_\multii} \mathfrak{M}_{\alpha, \beta} \, \CobloF_\beta, \qquad \textnormal{for all } \alpha \in \LP_\multii, 
\end{align*}
where $\mathfrak{M}_{\alpha, \beta} = 0$ unless $\alpha \DPleq \beta$ and $\mathfrak{M}_{\alpha, \alpha} = \# \LP_{\geq_\multii \alpha} \geq 1$. 
Hence, upon ordering the elements of $\LP_\multii$ in an increasing order with respect to the partial order $\DPleq$, 
the matrix $\mathfrak{M}_{\alpha, \beta}$ becomes upper-triangular with non-zero diagonal elements. It thus follows that $\{ \PartF_\alpha \colon \alpha \in \LP_\multii \}$ and $\{ \CobloF_\beta \colon \beta \in \LP_\multii \}$ span, and are bases for, the same $|\LP_\multii|$-dimensional function space.
\end{proof}

\bigskip

\section{Global multiple $\SLE_4$: existence and uniqueness}
\label{sec:uniqueness}

We define a (topological) \emph{polygon} to be a $(p+1)$-tuple $(\domain;x_1,\dots,x_p)$ consisting of a simply connected domain $\domain$ and distinct boundary points $x_1,\dots,x_p \in \partial \domain$, $p \geq 2$, 
ordered in a counterclockwise manner on $\partial \domain$.
For $p=2$, we let $X_0(\domain;x_1,x_2)$ denote the set of continuous unparametrized curves $\eta$ in $\domain$ connecting $x_1$ and $x_2$ which intersect the boundary $\partial \domain$ only at $\{x_1,x_2\}$. 
For a general $\multii$-valenced link pattern $\alpha \in \LP_\multii$ as in~(\ref{eq: definition of N},~\ref{eq:alpha}), 
we let $X_0^\alpha(\domain;x_1,\dots,x_p)$ be the collection of families $\bs\eta = (\eta_1,\dots,\eta_N)$ of curves which are pairwise disjoint in $\domain$ such that $\eta_j \in X_0(\domain;x_{a_j},x_{b_j})$ for each $j \in \{1,\dots,N\}$.  
The curves are allowed to intersect at $\partial \domain$ but not in $\domain$, though one could generalize the definition further if necessary (for example, when considering $\SLE_\kappa$ curves with $\kappa > 4$).

\begin{definition}\label{def:fused_multiple_sle}
For any $\multii$-valenced link pattern $\alpha \in \LP_\multii$, we call a probability measure on 
$(\eta_1, \ldots, \eta_N)\in X_0^{\alpha}(\domain; x_1, \ldots, x_p)$ 
\emph{a~global $\aSLE_\kappa$} if, for each $j\in\{1, \ldots, N\}$, 
the conditional law of the curve $\eta_j$, given the curves $\{\eta_1, \ldots, \eta_{j-1}, \eta_{j+1}, \ldots, \eta_N\}$, 
is that of a chordal $\SLE_\kappa$ curve connecting $x_{a_j}$ and $x_{b_j}$
in the connected component of the domain $\domain \setminus \{\eta_1, \ldots, \eta_{j-1}, \eta_{j+1}, \ldots, \eta_N\}$
containing the endpoints $x_{a_j},x_{b_j}$ of $\eta_j$ on its boundary.
\end{definition}

The description of the conditional law in Definition~\ref{def:fused_multiple_sle} is conventionally termed the \emph{resampling property}. 
Using the frontier of the first passage sets of the GFF, we obtain a construction for $\kappa=4$ and $\alpha \in \PLP_\multii$.
The resampling property arises naturally from the domain Markov property of the GFF.
For general $\kappa \in (0,8)$, the existence and uniqueness of global $\aSLE_\kappa$ is well known for non-valenced link patterns $\alpha \in \LP_N$ with $\multii = (1)^{2N}$ 
(i.e., having $2N$ entries $1$)~\cite{Kozdron-Lawler:Configurational_measure_on_mutually_avoiding_SLEs, Miller-Sheffield:Imaginary_geometry2, Peltola-Wu:Global_and_local_multiple_SLEs_and_connection_probabilities_for_level_lines_of_GFF, BPW:On_the_uniqueness_of_global_multiple_SLEs, FLPW:Multiple_SLEs_Coulomb_gas_integrals_and_pure_partition_functions, Sun-Yu:SLE_partition_functions_via_conformal_welding_of_random_surfaces, Zhan:Existence_and_uniqueness_of_nonsimple_multiple_SLE, AMY:Multiple_SLE_from_CLE, AHSY:Conformal_welding_of_quantum_disks_and_multiple_SLE_the_non-simple_case}.
The main goal of this section is to prove a more general result for $\kappa=4$.

\begin{theorem}\label{thm:existence_fused_multiple_sle}
In each polygon, there exists a unique global $\aSLE_4$ for each $\alpha \in \PLP_\multii$.
\end{theorem}
We prove existence in Section~\ref{subsec:resampling_property} and uniqueness in Section~\ref{subsec:uniqueness_proof}.
The defining resampling property is obtained in Theorem~\ref{thm:resampling_property}. 
By conformal invariance of $\SLE$ and GFF, it suffices to consider the case where the polygon is the upper half-plane $\domain = \bH$, and $x_1 <\cdots < x_p$.

While our methods only work in the case of $\alpha \in \PLP_\multii$, we believe that the result holds for general $\alpha \in \LP_\multii$ and $\kappa \in (0,8)$, and state it as a conjecture for potential later investigations.
For example, considering limits of the known cases from~\cite{Kozdron-Lawler:Configurational_measure_on_mutually_avoiding_SLEs, Peltola-Wu:Global_and_local_multiple_SLEs_and_connection_probabilities_for_level_lines_of_GFF, Zhan:Existence_and_uniqueness_of_nonsimple_multiple_SLE, FLPW:Multiple_SLEs_Coulomb_gas_integrals_and_pure_partition_functions}, 
or refining the Liouville quantum gravity (LQG) techniques used in~\cite{Sun-Yu:SLE_partition_functions_via_conformal_welding_of_random_surfaces, AHSY:Conformal_welding_of_quantum_disks_and_multiple_SLE_the_non-simple_case}
give possible routes to verify Conjecture~\ref{conj:existence_fused_multiple_sle}. We will not attempt this in the present work.

\begin{conjecture}\label{conj:existence_fused_multiple_sle}
In each polygon, there exists a unique global $\aSLE_\kappa$ for each $\multii$-valenced link pattern $\alpha \in \LP_\multii$ and for each $\kappa \in (0,8)$.
\end{conjecture}

Now, let us consider the set of level lines of $\Psi = \Phi + u_\beta$ given in Equation~\eqref{eq:OML}, 
\begin{align*}
\OML_\beta = \OML_\beta(\Phi) \coloneqq \big\{ \eta \;|\; \eta \textnormal{ is a level line of } \Phi + u_\beta \textnormal{ of height } (2k+1)\lambda \textnormal{ for some } k \in \bZ \big\} ,
\end{align*} 
whose heights are odd integer multiples of $\lambda$.
By Lemma~\ref{lem:frontiers_are_level_lines}, the collection $\bs \eta$ of frontiers of the first passage sets $\bA_0^{u_\beta},\ldots,\bA_{2(\maxu-1)\lambda}^{u_\beta}$ with $\maxu = \maxu_\beta$ is exactly the set $\OML_\beta = \{\bs \eta\}$,
and they form a random topological configuration that can be combinatorially encoded into a random valenced link pattern $\conn = \conn(\bs{\eta}) = \conn(\OML_\beta)$ in $\PLP_\multii$ (cf.~Theorem~\ref{thm:corss_proba_H}).

We begin in Section~\ref{subsec:resampling_property} by showing that the family of curves comprising $\OML_\beta$ satisfies the resampling property. From this, we deduce the existence part of Theorem~\ref{thm:existence_fused_multiple_sle}. Then, in Section~\ref{subsec:uniqueness_proof} we show that any collection of curves such that $\conn(\bs \eta) = \alpha \in \PLP_\multii$  
satisfying the resampling property can be coupled as the level lines $\OML_\beta$ of $\Phi+u_\beta$, 
for any $\smash{\beta \KWleq \alpha}$. From this, we can then prove the uniqueness part of Theorem~\ref{thm:existence_fused_multiple_sle} and conclude the result. 
As a byproduct of the coupling proof, we establish that the conditional law of $\OML_\beta$ given the connection event $\{\conn = \alpha\}$ (for some $\alpha \in \PLP_\multii(\beta)$) is independent~of~$\beta$. 

In Section~\ref{subsec:local_description_unconditional_case}, 
we will give a local description of a single level line $\eta \in \OML_\beta$ (marginal law), 
namely as an $\SLE_4$ curve with partition function given by the conformal block function $\CobloF_\beta$ defined in Equation~\eqref{eq:CobloF_ms}. 
Lastly, in Proposition~\ref{prop:alpha-conditional marginal law of one level line} in Section~\ref{subsec:local_description_conditional_case} we shall 
describe the conditional law of one curve given the configuration $\conn$, which might also be of interest.

\subsection{Resampling property and existence of global multiple $\SLE_4$}\label{subsec:resampling_property}

We begin by proving that $\{\bs\eta\}=\{\eta_1,\ldots,\eta_N\} = \OML_\beta$ 
satisfy the resampling property which, together with the positivity of the probability of each link pattern, implies the existence of global multiple $\SLE_4$ for $\alpha \in \PLP_\multii$. 
For each $j \in \{1,\ldots,N\}$, we let $D_j$ denote the domain of $\bH \setminus \bigcup_{k \neq j} \eta_k$ which contains $\eta_j$. 
Furthermore, we let $\pp_j$ (resp.~$\np_j$) denote the starting point (resp.~endpoint) of $\eta_j$. 
(Note that $\np_j$ is random, but determined by $\eta_1,\ldots,\eta_{j-1},\eta_{j+1},\ldots,\eta_N$.)

\begin{theorem}[Resampling property]\label{thm:resampling_property}
The family $\{\bs\eta\} = \OML_\beta$ satisfies the property:
for each $j$, the conditional law of $\eta_j$, given $\eta_1,\ldots,\eta_{j-1},\eta_{j+1},\ldots,\eta_N$, is that of a chordal $\SLE_4$ in $(D_j;\pp_j,\np_j)$.
\end{theorem}

\begin{figure}[ht!]
\centering
\includegraphics[width=.6\textwidth]{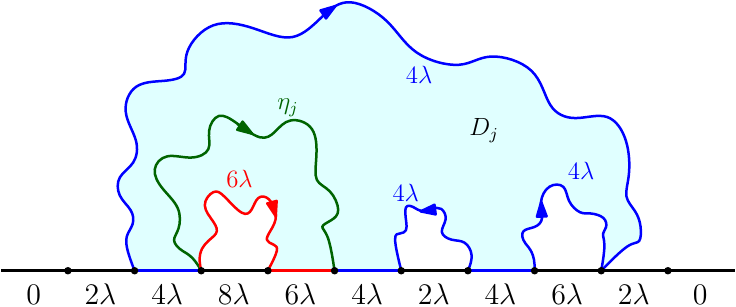}
\caption{Illustration of the proof of the resampling property, Theorem~\ref{thm:resampling_property}. 
In the figure, we explore the domain $D_j$ (colored in light blue) of the curve $\eta_j$. We begin by exploring the level lines of height $(2m-3)\lambda$ (in blue) and then the level lines of height $(2m+1)\lambda$ (in red). 
Then, we further explore the level lines of height $(2m-1)\lambda$ which are not $\eta_j$ (if there are any) to obtain $D_j$. 
Then, it is clear from the boundary data of the GFF that the conditional law of $\eta_j$ given the other curves is that of an $\SLE_4$ in $D_j$.
\label{fig:proof_resampling_property}}
\end{figure}

\begin{proof}
Fix $\pp_j \in \PJ_m$ for some $m \in \{1,\ldots,\maxu\}$, where $\maxu = \tfrac{1}{2\pi} \max u_\beta$, and let $\eta_j$ be the level line of height $2(m-1)\lambda$ starting from $\pp_j$.
Recall from Lemma~\ref{lem:frontiers_are_level_lines} that for each $k \in \{1,\dots,\maxu\}$, every level line of height $(2k-1)\lambda$ grows from a point in $\PJ_k$ and terminates at a point in $\NJ_k$, 
and the boundary data of GFF on the left (resp.~right) of the curve is $2(k-1)\lambda$ (resp.~$2k\lambda$). 
Moreover, recall from Remark~\ref{rmk:conditional_law_level_lines} that the conditional law of a level line given a union $A$ of other level lines 
is that of a level line of the field $\Phi^A + h_A + u_\beta$. Thus, if we sample the level line $\eta^{a,m-1}$ of height $(2m-3)\lambda$ from each point $x_a^{m-1} \in \PJ_{m-1}$, then $\bigcup_{a = 1}^{n_{m-1}} \eta^{a,m-1}$ cuts $\bH$ into regions where the boundary data of the GFF is at most $2(m-2)\lambda$ (regions which lie to the left of the curves) and regions where the boundary data is at least $2(m-1)\lambda$ (regions which lie to the right of the curves). Similarly, if we in addition sample the level lines $\eta^{a,m+1}$ of height $(2m+1)\lambda$ from the points $\PJ_{m+1}$, then the resulting collection of curves 
$\big(\bigcup_{a = 1}^{n_{m-1}} \eta^{a,m-1}\big) \cup \big(\bigcup_{a = 1}^{n_{m+1}} \eta^{a,m+1}\big)$ cuts $\bH$ into domains of three types: 
\begin{enumerate}[leftmargin=*, label=(\roman*)]
\item domains where the boundary data of the GFF is at most $2(m-2)\lambda$ (domains which lie to the left of the curves $\eta^{a,m-1}$, $a \in \{ 1,\dots,n_{m-1}\}$);
\item domains where the boundary data of the GFF is at least $2(m+1)\lambda$ (domains which lie to the right of the curves $\eta^{a,m+1}$, $a \in \{ 1,\dots,n_{m+1}\}$); and 
\item\label{it:domain_type_iii} domains where the boundary data of the GFF belongs to $\{ 2(m-1)\lambda,2m\lambda \}$ (domains which lie to the right of curves $(\eta^{a,m-1})$ 
and to the left of the curves $(\eta^{a,m+1})$.
\end{enumerate}
Then, any level line of height $2(m-1)\lambda$ is contained in a domain of type~\ref{it:domain_type_iii}. We let $\wt{D}_j$ be the connected component of 
$\bH \setminus \big(\big(\bigcup_{a = 1}^{n_{m-1}} \eta^{a,m-1}\big) \cup \big(\bigcup_{a = 1}^{n_{m+1}} \eta^{a,m+1}\big)\big)$ 
which contains $\eta_j$. Next, we sample the level lines of height $2(m-1)\lambda$ from the points $(\PJ_m \cap \partial \wt{D}_j) \setminus \{ \pp_j \}$. 
Upon doing so, we obtain the domain $D_j$ where the boundary data only has two points of discontinuities, $\pp_j$ and $\np_j$, and such that on the clockwise (resp.~counterclockwise) arc from $\pp_j$ to $\np_j$, the boundary data of the GFF is $2(m-1)\lambda$ (resp.~$2m\lambda$). Consequently, the conditional law of $\eta_j$ given the other level lines $\eta_1,\ldots,\eta_{j-1},\eta_{j+1},\ldots,\eta_N$ is that of a chordal $\SLE_4$ in $D_j$ from $\pp_j$ to $\np_j$.
This is what we sought to prove. 
\end{proof}

We can now conclude the existence part of the proof of Theorem~\ref{thm:existence_fused_multiple_sle}.

\begin{proof}[Proof of Theorem~\ref{thm:existence_fused_multiple_sle}: Existence.]
By Theorem~\ref{thm:resampling_property}, the collection $\{ \bs{\eta} \} = \OML_\beta$ satisfies the resampling property; the defining property of global $\SLE_4$. 
Thus, the proof of existence of global $\aSLE_4$ for $\alpha \in \PLP_\multii$ is reduced to showing that $\conn(\bs{\eta}) = \alpha$ with positive probability, whenever $\smash{\beta \KWleq \alpha}$. 
This is a consequence of the proof of Theorem~\ref{thm:CFT properties}~(POS).
\end{proof}

\subsection{Uniqueness of global multiple $\SLE_4$}
\label{subsec:uniqueness_proof}

\begin{proposition}\label{prop:coupling_fused_multiple_sle}
Fix $\alpha \in \PLP_\multii$. 
Suppose a probability measure on curves $\bs \eta \in X_0^{\alpha}(\bH; x_1, \ldots, x_p)$ satisfies the resampling property, i.e., 
it is a~global $\aSLE_4$ in the sense of Definition~\ref{def:fused_multiple_sle}.
Then, there exists a coupling $(\Phi,\bs \eta)$ of the curves $\bs \eta$ with a GFF $\Phi$ so that $\{\bs \eta\}$ has the same law as the collection $\OML_\beta(\Phi)$ 
of level lines of $\Phi + u_\beta$, conditional on the event $\{ \conn = \alpha \}$,
for any boundary data $\beta \in \GLP_\multii$ such that $\alpha \in \PLP_\multii(\beta)$.
\end{proposition}

We note that it is a priori not completely obvious that Proposition~\ref{prop:coupling_fused_multiple_sle} holds. 
The reason for this is that, while we constructed $\aSLE_4$ by virtue of a coupling of $\SLE_4$-type curves with a GFF, 
we do not a priori know that there is no other probability measure on curves satisfying the resampling property, which does not naturally couple with a GFF.  
The most important ingredient in the proof is that for a given $\aSLE_4$ with $\alpha \in \PLP_\multii$, 
we know that $\imath(\alpha) \in \LP_N$ and $\iaSLE_4$ is a global multiple $\SLE_4$ with link pattern $\imath(\alpha)$. This law is unique, by \cite[Theorem~1.2]{BPW:On_the_uniqueness_of_global_multiple_SLEs}, and we constructed it as $\OML_{\imath(\beta)}(\Phi)$. 
Thus, we can construct a field $\Psi$ which, upon mapping out initial and terminal segments of the $\aSLE_4$ curves of small diameter $\epsilon > 0$ (so that we get a family of $\iaSLE_4$ curves) is coupled with $\iaSLE_4$. From this, we can then deduce the coupling by sending $\epsilon \to 0$.

\begin{proof}[Proof of Proposition~\ref{prop:coupling_fused_multiple_sle}]
Fix $\alpha \in \PLP_\multii$. 
We enumerate the curves $\bs \eta = (\eta_1,\ldots,\eta_N)$ from left to right: 
if $a_j \coloneqq \min(\eta_j \cap \bR)$, 
then $a_1 \leq \cdots \leq a_N$ and if $a_j = a_{j+1}$, 
then $\eta_{j+1}$ lies in the bounded connected component of $\bH \setminus \eta_j$.  
Write $b_j \coloneqq \max(\eta_j \cap \bR)$. For each $j \in \{1,\ldots,N\}$, we let $D_j$ denote the connected component of $\bH \setminus \bigcup_{k \neq j} \eta_k$ which contains $\eta_j$ and we let $D_j^-$ (resp.~$D_j^+$) denote the intersection of $D_j$ with the unbounded (resp.~bounded) connected component of $\bH \setminus \eta_j$.  
Since the family $\bs \eta$ satisfies the resampling property, the conditional law of $\eta_j$ given $(\eta_k)_{k \neq j}$ is that of an $\SLE_4$ in $D_j$ from $a_j$ to $b_j$ 
(note here that, by the reversibility of $\SLE_4$, the time-reversed curve is an $\SLE_4$ from $b_j$ to $a_j$). Finally, we denote by $(O_j)$ the collection of connected components of $\bH \setminus \bigcup_{j=1}^N \eta_j$ and note that the unbounded component $O_1$ coincides with $D_j^-$ for at least one $j \in \{1,\dots,N\}$. Any other domain $O_k$, with $k \neq 1$, 
coincides with $D_j^+$ for exactly one $j$, and can coincide with $D_\ell^-$ for one or more $\ell$. 
We shall construct a field by assigning to each domain $O_k$ a zero-boundary GFF plus some constant, but before writing down the choice, let us define some quantities of interest.

Similarly to the proof of Theorem~\ref{thm:corss_proba_H}, for all $j \in \{1,\ldots,N\}$, we define 
\begin{align*}
\; &
\begin{cases}
\tau_j^\epsilon \coloneqq  \inf \{t \geq 0 \colon \eta_j(t) \in \partial B(a_j,\epsilon) \} , \\
\sigma_j^\epsilon \coloneqq  \sup \{t \geq 0 \colon \eta_j(t) \in \partial B(b_j,\epsilon) \}, 
\end{cases}
\qquad 
\textnormal{for all $j \in \{1,\ldots,N\}$;} \qquad 
\\
\; &
\Gamma_\epsilon^j \coloneqq  \eta_j\big([0,\tau_j^\epsilon] \cup [\sigma_j^\epsilon,\infty)\big) 
\qquad \textnormal{and} \qquad 
\Gamma_\epsilon \coloneqq \bigcup_{j=1}^N \Gamma_\epsilon^j.
\end{align*}
Let $g_{\Gamma_\epsilon} \colon \bH \setminus \Gamma_\epsilon \to \bH$ be the conformal map normalized as $|g_{\Gamma_\epsilon}(z) - z| \to 0$ as $|z| \to \infty$. 
Write
\begin{align*}
\eta_j^\epsilon = g_{\Gamma_\epsilon}(\eta_j[\tau_j^\epsilon,\sigma_j^\epsilon]), \qquad 
a_j^\epsilon = g_{\Gamma_\epsilon}(\eta_j(\tau_j^\epsilon)), \qquad \textnormal{and} \qquad  
b_j^\epsilon = g_{\Gamma_\epsilon} ( \eta_j(\sigma_j^\epsilon)),
\end{align*}
for $j \in \{1,\ldots,N\}$, and note that the curves $\bs{\eta}^\epsilon \coloneqq (\eta_1^\epsilon,\ldots,\eta_N^\epsilon)$ 
have the unique law of an $\iaSLE_4$
by~\cite[Theorem~1.2]{BPW:On_the_uniqueness_of_global_multiple_SLEs}. 
Moreover, we already know by Theorem~\ref{thm:resampling_property} that $\iaSLE_4$ can be obtained as a set of level lines for some zero-boundary GFF. 

Take $\beta \in \GLP_\multii$ such that $\alpha \in \PLP_\multii(\beta)$, and let $u_\beta$ be the harmonic function defined via Equation~\eqref{eq:harmonic_function} 
with $\{x_1,\ldots,x_p\} = \{a_1,\ldots,a_N,b_1,\ldots,b_N\}$ (not counting the same point twice).  
Similarly,  let $u_{\imath(\beta)}^\epsilon$ denote the harmonic function defined via Equation~\eqref{eq:harmonic_function}  
with $\{x_1,\ldots,x_p\} = \{a_1^\epsilon,\ldots,a_N^\epsilon,b_1^\epsilon,\ldots,b_N^\epsilon\}$ and $p = 2N$, 
where the boundary data is such that $(u_{\imath(\beta)}^\epsilon \circ g_{\Gamma_\epsilon}^{-1})|_{\bR} = u_\beta|_{\bR}$. 
Since the Dyck path $\imath(\beta)$ only takes steps of size $\pm 1$, 
on the intervals between the points $\{a_1^\epsilon,\ldots,a_N^\epsilon,b_1^\epsilon,\ldots,b_N^\epsilon\}$, the boundary data of $u_{\imath(\beta)}^\epsilon$ jumps by $\pm 2\lambda$: 
%(see Figure~\ref{fig:boundary_data_unfused}).

\begin{figure}[h!]
\centering
\includegraphics[width=.6\textwidth]{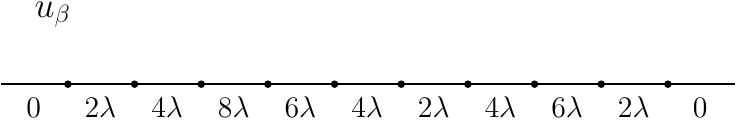} \\
\vspace{1.5em}
\includegraphics[width=.6\textwidth]{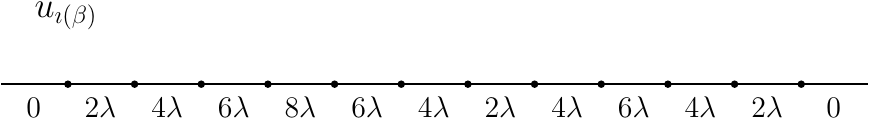}
\caption{Illustration of the boundary data of the functions $u_\beta$ and $u_{\imath(\beta)}$, for $\beta = (0,1,2,4,3,2,1,2,3,1,0)$ (so that $\imath(\beta) = (0,1,2,3,4,3,2,1,2,3,2,1,0)$). 
\label{fig:boundary_data_unfused}}
\end{figure}

For each $k$, we let $\Phi_k$ be a zero-boundary GFF in $O_k$, independent of $(\Phi_\ell)_{\ell \neq k}$ and the curves $\bs{\eta}$. We will define a coupling, by choosing the boundary data for the fields as follows.  In the domain $O_k$, we choose $H_k = u_{\imath(\beta)}^\epsilon(g_{\Gamma_\epsilon}(\partial O_k) \cap \bR)$, and set
\begin{align*}
\Psi|_{O_k} \coloneqq \Phi_k + H_k.
\end{align*}
We write $\Phi_j^\pm = \Phi_k$ and $h_j^\pm = H_k$ if $D_j^\pm = O_k$ and note that if $O_k$ and $O_\ell$ share a boundary segment, then, since $u_{\imath(\beta)}^\epsilon$ jumps by $\pm 2\lambda$, 
we have $H_k - H_\ell \in \{-2\lambda,2\lambda\}$. Thus, it follows that $h_j^+ - h_j^- \in \{-2\lambda,2\lambda\}$. Moreover, we let 
\begin{align*}
\pp_j = 
\begin{cases}
a_j , &\textnormal{if} \ a_j \in \PJ , \\
b_j , &\textnormal{if} \ b_j \in \PJ ,
\end{cases}
\qquad \textnormal{and} \qquad 
\np_j \in \{a_j,b_j\} \setminus \{ \pp_j \}. 
\end{align*}
By construction, the following now holds:
\begin{itemize}[leftmargin=*]
\item $\Psi|_{D_j^-} = \wt{\Phi}_j^- +\,  h_j^-$;
\item $\Psi|_{D_j^+} = \wt{\Phi}_j^+ +\, h_j^+$; and
\item $\eta_j$ is an $\SLE_4$ in $D_j$ from $\pp_j$ to $\np_j$. 
\end{itemize}
Thus, since $h_j^+ - h_j^- \in \{-2\lambda,2\lambda\}$, we obtain a coupling $(\Psi|_{D_j},\eta_j)$ such that $\Psi|_{D_j} = \smash{\wt{\Phi}_j + u_\beta^j}$, 
where $\smash{\wt{\Phi}_j}$ is a zero-boundary GFF in $D_j$ and $\smash{u_\beta^j}$ is the bounded harmonic function in $D_j$ with boundary data $h_j^-$ (resp.~$h_j^+$) on the clockwise (resp.\ counterclockwise) arc from $a_j$ to $b_j$, and $\eta_j$ is the level line of $\Psi|_{D_j}$ of height $h_j^+ - \tfrac{1}{2}(h_j^+ - h_j^-) = \tfrac{1}{2}( h_j^+ + h_j^-)$.  
Note that if $\pp_j = a_j$ (resp.~$\pp_j = b_j$), then $h_j^+ - h_j^- = 2\lambda$ (resp.~$h_j^+ - h_j^- = -2\lambda$). 
See Figure~\ref{fig:exploration}.

Next, we will describe what the couplings imply for the family of curves with mapped out initial segments of diameter $\epsilon$.  
On the one hand, as noted above, ${\bs \eta}^\epsilon$ has the unique law of an $\iaSLE_4$
by~\cite[Theorem~1.2]{BPW:On_the_uniqueness_of_global_multiple_SLEs} --- 
and by Theorem~\ref{thm:resampling_property}, $\iaSLE_4$ can be obtained as the set $\OML_{\imath(\beta)}(\Phi)$ for some zero-boundary GFF $\Phi$, on the event $\{ \conn(\Phi + u_{\imath(\beta)}) = \imath(\alpha) \}$. 
On the other hand, let us consider the above construction. 
Here, $\smash{\Gamma_\epsilon^j}$ is a local set for the zero-boundary GFF $\smash{\wt{\Phi}_j}$. We write its Markovian decomposition with respect to $\Gamma_\epsilon^j$ as
\begin{align*}
\wt{\Phi}_j = \wt{\Phi}_j^{\Gamma_\epsilon^j} + (\wt{\Phi}_j)_{\Gamma_\epsilon^j},
\end{align*}
and denote by $\wt{h}_j^{\Gamma_\epsilon^j}$ the bounded function which is harmonic and satisfies
\begin{align*}
\wt{h}_j^{\Gamma_\epsilon^j}|_{D_j \setminus \Gamma_\epsilon^j} = (\wt{\Phi}_j)_{\Gamma_\epsilon^j} 
\qquad \textnormal{and} \qquad
\wt{h}_j^{\Gamma_\epsilon^j}|_{\Gamma_\epsilon^j} = 0 .
\end{align*}
Then, for each $\epsilon > 0$, the curve $\eta_j^\epsilon$ is a level line of $(\smash{\wt{\Phi}_j}^{\Gamma_\epsilon^j} + \wt{h}_j^{\Gamma_\epsilon^j} + u_\beta^j) \circ g_{\Gamma_\epsilon}^{-1}$ of height $\tfrac{1}{2}( h_j^+ + h_j^-)$. 
Furthermore, $(\Psi \circ g_{\Gamma_\epsilon}^{-1})|_{g_{\Gamma_\epsilon}(O_k)} = (\Phi_k + H_k) \circ g_{\Gamma_\epsilon}^{-1}$ is a zero-boundary GFF plus $H_k$ in $g_{\Gamma_\epsilon}(O_k)$,
and given $\Gamma_\epsilon$, in the complementary connected components of the curves ${\bs \eta}^\epsilon$, the field $\Psi$ constructed above is precisely such that 
$\{ {\bs \eta}^\epsilon \}$ has the same law as the set $\OML(\Psi \circ g_{\Gamma_\epsilon}^{-1})$. 
Thus, it follows that for each $\epsilon > 0$, the conditional law of $\Psi \circ g_{\Gamma_\epsilon}^{-1}$, given $\Gamma_\epsilon$, is that of a zero-boundary GFF plus the harmonic function $u_{\imath(\beta)}^\epsilon$, conditioned on the event $\{ \conn(\Phi + u_{\imath(\beta)}^\epsilon) = \imath(\alpha) \}$. 
In particular, taking $\epsilon \to 0$ we obtain a coupling $(\Phi,{\bs \eta})$ between a zero-boundary GFF $\Phi$ on $\bH$, conditioned on the event $\{ \conn(\OML_\beta(\Phi)) = \alpha \}$, and ${\bs \eta}$ such that $\{ {\bs \eta} \} = \OML_\beta(\Phi)$, as desired. 
Since $\beta \in \GLP_\multii$ was arbitrary, the proof is complete.
\end{proof}

\begin{corollary}\label{cor:independence_of_beta}
Fix $\alpha \in \PLP_\multii$. Then, for any $\beta \in \GLP_\multii$ such that $\alpha \in \PLP_\multii(\beta)$, 
 the conditional law of $\{ \bs{\eta}\}=\OML_\beta$ given the event $\{ \conn = \alpha \}$ is independent of $\beta$.
\end{corollary}

\begin{proof}
By Theorem~\ref{thm:resampling_property}, $\bs{\eta}$ satisfy the resampling property for any $\beta \in \GLP_\multii$. 
Moreover, if a given realization of $\bs{\eta}$ satisfies $\conn(\bs{\eta}) = \alpha$ for some $\alpha \in \PLP_\multii(\beta)$, 
then we can couple it as $\OML_{\tilde{\beta}}(\wt{\Phi}) = \{\bs{\eta}\}$ for some zero-boundary GFF $\wt{\Phi}$ and boundary data $\tilde{\beta} \in \GLP_\multii$. 
Since the marginal law of $\{ \bs{\eta}\}$ remains the same, said conditional law does not depend on $\beta$.
\end{proof}

\begin{proof}[Proof of Theorem~\ref{thm:existence_fused_multiple_sle}: Uniqueness.]
By Proposition~\ref{prop:coupling_fused_multiple_sle}, any global $\aSLE_4$ can be coupled with a zero-boundary GFF $\Phi$, as $\OML_\beta(\Phi)$ for some $\beta \in \GLP_\multii$ for which $\alpha \in \PLP_\multii(\beta)$. By Corollary~\ref{cor:independence_of_beta}, this law is independent of $\beta$. Hence, there is a unique global $\aSLE_4$. 
Recalling that existence was proved in Section~\ref{subsec:resampling_property}, the proof of Theorem~\ref{thm:existence_fused_multiple_sle} is complete.
\end{proof}

\begin{remark}
If there is some $k$ such that $\dir \beta_j > 0$ for all $j \leq k$ and $\dir \beta_j < 0$ for all $j > k$, then there is only one possible link pattern (the rainbow) for the frontiers of $\bA_0^u,\ldots,\bA_{2(\maxu-1)\lambda}^u$.
\end{remark}

\subsection{Marginal law: the unconditional case}\label{subsec:local_description_unconditional_case}

At this point, we have established the existence and uniqueness of global $\aSLE_4$ measures. 
In the remainder of this section, we give a local description of $\aSLE_4$. 
We begin by describing the marginal law of one curve, and verify that for starting points of valence one, 
the conformal block $\CobloF_\beta$ is the well-known 
\quote{$\SLE_\kappa(\rho)$ partition function} (cf.~\cite{Schramm-Wilson:SLE_coordinate_changes,
Kytola:On_CFT_of_SLE_kappa_rho, Lawler:Partition_functions_loop_measure_and_versions_of_SLE}).

\subsubsection{$\SLE_4(\ul{\rho})$ description}
\label{subsec:SLE_kappa_rho_description}

Since the law of each level line of $\Phi + u_\beta$ is that of an $\SLE_4(\ul{\rho})$ curve, with explicit $\ul{\rho}$ (see \cite[Theorem~1.1.1]{Wang-Wu:Level_lines_of_Gaussian_free_field_I}), 
we know the exact Loewner description of said level line.
For each $j \in \{1,\ldots,p\}$, we let
\begin{align*}
\rho_j \coloneqq 2(\beta_j-\beta_{j-1}) = 2 \dir \beta_j.
\end{align*}
Assuming $\rho_k = 2$ (so $\dir \beta_k = 1$, and the jump in boundary data for $u_\beta$ at $x_k$ is $2\lambda$), 
the level line $\smash{\eta_{(2\beta_k -1)\lambda}^{x_k}}$ of $\Phi+u_\beta$ of height $(2\beta_k - 1)\lambda$ emanating from $x_k$ is an $\SLE_4(\ul{\rho}^\Left;\ul{\rho}^\Right)$ process in $(\bH;x_k,\infty)$ with force points 
$\ul{x}^\Left = (x_1,\ldots,x_{k-1})$ and $\ul{x}^\Right = (x_{k+1},\ldots,x_p)$ 
and weights $\ul{\rho}^\Left = (\rho_{k-1},\ldots,\rho_1)$ and $\ul{\rho}^\Right = (\rho_{k+1},\ldots,\rho_p)$. 
Consequently, the SDEs
\begin{align}
\label{eq:GFF level line as SLE kappa rho}
\ud V_k(t) 
= \; & 2 \, \ud B(t) + 2\sum_{\substack{ 1 \leq j \leq p \\ j \neq k}} \frac{(\dir \beta_k)(\dir \beta_j) \, \ud t}{ V_k(t) - V_j(t)} , \qquad V_k(0) = x_k , \\
\nonumber
\ud V_j(t) = \; & \frac{2 \, \ud t }{ V_j(t) - V_k(t)}, \qquad V_j(0) = x_j , \quad j \in \{1,\ldots,p\} \setminus \{ k \} ,
\end{align}
describe the evolution of the image of its tip (and hence its driving function).
If we instead assume that $\rho_k = -2$ (so $\dir \beta_k = -1$, and the jump in boundary data for $u_\beta$ at $x_k$ is $-2\lambda$), then $x_k$ is the terminal point of a level line of $\Phi + u_\beta$.  
However, as described at the end of Section~\ref{subsubsec:level_lines}, 
we can explore this level line in reverse time, by exploring the level line from $x_k$ of the field $-(\Phi + u_\beta)$.  
In the coupling, this corresponds to changing the sign of each weight $\rho_j$: 
the evolution of the tip of the time-reversed level line is obtained by considering the evolution with weights $-\rho_j = (\dir \beta_k) \rho_j = 2 (\dir \beta_k) (\dir \beta_j)$. The same SDE system as in Equation~\eqref{eq:GFF level line as SLE kappa rho} is then obtained.

\subsubsection{Description as an $\SLE_4$ with partition function}

Denote 
\begin{align*}
\chamber_p \coloneqq \{ (x_1 , \ldots , x_p) \in \bR^p \;|\; x_1  < \cdots < x_p \}.
\end{align*}
An \emph{$\SLE_\kappa$ process with partition function $\PartF \colon \chamber_p \to \bR$} started from $x_k$, 
with initial points $(x_1 , \ldots , x_p) \in \chamber_p$, is defined as the Loewner chain driven by $V_k=(V_k(t))_{t \geq 0}$ obtained from the SDEs
\begin{align}\label{eq:SLE with partition function}
\begin{split}
\ud V_k(t) = \; & \sqrt{\kappa} \, \ud B(t) + \kappa \, \partial_k \log \PartF (V_1(t) , \ldots , V_p(t)) \, \ud t , \qquad V_k(0) = x_k , \\
\ud V_j(t) = \; &\frac{2 \, \ud t }{ V_j(t) - V_k(t)} , \qquad V_j(0) = x_j , \quad j \in \{1,\ldots,p\} \setminus \{k \}.
\end{split}
\end{align}
A simple direct computation gives the following result.

\begin{proposition}
\label{prop:uncond one-curve marginal}
Fix $k \in \{1,\ldots,p\}$, points $(x_1 , \ldots , x_p) \in \chamber_p$, valences $\multii=(s_1, \ldots, s_p) \in \bZpos^p$ with $s_k = 1$, 
and boundary data $\beta \in \GLP_\multii$.
The Loewner driving function $V_k$ solving the SDEs~\eqref{eq:GFF level line as SLE kappa rho} 
for the GFF level line given by $\SLE_4(\rho)$ coincides with the solution of the SDEs~\eqref{eq:SLE with partition function}
given by the $\SLE_4$ process with partition function $\PartF = \CobloF_\beta$ given by the conformal block function
\begin{align*}
\CobloF_\beta (x_1,\ldots,x_p) \coloneqq \prod_{1 \leq i < j \leq p} (x_j - x_i)^{\frac{1}{2} \dir \beta_{j} \,\dir \beta_i} ,
\end{align*}
corresponding to the boundary data $\beta$ \textnormal{(}Equation~\eqref{eq:CobloF_ms}\textnormal{)}.
\end{proposition}

\begin{proof}
A direct computation of the logarithmic derivative of the explicit function $\CobloF_\beta$ yields
\begin{align*}
\partial_k \log \CobloF_\beta (x_1,\ldots,x_p)
= \frac{1}{2} \sum_{\substack{1 \leq j \leq p \\ j \neq k}} \frac{(\dir \beta_k) (\dir \beta_j)}{x_k - x_i} ,
\end{align*}
and we see that the processes solving Equations~\eqref{eq:GFF level line as SLE kappa rho} and~\eqref{eq:SLE with partition function} agree.
\end{proof}

\begin{remark}\label{rem:SLE kappa rho special times} \
\begin{enumerate}[leftmargin=*]
\item From $\SLE_\kappa (\rho)$ theory, we know that if $\rho_{k-1}$ and $ \rho_{k+1}$ have the same sign as $\rho_k$, then the above process can be launched with $x_{k-1}=x_k=x_{k+1}$. Then, for any $t>0$ up to the continuation threshold, we have $V_{k-1}(t) < V_{k}(t) < V_{k+1}(t)$, and the SDEs~(\ref{eq:GFF level line as SLE kappa rho},~\ref{eq:SLE with partition function}) make sense.

\item The terminal point of an $\SLE$ curve with given partition function can be very non-trivial to solve (cf.~\cite{Peltola-Wu:Global_and_local_multiple_SLEs_and_connection_probabilities_for_level_lines_of_GFF, Karrila:Computation_of_pairing_probabilities_in_multiple-curve_models}). 
In our setup, we know the set of possible endpoints immediately from the force points and their weights.
The actual terminal point is random, and is an ingredient in determining the event 
$\{\conn(\bs \eta) = \alpha\} = \big\{\conn(\OML_\beta) = \alpha\big\}$. \qedhere
\end{enumerate}
\end{remark}

\subsection{Marginal law: the conditional case}\label{subsec:local_description_conditional_case}

We now consider the marginal laws of the collection of level lines $\OML_\beta$ given their configuration $\alpha \in \LP_\multii$.
We give a local description for these level lines, which in particular only depends on $\alpha$ but not on $\beta$.

\begin{proposition}
\label{prop:alpha-conditional marginal law of one level line}
Fix $\tilde \multii = (s_1, \ldots, s_{\tilde p}) \in \bZpos^{\tilde p}$, and $\tilde{ \beta } \in \GLP_{\tilde{\multii}}$, and $\tilde \alpha \in \PLP_{\tilde \multii} ( \tilde \beta)$. 
Consider the GFF with boundary data $\tilde{\beta}$ at $x_1 < \cdots < x_{\tilde{p}}$. 
Let $\multii = (s_1, \ldots, s_p) \in \bZpos^p$ be the valence vector of length $p=\tilde{p} + 2$ 
obtained by the following \quote{partial unfusing} of $\tilde{\multii}$ at a given index $k \in \{1,\ldots, \tilde{p} \}$\textnormal{:}
\begin{align*}
s_j = \; &  \tilde s_j, \quad j \in \{1,\ldots,k-1\}, \\
s_{j+2} = \; &  \tilde s_j, \quad  j \in \{k+1,\ldots,\tilde p\}, \\
s_{k+1} = \; & 1, \\
s_k + s_{k+2} + 1 = \; & \tilde s_k,
\end{align*}
and let $\alpha \in \PLP_\multii(\beta)$ and  $\beta \in \GLP_\multii$
be such that $\imath(\alpha) = \imath(\tilde \alpha)$ and $\imath(\beta) = \imath(\tilde \beta)$ \textnormal{(}note that this uniquely determines $\alpha,\beta$\textnormal{)}.  
Conditionally on the event $\smash{\big\{\conn(\OML_{\tilde \beta})= \tilde{\alpha}\big\}}$, the marginal law of the 
$(\summ_k+1)$:st\footnote{We use the notation $\summ_k$ from Equation~\eqref{eq: definition of N} and enumerate the level lines from left to right.} level line
$\eta \in \OML_{\tilde \beta}$, emanating from $x_k$, 
is that of an $\SLE_4$ curve with partition function $\PartF = \PartF_\alpha$ 
given by Equation~\eqref{eq:PartF}\textnormal{:} that is, its Loewner driving function $V_{k+1}(t)$ satisfies
\begin{align*}
\ud V_{k+1}(t) = \; &  2 \, \ud B(t) + 4 \partial_{k+1} \log \PartF_\alpha(V_1(t),\ldots,V_p(t)) \, \ud t, \\
\ud V_j(t) = \; &  \frac{2 \, \ud t }{ V_j(t) - V_{k+1}(t)}, \qquad j \neq k+1,
\end{align*}
where $B$ is a Brownian motion, and the process $(V_1(t),\ldots,V_p(t))$ is started from 
$V_j(0) = x_j$ for $j \in \{1,\ldots, k-1\}$, and 
$V_k(0) = V_{k+1}(0) = V_{k+2}(0) = x_k$, 
and $V_j(0) = x_{j+2}$ for $j \in \{k+1,\ldots,\tilde p\}$.
\end{proposition}

\begin{proof}
Let $(\sF_t)_{t \geq 0}$ be the completed right-continuous filtration of the Loewner driving function $V_{k+1}$ of the level line $\eta$ (in its half-plane capacity parametrization).
Since up to its terminal time $T$, each segment of $\eta$ is a local set for the GFF, 
Theorem~\ref{thm:corss_proba_H} implies that
\begin{align*}
\PR_{\tilde \beta} \big[ \conn(\OML_{\tilde \beta}) = \tilde \alpha \;|\; \sF_t \big] 
\; = \;  \frac{\PartF_\alpha (V_1(t), \ldots, V_p(t))}{\CobloF_\beta (V_1(t), \ldots, V_p(t))} =: M(t) , \qquad 0< t < T ,
\end{align*}
which is a tautological local martingale. Moreover, by Lemmas~\ref{lem:VF_fusion}~\&~\ref{lem:PartF_fusion} it can be continuously extended to time $t=0$ by setting
\begin{align*}
M(0) = \PR_{\tilde \beta} \big[ \conn(\OML_{\tilde \beta}) = \tilde \alpha  \big] = \frac{\PartF_{\tilde \alpha} (x_1, \ldots, x_{\tilde{p}}) }{\CobloF_{\tilde \beta } (x_1, \ldots, x_{\tilde{p}}) }.
\end{align*}
This yields the Radon--Nikodym derivative 
\begin{align*}
\frac{\ud \PR_{\tilde \beta}^{\tilde \alpha}}{ \ud \PR_{\tilde \beta} }\bigg|_{\sF_t} = \frac{M(t)}{M(0)} , \qquad \textnormal{for} \qquad 
\PR_{\tilde \beta}^{\tilde \alpha} \coloneqq \PR_{\tilde \beta} \big[ \, \cdot \, | \, \conn(\OML_{\tilde \beta}) = \tilde \alpha \big] .
\end{align*} 
By Girsanov's theorem, if $Y(t)$ denotes the $\PR_{\tilde \beta}$-local martingale part of $V_{k+1}(t)$, 
then
\begin{align*}
Y(t) - \langle Y, L \rangle(t) = 2 B(t) , \qquad L(t) = \int_{0}^t \frac{ \ud M(s)}{M(s)} ,
\end{align*}
where $B$ is a $\PR_{\tilde \beta} \big[ \,\cdot \, | \, \conn(\OML_{\tilde \beta})= \tilde \alpha\big]$-Brownian motion. 
Here $\langle Y, L \rangle(t)$ only depends on the $\PR_{\tilde \beta}$-local martingale parts of the semimartingales $V_{k+1}$ and $L$, and hence straightforward It\^{o} calculus gives
\begin{align*}
\ud Y(t) = 2 \, \ud B(t) + 4 \partial_{k+1} \log \bigg( \frac{\PartF_\alpha (V_1(t),\ldots,V_p(t) )}{\CobloF_\beta (V_1(t),\ldots,V_p(t) )} \bigg) \ud t .
\end{align*}
Summing with the $\PR_{\tilde \beta}$-finite variation part of $V_{k+1}(t)$ given in Proposition~\ref{prop:uncond one-curve marginal} implies that
$$\ud V_{k+1}(t) =  2 \, \ud B(t) + 4 \partial_{k+1} \log \PartF_\alpha(V_1(t),\ldots,V_p(t)) \, \ud t,$$
as claimed (see also Remark~\ref{rem:SLE kappa rho special times}). 
\end{proof}

Let us highlight the following immediate consequences, which follow from the well-known GFF coupling~\cite{Dubedat:SLE_and_free_field, Schramm-Sheffield:A_contour_line_of_the_continuum_GFF, Miller-Sheffield:Imaginary_geometry1, Wang-Wu:Level_lines_of_Gaussian_free_field_I}.

\begin{corollary}\label{cor:properties of one-curve marginal} \
\begin{enumerate}[leftmargin=*]
\item\label{it:conditional_law_indep_of_beta} The above conditional marginal law of one curve given the topological configuration $\tilde \alpha$ does not depend on the GFF boundary data $\tilde \beta$ \textnormal{(}as long as $\tilde \alpha \in \PLP_{\tilde \multii}(\tilde \beta)$\textnormal{)}.

\smallskip

\item 
Even with the initial data $V_k(0) = V_{k+1}(0) = V_{k+2}(0) = x_k$,
the $\SLE_4$ process with partition function $\PartF_\alpha$ is well-defined:
the processes $V_k,V_{k+1},V_{k+2}$ separate immediately and, e.g., $\PartF_\alpha(V_1(t),\ldots,V_p(t)) $ is well-defined for all positive times up to the termination time of the curve.

\smallskip

\item At the termination time, the curve ends at the boundary point paired with $x_k$ according to $\tilde \alpha$.
\end{enumerate}
\end{corollary}

\bigskip

\section{Metric graph GFF: FPSs and connection probabilities}
\label{sec:mGFF}

The main goal of this section is to establish Theorem~\ref{thm:corss_proba_MGFF}, 
which states that the connection probabilities of the metric graph GFFs converge to those of the continuum GFF.  
In order to do this, we first introduce the discrete GFF in Section~\ref{subsec:DGFF} and then the metric graph GFF in Section~\ref{subsec:MGFF}. 
Then, in Section~\ref{subsubsec:fps_metric_graph} we review the first passage sets of metric graph GFFs and state a convergence result for them from~\cite{ALS:First_passage_sets_of_2D_GFF}. 
We prove Theorem~\ref{thm:corss_proba_MGFF} in Section~\ref{subsec:mGFF_Connection_probabilities}.

\subsection{Discrete Gaussian free fields}\label{subsec:DGFF}

Let $G= (V, E)$ be a finite connected graph with given conductances $C \colon E \to \bRpos$ for the edges and a non-empty set $V^{\partial}$ of boundary vertices, 
and denote $V^\circ \coloneqq V \setminus V^\partial$.  
We write $x \sim y$ if $\langle v, w \rangle \in E$ and $C(v, w) \coloneqq C(\langle v, w \rangle)$, and we interpret $f \in \mathbb{R}^{V^{\circ}}$ as functions on $V$ by setting $f(v) := f_v$ and $f(w) = 0$ for $w \in V^\partial$.

The \emph{discrete Laplacian} on the graph $G$  maps a function $f \colon V^\circ \to \bR$ linearly to another function $(\Delta_G f) \colon V^\circ \to \bR$ (so it is a matrix $\Delta_G \in \mathbb{R}^{V^\circ \times V^\circ}$); it is defined by\footnote{Note that there are varying normalization conventions in the literature.}
\begin{align*}
(\Delta_G \, f)(v) \coloneqq \sum_{w \sim v} C(v,w) (f(w) - f(v)).
\end{align*}
The \emph{discrete Dirichlet inner product} for two functions $f, g : V^\circ \to \mathbb{R}$ is defined as
\begin{align*}
( f , g )_{\nabla_G} & \coloneqq \sum_{\langle v, w \rangle \in  E} (f(w) - f(v)) \, C(v,w) \, (g(w) - g(v)).
\end{align*}
Writing this as $1/2$ times the analogous sum over directed edges and re-indexing, we obtain
\begin{align*}
( f , g )_{\nabla_G} 
%=\; & \frac{1}{2} \sum_{w \in V^{\circ}} \sum_{ v \sim w } f(w)  \, C(w,v) \, (g(w) - g(v)) - \frac{1}{2} \sum_{v \in V^{\circ}} \sum_{   w \sim v } f(v)  \, C(v,w) \, (g(w) - g(v)) \\
= \; & - \sum_{v \in V^\circ} f(v) (\Delta_G \, g)(v) = -f^T \Delta_G g .
\end{align*}

Finally, the \emph{discrete GFF with zero boundary conditions} on $G$ is
a centered Gaussian vector $\phi = \phi^G$ in $\bR^{V^\circ}$ (which we still interpret as a function on $V$ as above) with law
\begin{align*}
\ud \PR [\phi] 
\coloneqq \; & \;  \frac{1}{Z_G} \exp \Big( \!-\!\tfrac{1}{2} ( \phi , \phi )_{\nabla_G} \Big) \prod_{w \in V^\circ} \ud\phi (w)  \\
= \; & \frac{1}{Z_G} \exp \Big( \tfrac{1}{2} \phi^T \Delta_G \phi \Big) \prod_{w \in V^\circ} \ud\phi (w),
\qquad 
Z_G = \sqrt{\frac{(2\pi)^{|V^\circ|}}{\det (-\Delta_G)}} .
\end{align*}

A function $f \colon V \to \bR$ is said to be \emph{discrete harmonic} on $\wt{V} \subset V$ if $(\Delta_G f)(v) = 0$ for all $v \in \wt{V}$. Given a function $f^\partial \colon V^\partial \to \bR$, there exists a unique function $f \colon V \to \bR$ such that $f|_{V^\partial} = f^\partial$ and $(\Delta_G f)(v) = 0$ for all $v \in V^\circ$,
namely, the harmonic extension of $f^\partial$ to $V$ or $G$.  
We will typically abuse notation slightly and use the same notation for a boundary function and its harmonic extension. Similarly to the continuum case, the \emph{discrete GFF with boundary data} $u^\partial: V^\partial \to \bR$ is defined as $\phi + u$, where $\phi$ a zero-boundary GFF on $G$ and $u$ the harmonic extension of the boundary function $u^\partial$.

Finally, the \emph{discrete GFF on a graph with edge length} refers to the above discrete Gaussian fields with the conductances being inverse edge lengths: $C(v, w) = 1/\mathrm{length}(\langle v, w\rangle)$.

\subsection{Metric graph Gaussian free field}\label{subsec:MGFF}

Given a graph $G=(V,E)$ as in Section~\ref{subsec:DGFF} with conductances $(C(e))_{e \in E}$, we can associate to it a \quote{metric graph} $\wt{G}$ by replacing each edge $e$ with a continuous line segment of length $1/C(e)$. 
The \emph{metric graph GFF} $\wt{\phi}$ on the metric graph $\wt{G}$ is the random model defined by sampling a discrete GFF $\phi$ on $G$ (with some prescribed boundary data, which may split some edge in the middle, thus simply reducing its length) and adding\footnote{That is, linearly interpolating along the each edge $\langle v, w \rangle$ of $G$ by adding a Brownian bridge of length $1/C(\langle v, w \rangle)$ from $\phi(v)$ to $\phi(w)$.} 
on top of each edge $e \in E$ a Brownian bridge of length $1/C(e)$.   
The metric graph GFF on $\wt{G}$ with boundary data $u^\partial \colon V^\partial \to \bR$ is defined as $\wt{\phi} + u$, 
where $u$ is the harmonic extension of $u^\partial$ to $G$ and extended by linear interpolation inside the edges. 
We shall often just call this the \quote{metric graph GFF with boundary data $u$.}

Just as in the case of the continuum GFF, the metric graph GFF satisfies a strong Markov property. We say that a random subset $A$ of $\wt{G}$ is \emph{optional} for $\wt{\phi}$ if for every open deterministic subset $O$ of $\wt{G}$, the event $\{ A \subseteq O \}$ is measurable with respect to the restriction of $\wt{\phi}$ to $O$. The optional sets of a metric graph GFF play the role of local sets of the continuum GFF. Out of convenience, assume that $A$ almost surely has finitely many connected components, so that $\wt{G} \setminus A$ has finitely many connected components and the closure of each component is a metric graph (where an edge can be split by $A$ into different components, or partially covered by $A$).  
Then, the following strong Markov property holds, see~\cite[Section~3]{Lupu:From_loop_clusters_and_random_interlacements_to_the_free_field}.

\begin{proposition}[Strong Markov property, {\cite{Lupu:From_loop_clusters_and_random_interlacements_to_the_free_field}}]
Let $A$ be a random compact subset of $\wt{G}$ with finitely many connected components which is optional for the metric graph GFF $\wt{\phi} = \wt{\phi}^{\wt{G}}$. Then, 
\begin{align*}
\wt{\phi} = \wt{\phi}_A + \wt{\phi}^A,
\end{align*}
where, conditionally on $A$, 
the function $\wt{\phi}^A$ is a zero-boundary metric graph GFF on $\wt{G} \setminus A$, extended to be $0$ on $A$, 
and the function $\wt{\phi}_A$ agrees with $\wt{\phi}$ on $A$, restricts to a harmonic function $h_A$ on $\wt{G} \setminus A$ with boundary values given by $\wt{\phi}|_{\wt{G} \setminus A}$, and is independent of $\wt{\phi}^A$. 
\end{proposition}

\subsection{First passage sets and convergence}\label{subsubsec:fps_metric_graph}

Consider a metric graph GFF $\wt{\phi}+u$ on $\wt{G}$ with boundary data $u$.  
The (upper) first passage set (FPS) of height $a \in \bR$ of $\wt{\phi}+u$ is defined as
\begin{align*}
\wt{\bA}_a^u(\wt{\phi}) = \wt{\bA}_a(\wt{\phi} + u) 
\coloneqq \big\{ x \in \wt{G} \colon \exists \ \textnormal{continuous path} \ P \ \textnormal{from} \ V^\partial \ \textnormal{to} \ x \ \textnormal{such that} \ (\wt{\phi} + u)|_P \geq a \big\}.
\end{align*}
(Note that, contrary to FPSs of the continuous GFF, those of the metric graph GFF do not automatically include the entire boundary of the domain.)
Moreover, such a set $\wt{\bA}_a^u$ is optional for $\wt{\phi}$ and on $\wt{\bA}_a^u$, 
the function $\wt{\phi}_{\wt{\bA}_a^u}$ takes the boundary value $a$.

We now consider the convergence of metric graph GFF FPSs to continuum GFF FPSs, explored in~\cite[Section~4]{ALS:First_passage_sets_of_2D_GFF}. 
We shall, however, be concerned with polygons rather than complicated domains, so our setting is a bit less technical. 
We let $\wt{\bZ}_n^2$ denote the metric graph induced by $\tfrac{1}{n} \bZ^2$. 
Our precise convergence setup pertaining to the FPSs and their connection probabilities is as follows.
(We have chosen a setting that seems convenient; one could modify the setup in various ways, which we leave to dedicated readers.)

Consider a polygon $(\Omega;y_1,\ldots,y_N)$ such that $\Omega \subset [-C,C]^2$ for some $C > 0$. Furthermore, for each $n \in \bZpos$, 
let $(\Omega^n;y_1^n,\ldots,y_N^n)$ be a polygon such that $\Omega^n \subset [-C,C]^2$ and such that
\begin{enumerate}[label=(\roman{*}), ref=(\roman{*})]
\item $[-C,C]^2 \setminus \Omega^n$ converges to $[-C,C]^2 \setminus \Omega$ in the Hausdorff metric as $n \to \infty$; and 

\smallskip

\item $y_j^n \to y_j$ (as prime ends, if necessary) as $n \to \infty$, for all $j \in \{1,\ldots, N\}$.
\end{enumerate}
Moreover, we set $\wt{\Omega}^n \coloneqq \Omega^n \cap \wt{\bZ}_n^2$ and let $\wt{\phi}^n$ be a metric graph GFF on $\wt{\Omega}^n$. 
(Observe that the boundary values on $\Omega^n$ give the boundary values on $\wt{\Omega}^n$, while $\partial \wt{\Omega}^n$ is not necessarily comprised of vertices of $\tfrac{1}{n} \bZ^2$.)

Next, if $\wt{f}^n$ is a function on $\wt{\Omega}^n$, then we denote by $\wh{f}^n$ the harmonic extension of $\wt{f}^n$ to $[-C,C]^2$ with zero boundary values on $\partial [-C,C]^2$. 
We can also extend $\wt{\phi}^n$ this way and denote the extension by $\wh{\phi}^n$. We will make use of the following result.

\begin{proposition}[{\cite[Proposition~4.7]{ALS:First_passage_sets_of_2D_GFF}}]\label{prop:metric_graph_gff_fps_convergence}
Let $\Phi$ be a zero-boundary GFF on $\Omega$, let $u$ be a bounded harmonic function with piecewise constant boundary values.
Let $(u^n)_{n \geq 1}$ be a collection of bounded harmonic functions in $\wt{\Omega}^n$, which converge uniformly on compact subsets to $u$ as $n \to \infty$. 
Then, we can couple $(\wh{\phi}^n)_{n \geq 1}$ and $\Phi$ such that $\wh{\phi}^n$ converge to $\Phi$ in probability as generalized functions, and upon doing so, 
$(\wh{\phi}^n,\wt{\bA}_a^{u^n}(\wt{\phi}^n) \cup \partial \Omega) \; \smash{\overset{n \to \infty}{\longrightarrow}} \; (\Phi,\bA_a^u(\Phi))$ in probability, where the convergence in the second coordinate is with respect to the Hausdorff metric~\eqref{eq:hausdorff_distance}.
\end{proposition}

\subsection{Connection probabilities} \label{subsec:mGFF_Connection_probabilities}

Similarly to the continuum case, we have the corresponding connection probability result for the frontiers of the metric graph GFF.

Consider a boundary condition $u = u_\beta$ analogous to Equation~\eqref{eq:harmonic_function}, where $\bH$ is replaced by $\Omega$ and $x_1,\ldots,x_p$ by $y_1,\ldots,y_p$. 
We let $\Omega^n$ be a polygon approximating $\Omega$ as in Section~\ref{subsubsec:fps_metric_graph}, 
and we let $u^n \colon \partial \Omega^n \to \bR$ be the boundary data such that $u^n|_{(y_j^n,y_{j+1}^n)} = u|_{(y_j,y_{j+1})}$ for $j \in \{1,\dots,N\}$ (with the convention that $y_{N+1} = y_1$). 
Let $\conn^n$ denote the random valenced link pattern formed by the frontiers of the first passage sets $\wt{\bA}_0^{u^n}(\wt{\phi}^n),\ldots,\wt{\bA}_{2(\maxu-1)\lambda}^{u^n}(\wt{\phi}^n)$. Then, the following holds (see~\cite[Lemma~5.5]{Liu-Wu:Scaling_limits_of_crossing_probabilities_in_metric_graph_GFF} for a similar statement).

\begin{citedtheorem} \label{thm:corss_proba_MGFF}
In the convergence setup of Section~\ref{subsubsec:fps_metric_graph}, we have \textnormal{(}denoting the coupling measure $\PR_\beta$\textnormal{)}
\begin{align*}
\lim_{n \to \infty} \PR_\beta[ \conn^n = \alpha] = \PR_\beta [\conn = \alpha] , \qquad \alpha \in \PLP_\multii.
\end{align*}
\end{citedtheorem}

\begin{proof}
It suffices to prove that for every subsequence of $(n_k)_{k}$, there is a further subsequence $(n_{k_\ell})_\ell$ along which the claimed convergence holds. We may, and will, choose $k_\ell$ so that the convergence in probability of Proposition~\ref{prop:metric_graph_gff_fps_convergence} also holds in the almost sure sense along $(n_{k_\ell})_\ell$. 
We now claim that along this subsequence denoted by $\ell$, the convergence $\smash{\conn^{\ell} \to \conn}$ as $\ell \to \infty$ holds almost surely. 
To this end, note that $\conn$ is uniquely determined by which segments $(y_j, y_{j+1})$, for $1 \leq j \leq p$, are connected to each other within the FPS complements $ \Omega \setminus ( \bA_{2m\lambda}^u \setminus \partial \Omega )$ for $m \in \{0,1,\ldots,\maxu-1\}$. A similar property holds for $\smash{\conn^{\ell}}$ and the metric graph FPSs. 
Observe that, if some segments $(y_j, y_{j+1})$ and $(y_i, y_{i+1})$ are connected by a path $P$ in $ \Omega \setminus ( \bA_{2m\lambda}^u \setminus \partial \Omega )$, 
then, since this FPS was cut out by a chordal curve inside $\Omega$, there exists $\delta >0$ such that a $\delta$-neighborhood of $P$ lies inside $\Omega \setminus ( \bA_{2m\lambda}^u \setminus \partial \Omega )$. 
Now, by the almost sure Hausdorff convergence, the corresponding boundary arcs $\smash{(y^{\ell}_j, y^{\ell}_{j+1})}$ and $\smash{(y^{\ell}_i, y^{\ell}_{i+1})}$ 
are also connected in $ \Omega^{\ell} \setminus \bA_{2m\lambda}^{u^{\ell}}(\smash{\wt{\phi}^{n_{k_\ell}}})$, for all $\ell$ large enough. 
Since there are finitely many possible values for $j,i$, and $m$, we conclude that taking $\ell$ large enough, 
we have $\smash{\conn^{\ell}} = \conn$. This proves the desired convergence along the subsequence, which implies the result.
\end{proof}

\newcommand{\changeurlcolor}[1]{\hypersetup{urlcolor=#1}}      
\changeurlcolor{black}

\bigskip{}
\bibliographystyle{annotate}

\begin{thebibliography}{FLPW24}

\bibitem[AHSY26]{AHSY:Conformal_welding_of_quantum_disks_and_multiple_SLE_the_non-simple_case}
Morris Ang, Nina Holden, Xin Sun, and Pu~Yu.
\newblock Conformal welding of quantum disks and multiple {SLE}: the non-simple
  case.
\newblock {\em Probab. Theory Related Fields, to appear}, 2026.
\newblock Preprint in arXiv:2310.20583.


\bibitem[ALS20a]{ALS:First_passage_sets_of_the_2D_continuum_GFF}
Juhan Aru, Titus Lupu, and Avelio Sep\'ulveda.
\newblock First passage sets of the 2{D} continuum {G}aussian free field.
\newblock {\em Probab. Theory Related Fields}, 176(3-4):1303--1355, 2020.


\bibitem[ALS20b]{ALS:First_passage_sets_of_2D_GFF}
Juhan Aru, Titus Lupu, and Avelio Sep\'ulveda.
\newblock The first passage sets of the {$2D$} {G}aussian free field:
  convergence and isomorphisms.
\newblock {\em Comm. Math. Phys.}, 375(3):1885--1929, 2020.


\bibitem[AMY25]{AMY:Multiple_SLE_from_CLE}
Valeria Ambrosio, Jason Miller, and Yizheng Yuan.
\newblock Multiple $\mathrm{SLE}_\kappa$ from $\mathrm{CLE}_\kappa$ for $\kappa
  \in (4,8)$.
\newblock Preprint in arXiv:2503.08958, 2025.

\bibitem[AS18]{Aru-Sepulveda:Two-valued_local_sets_of_2D_continuum_GFF_connectivity_labels_and_induced_metrics}
Juhan Aru and Avelio Sep\'ulveda.
\newblock Two-valued local sets of the 2{D} continuum {G}aussian free field:
  connectivity, labels, and induced metrics.
\newblock {\em Electron. J. Probab.}, 23(61):1--35, 2018.


\bibitem[ASW19]{ASW:BTLS}
Juhan Aru, Avelio Sep\'ulveda, and Wendelin Werner.
\newblock On bounded-type thin local sets of the two-dimensional gaussian free
  field.
\newblock {\em J. Inst. Math. Jussieu}, 18(3):591--618, 2019.


\bibitem[BBK05]{BBK:Multiple_SLEs_and_statistical_mechanics_martingales}
Michel Bauer, Denis Bernard, and Kalle Kyt{\"o}l{\"a}.
\newblock Multiple {S}chramm-{L}oewner evolutions and statistical mechanics
  martingales.
\newblock {\em J. Stat. Phys.}, 120(5-6):1125--1163, 2005.


\bibitem[BP25]{Berestycki-Powell:Gaussian_free_field_and_Liouville_quantum_gravity}
Nathana{\"e}l Berestycki and Ellen Powell.
\newblock {\em Gaussian free field and {L}iouville quantum gravity}.
\newblock Cambridge Studies in Advanced Mathematics. Cambridge University
  Press, 2025.


\bibitem[BPW21]{BPW:On_the_uniqueness_of_global_multiple_SLEs}
Vincent Beffara, Eveliina Peltola, and Hao Wu.
\newblock On the uniqueness of global multiple $\mathrm{SLE}$s.
\newblock {\em Ann. Probab.}, 49(1):400--434, 2021.


\bibitem[Car03]{Cardy:SLE_and_Dyson_circular_ensembles}
John~L. Cardy.
\newblock Stochastic {L}oewner evolution and {D}yson's circular ensembles.
\newblock {\em J. Phys. A}, 36(24):L379--L386, 2003.


\bibitem[Dub07a]{Dubedat:Commutation_relations_for_SLE}
Julien Dub{\'e}dat.
\newblock Commutation relations for {$\mathrm{SLE}$}.
\newblock {\em Comm. Pure Appl. Math.}, 60(12):1792--1847, 2007.


\bibitem[Dub07b]{Dubedat:Duality_of_SLE}
Julien Dub{\'e}dat.
\newblock Duality of {S}chramm-{L}oewner {E}volutions.
\newblock {\em Ann. Sci. {\'E}c. Norm. Sup{\'e}r.}, 42(5):697--724, 2007.


\bibitem[Dub09]{Dubedat:SLE_and_free_field}
Julien Dub{\'e}dat.
\newblock $\mathrm{SLE}$ and the free field: partition functions and couplings.
\newblock {\em J. Amer. Math. Soc.}, 22(4):995--1054, 2009.


\bibitem[FLPW24]{FLPW:Multiple_SLEs_Coulomb_gas_integrals_and_pure_partition_functions}
Yu~Feng, Mingchang Liu, Eveliina Peltola, and Hao Wu.
\newblock Multiple {SLE}s for {$\kappa\in (0,8)$:} {C}oulomb gas integrals and
  pure partition functions.
\newblock Preprint in arXiv:2406.06522, 2024.

\bibitem[HL21]{Healey-Lawler:N_sided_radial_SLE}
Vivian~O. Healey and Gregory~F. Lawler.
\newblock {$N$}-sided radial {S}chramm–{L}oewner evolution.
\newblock {\em Probab. Theory Related Fields}, 181(1-3):451--488, 2021.


\bibitem[HPW25]{HPW:Multiradial_SLE_with_spiral}
Chongzhi Huang, Eveliina Peltola, and Hao Wu.
\newblock Multiradial {SLE} with spiral: resampling property and boundary
  perturbation.
\newblock Preprint in arXiv:2509.22045, 2025.

\bibitem[Kac80]{Kac:Highest_weight_representations_of_infinite_dimensional_Lie_algebras}
Victor~G. Kac.
\newblock Highest weight representations of infinite dimensional {L}ie
  algebras.
\newblock In {\em Proceedings of the ICM 1978, Helsinki, Finland}, volume~1,
  pages 299--304. Acad. Sci. Fenn., 1980.

\bibitem[Kar19]{Karrila:Multiple_SLE_local_to_global}
Alex Karrila.
\newblock Multiple {$\mathrm{SLE}$} type scaling limits: from local to global.
\newblock Preprint in arXiv:1903.10354, 2019.

\bibitem[Kar20]{Karrila:UST_branches_martingales_and_multiple_SLE2}
Alex Karrila.
\newblock U{ST} branches, martingales, and multiple {$\mathrm{SLE}(2)$}.
\newblock {\em Electron. J. Probab.}, 25:1--37, 2020.


\bibitem[Kar26]{Karrila:Computation_of_pairing_probabilities_in_multiple-curve_models}
Alex Karrila.
\newblock A new computation of pairing probabilities in several multiple-curve
  models.
\newblock {\em ALEA Lat. Am. J. Probab. Math. Stat., to appear}, 2026.
\newblock Preprint in arXiv:2208.06008.


\bibitem[Kem17]{Kemppainen:SLE_book}
Antti Kemppainen.
\newblock {\em Schramm-{L}oewner evolution}, volume~24 of {\em SpringerBriefs
  in Mathematical Physics}.
\newblock Springer Cham, 2017.


\bibitem[KKP20]{KKP:Boundary_correlations_in_planar_LERW_and_UST}
Alex Karrila, Kalle Kyt{\"o}l{\"a}, and Eveliina Peltola.
\newblock Boundary correlations in planar {LERW} and {UST}.
\newblock {\em Comm. Math. Phys.}, 376(3):2065--2145, 2020.


\bibitem[KL07]{Kozdron-Lawler:Configurational_measure_on_mutually_avoiding_SLEs}
Michael~J. Kozdron and Gregory~F. Lawler.
\newblock The configurational measure on mutually avoiding {$\mathrm{SLE}$}
  paths.
\newblock In {\em Universality and renormalization}, volume~50 of {\em Fields
  Inst. Commun.}, pages 199--224. Amer. Math. Soc., Providence, RI, 2007.

\bibitem[KP26]{Karrila-Peltola:Boundary_double-dimer_patterns_and_CFT}
Alex Karrila and Eveliina Peltola.
\newblock Boundary double-dimer patterns and conformal field theory.
\newblock In preparation, 2026.

\bibitem[KW11]{Kenyon-Wilson:Double_dimer_pairings_and_skew_Young_diagrams}
Richard~W. Kenyon and David~B. Wilson.
\newblock Double-dimer pairings and skew {Y}oung diagrams.
\newblock {\em Electron. J. Combin.}, 18(1):1--22, 2011.


\bibitem[Kyt06]{Kytola:On_CFT_of_SLE_kappa_rho}
Kalle Kyt{\"o}l{\"a}.
\newblock On conformal field theory of $\mathrm{SLE}(\kappa,\rho)$.
\newblock {\em J. Stat. Phys.}, 123(6):1169--1181, 2006.


\bibitem[Law09a]{Lawler:Partition_functions_loop_measure_and_versions_of_SLE}
Gregory~F. Lawler.
\newblock Partition functions, loop measure, and versions of {$\mathrm{SLE}$}.
\newblock {\em J. Stat. Phys.}, 134(5-6):813--837, 2009.


\bibitem[Law09b]{Lawler:SLE}
Gregory~F. Lawler.
\newblock {S}chramm-{L}oewner evolution.
\newblock In {\em Statistical Mechanics}, IAS/Park City mathematical series,
  pages 231--295. Amer. Math. Soc., 2009.

\bibitem[LPR25]{LPR:Fused_Specht_polynomials_and_c_equals_1_degenerate_conformal_blocks}
Augustin Lafay, Eveliina Peltola, and Julien Roussillon.
\newblock Fused {S}pecht polynomials and $c=1$ degenerate conformal blocks.
\newblock {\em Trans. Amer. Math. Soc.}, 12:1043--1100, 2025.


\bibitem[LSW04]{LSW:Conformal_invariance_of_planar_LERW_and_UST}
Gregory~F. Lawler, Oded Schramm, and Wendelin Werner.
\newblock Conformal invariance of planar loop-erased random walks and uniform
  spanning trees.
\newblock {\em Ann. Probab.}, 32(1B):939--995, 2004.


\bibitem[Lup16]{Lupu:From_loop_clusters_and_random_interlacements_to_the_free_field}
Titus Lupu.
\newblock From loop clusters and random interlacements to the free field.
\newblock {\em Ann. Probab.}, 44(3):2117--2146, 2016.


\bibitem[LW21]{Liu-Wu:Scaling_limits_of_crossing_probabilities_in_metric_graph_GFF}
Mingchang Liu and Hao Wu.
\newblock Scaling limits of crossing probabilities in metric graph {GFF}.
\newblock {\em Electron. J. Probab.}, 26(37):1--46, 2021.


\bibitem[MS16a]{Miller-Sheffield:Imaginary_geometry1}
Jason Miller and Scott Sheffield.
\newblock Imaginary geometry {I}: interacting $\mathrm{SLE}$s.
\newblock {\em Probab. Theory Related Fields}, 164(3-4):553--705, 2016.


\bibitem[MS16b]{Miller-Sheffield:Imaginary_geometry2}
Jason Miller and Scott Sheffield.
\newblock Imaginary geometry {II}: reversibility of
  $\mathrm{SLE}_\kappa(\rho_1,\rho_2)$ for $\kappa \in (0,4)$.
\newblock {\em Ann. Probab.}, 44(3):1647--1722, 2016.


\bibitem[MS16c]{Miller-Sheffield:Imaginary_geometry3}
Jason Miller and Scott Sheffield.
\newblock Imaginary geometry {III}: reversibility of $\mathrm{SLE}_\kappa$ for
  $\kappa \in (4,8)$.
\newblock {\em Ann. of Math.}, 184(2):455--486, 2016.


\bibitem[MSW20]{MSW:Non-simple_SLE_curves_are_not_determined_by_their_range}
Jason Miller, Scott Shaffield, and Wendelin Werner.
\newblock Non-simple {SLE} curves are not determined by their range.
\newblock {\em J. Eur. Math. Soc.}, 22(3):669--716, 2020.


\bibitem[MW17]{Miller-Wu:Intersections_of_SLE_paths:_the_double_and_cut_point_dimension_of_SLE}
Jason Miller and Hao Wu.
\newblock Intersections of {SLE} paths: the double and cut point dimension of
  {SLE}.
\newblock {\em Probab. Theory Related Fields}, 167(1-2):45--105, 2017.


\bibitem[OPS88]{OPS:Extremals_of_determinants_of_Laplacians}
Brad Osgood, Ralph Phillips, and Peter Sarnak.
\newblock Extremals of determinants of {L}aplacians.
\newblock {\em J. Funct. Anal.}, 80(1):148--211, 1988.


\bibitem[PW19]{Peltola-Wu:Global_and_local_multiple_SLEs_and_connection_probabilities_for_level_lines_of_GFF}
Eveliina Peltola and Hao Wu.
\newblock Global and local multiple $\mathrm{SLE}$s for $\kappa \leq 4$ and
  connection probabilities for level lines of {GFF}.
\newblock {\em Comm. Math. Phys.}, 366(2):469--536, 2019.


\bibitem[PW21]{Powell-Werner:Lecture_notes_on_the_GFF}
Ellen Powell and Wendelin Werner.
\newblock Lecture notes on the {G}aussian free field.
\newblock In {\em Cours Sp\'ecialis\'es}, number~28. Soci\'et\'e Math\'ematique
  de France, 2021.

\bibitem[RS05]{Rohde-Schramm:Basic_properties_of_SLE}
Steffen Rohde and Oded Schramm.
\newblock Basic properties of $\mathrm{SLE}$.
\newblock {\em Ann. of Math.}, 161(2):883--924, 2005.


\bibitem[Sch00]{Schramm:Scaling_limits_of_LERW_and_UST}
Oded Schramm.
\newblock Scaling limits of loop-erased random walks and uniform spanning
  trees.
\newblock {\em Israel J. Math.}, 118(1):221--288, 2000.


\bibitem[Sch06]{Schramm:ICM}
Oded Schramm.
\newblock Conformally invariant scaling limits, an overview and a collection of
  problems.
\newblock In {\em Proceedings of the ICM 2006, Madrid, Spain}, volume~1, pages
  513--543. European Mathematical Society, 2006.

\bibitem[She07]{Sheffield:GFF_for_mathematicians}
Scott Sheffield.
\newblock Gaussian free field for mathematicians.
\newblock {\em Probab. Th. Rel. Fields}, 139(3):521--541, 2007.


\bibitem[SS09]{Schramm-Sheffield:Contour_lines_of_2D_discrete_GFF}
Oded Schramm and Scott Sheffield.
\newblock Contour lines of the two-dimensional discrete {G}aussian free field.
\newblock {\em Acta Math.}, 202(1):21--137, 2009.


\bibitem[SS13]{Schramm-Sheffield:A_contour_line_of_the_continuum_GFF}
Oded Schramm and Scott Sheffield.
\newblock A contour line of the continuum {G}aussian free field.
\newblock {\em Probab. Theory Related Fields}, 157(1):47--80, 2013.


\bibitem[SW05]{Schramm-Wilson:SLE_coordinate_changes}
Oded Schramm and David~B. Wilson.
\newblock $\mathrm{SLE}$ coordinate changes.
\newblock {\em New York J. Math.}, 11:659--669, 2005.


\bibitem[SY24]{Sun-Yu:SLE_partition_functions_via_conformal_welding_of_random_surfaces}
Xin Sun and Pu~Yu.
\newblock {SLE} partition functions via conformal welding of random surfaces.
\newblock {\em Int. Math. Res. Not.}2024(24):14763--14801, 2024.


\bibitem[WW17]{Wang-Wu:Level_lines_of_Gaussian_free_field_I}
Menglu Wang and Hao Wu.
\newblock Level lines of {G}aussian free field {I}: {Z}ero-boundary {GFF}.
\newblock {\em Stochastic Process. Appl.}, 127(4):1045--1124, 2017.


\bibitem[Zha08]{Zhan:Reversibility_of_chordal_SLE}
Dapeng Zhan.
\newblock Reversibility of chordal {$\mathrm{SLE}$}.
\newblock {\em Ann. Probab.}, 36(4):1472--1494, 2008.


\bibitem[Zha24]{Zhan:Existence_and_uniqueness_of_nonsimple_multiple_SLE}
Dapeng Zhan.
\newblock Existence and uniqueness of nonsimple multiple {SLE}.
\newblock {\em J. Stat. Phys.}, 191(8):1--15, 2024.


\end{thebibliography}

\end{document}